\documentclass[10pt,journal]{IEEEtran}

\usepackage{amsmath}
\usepackage[ruled,vlined]{algorithm2e} 
\usepackage{subfigure} 
\usepackage{enumitem} 
\usepackage{pgfplots} 
\pgfplotsset{compat=1.6}

\usepackage{bookmark}
\usepackage{algpseudocode}
\usepackage{amssymb}
\usepackage{amsthm}  
\usepackage{tikz}
\usepackage{gnuplot-lua-tikz}
\usepackage{multirow}
\usepackage{rotating}
\usepackage{pgfplots}
\pgfplotsset{compat=1.6}
\usepackage{pgfplotstable}
\usetikzlibrary{arrows,automata}
\usetikzlibrary{decorations.text}
\usepackage{verbatim}
\usepackage{url}
\usepackage{color}
\usepackage{array}
\usepackage{tabularx}
\usepackage{epstopdf}

\newtheorem{theorem}{Theorem}
\newtheorem{remark}[theorem]{Remark} 
\newtheorem{problem}{Problem} 

\newtheorem{definition}{Definition}
\newtheorem{lemma}{Lemma}
\newtheorem{example}{Example}
\newtheorem{corollary}[theorem]{Corollary}

\newtheorem{conjecture}{Conjecture}

\newcommand{\ceil}[1]{\lceil #1 \rceil}
\newcommand{\floor}[1]{\lfloor #1 \rfloor}

\newcommand{\errLsum}{\ensuremath{L_1}}
\newcommand{\errLinf}{\ensuremath{L_\infty}}
\newcommand{\errLinfPos}{\ensuremath{L_\infty^+}}
\newcommand{\errLinfPosRel}{\ensuremath{L_{\infty,r}^+}}
\newcommand{\errLinfRel}{\ensuremath{L_{\infty,r}}}

\newcommand{\bigOnotation}[1]{\ensuremath{O\big (#1 \big )}}
\newcommand{\justTc}{\min(W \cdot k, k + n \cdot \lg k)}
\newcommand{\runtimeTc}{\bigOnotation{\min(W \cdot k, k + n \cdot \lg k)}}
\newcommand{\runtimeNiagaraWithN}{\bigOnotation{k + n \cdot \lg k}}
\newcommand{\runtimeLiftingTwosidedWithT}{\bigOnotation{W \cdot T}}
\newcommand{\runtimeLiftingOnesidedWithT}{\bigOnotation{W(k \cdot \lg W + T)}}
\newcommand{\runtimeLiftingTwosided}{\bigOnotation{W \cdot \justTc{}}}
\newcommand{\runtimeLiftingOnesided}{\bigOnotation{W(k \cdot \lg W + \justTc{})}}
\newcommand{\runtimeReductionTaTc}{\bigOnotation{W \cdot (T_A + T_C)}}

\definecolor{myblue}{RGB}{80,80,160}
\definecolor{mygreen}{RGB}{80,160,80}

\newcommand{\TRANS}[3] {\ensuremath{[{#1} \to_{#2} {#3}]}} 

\begin{document}

\title{Optimal Weighted Load Balancing in TCAMs}
\author{Yaniv Sadeh, Ori Rottenstreich and Haim Kaplan}
\maketitle

{\let\thefootnote\relax\footnotetext{
This manuscript is an extended version of the paper \cite{sadehconext}, presented at ACM CoNEXT 2020.
Yaniv Sadeh, Tel-Aviv University, Israel (yanivsadeh@mail.tau.ac.il). 
Ori Rottenstreich, Technion - Israel Institute of Technology,  Israel  (or@technion.ac.il). 
Haim Kaplan, Google and Tel-Aviv University, Israel (haimk@tau.ac.il). }}

\setlength{\textfloatsep}{3pt} 

\begin{abstract}
\looseness=-1
Traffic splitting is a required functionality in networks, for example for load balancing over multiple paths or among different servers. The capacities of the servers determine the partition by which traffic should be split. A recent approach implements traffic splitting within the ternary content addressable memory (TCAM), which is often available in switches. It is important to reduce the amount of memory allocated for this task since TCAMs are power consuming and are often also required for other tasks such as classification and routing. Previous work showed how to compute the smallest prefix-matching TCAM necessary to implement a given partition exactly. In this paper we solve the more practical case, where at most $n$ prefix-matching TCAM rules are available, restricting the ability to implement exactly the desired partition. We give simple and efficient algorithms to find $n$ rules that generate a partition closest in $L_\infty$ to the desired one. We do the same for a one-sided version of $L_\infty$ which equals to the maximum overload on a server and for a relative version of it. We use our algorithms to evaluate how the expected error changes as a function of the number of rules, the number of servers, and the width of the TCAM.
\end{abstract}

\section{Introduction}
\label{section_indtroduction}
In many networking applications, traffic has to be split into multiple possible targets. For example, this is required in order to partition traffic among multiple paths to a destination based on link capacities (e.g.~\cite{split_sigcomm05, split_conext07, split_sigcomm14}), and when sending traffic to one of multiple servers proportionally to their CPU or memory resources.

It is increasingly common to rely on network switches  to perform the split~\cite{al2008scalable,Ananta}. Equal cost multipath routing (ECMP)~\cite{RFC2992} and its generalization
WCMP (Weighted ECMP) \cite{WCMP} use hashing for this task. The possible target values are written to memory entries (with repetitions in WCMP).  Then, a flow is randomly hashed into one of the entries, generating a distribution according to the number of appearances of each possible target.

The implementation of some distributions in WCMP can be costly in terms of the number of  memory entries required. While for instance implementing a 1:2 ratio can be done with three entries (one for the first target and two for the second), the implementation of a  ratio like $1:2^W-1$ is expensive, requiring $2^W$ entries. Memory can grow quickly for particular distributions over many targets, even if they are only being approximated.

More recently, a natural approach was taken to implement traffic distributions within the Ternary Content Addressable Memory (TCAM), available in commodity switch architectures. For some distributions this allows a much cheaper representation~\cite{WangBR11, Niagara, AccurateExp}. In particular, a partition of the form $1:2^W-1$ can be implemented with only two entries. A nice feature of TCAM is that all the rules are checked in parallel, and multiple-matches are resolved to the highest priority rule (no ties), all done directly by the hardware. Unfortunately, TCAMs are power consuming and thus are of limited size~\cite{appelman2012performance, McKeownABPPRST08}. Therefore one often needs to represent a partition using a predefined TCAM quota.

Finding a representation of a partition becomes more difficult when the number of possible targets is large.
Focusing on the \emph{Longest Prefix Match model}, \cite{Niagara} suggested an algorithm named \emph{Niagara}, showed that it is very efficient in practice, and considered a tradeoff of reduced accuracy for less rules. \cite{BitMatcher} suggested an optimal algorithm named \emph{Bit Matcher} that computes a minimal size TCAM for a desired partition, and proved that Niagara is also optimal.

The work of \cite{BitMatcher} does not address the common scenario in which the available number of TCAM entries is  smaller than the minimum needed to represent the desired partition exactly and therefore an approximate solution is necessary. However \cite{BitMatcher} showed experimentally that a truncation of an implementation of the exact partition gives a good approximation. Specifically, they showed that the subsets of the `less specific' rules produced by Bit Matcher and Niagara, provide a good approximate partition according to several metrics. Unfortunately, they did not prove any worst case approximation guarantee for this approach.

\looseness=-1
In this paper we focus on the problem of finding the best approximate partition that fits  a fixed ``budget'' of TCAM entries.
This problem arises since the same TCAM is often  used for multiple tasks, one of which is traffic splitting. Thus, the number of rules allocated for traffic slitting may be limited in order to prevent starvation of other tasks that require rules, or simply because we allocate to load balancing the space remaining after higher-priority tasks had been allocated their rules.

When we must compromise the accuracy of a desired partition of $2^W$ $P = (p_1,\ldots,p_k)$ due to bounded memory size, then we need to define an appropriate notion of approximation $P' = (p'_1,\ldots,p'_k)$. Different applications may prefer different measures. Two natural measures are the $\errLinf$ and $\errLsum$ norms between partitions viewed as vectors. That is $|P-P'|_\infty = \max_i{|p_i-p'_i|}$ and $|P-P'|_1 = \sum_i{|p_i-p'_i|}$. Another interesting measure is a ``one-sided'' variant of $\errLinf$, denoted $\errLinfPos$, which equals to the maximum among the positive differences between the entries ($\max_i{(p'_i-p_i)}$).
This measure is not symmetric and equals to the
maximum of the overload on a server ignoring underloaded servers.\footnote{Note that $\errLsum$ does not have a one-sided version since the total overload equals to the total underload.} 
Another measure is minimizing the maximum relative overload, $\max_i \frac{p'_i-p_i}{p_i}$, denoted $\errLinfPosRel$. For example, if $\errLinfPosRel = 0.1$ it means that no server gets traffic larger by more than 10\% of its desired load.

In this work we focus on the $\errLinf$ distance (Sections~\ref{section_characterizing_lifting_cases}-\ref{section_solving_lifting}) and the $\errLinfPos$ and $\errLinfPosRel$ distances (Section~\ref{section_one_sided_metric}). We present efficient algorithms that find the closest partition (with respect to each of these distances) to a given one among all partitions that can be realized with at most $n$ rules.

We want to use the $\errLinfPosRel$ distance when stronger servers can tolerate larger overloads. That is, when an overload of 5\% is equally painful whether it occurs on a weak server or on a strong server. For example, when 
deviations result in delay, faster servers can tolerate larger deviations. In contrast when the deviating traffic is dropped then $\errLinfPos$ is more appropriate. When we also want to avoid servers which are severely underloaded then we should use $\errLinf$. Solving the problem for the relative version $\errLinfRel$ (minimize $\max_i \frac{|p_i'-p_i|}{p_i}$) is  open. When the desired loads on the servers are similar then optimizing with respect to $\errLinf$ and $\errLinfRel$ should give similar TCAMs.

\looseness=-1
A naive approach to get $n$ rules that may induce a close partition 
is by taking the widest $n$ rules (i.e.\ with maximum number of wildcards) of a smallest set of rules that induce the desired partition exactly. We can do this efficiently using a solution computed by either Bit Matcher or Niagara. 
Unfortunately this simple approach may not give the closest partition (according to the distances mentioned above) that we can induce by $n$ rules. This is demonstrated in Fig.~\ref{figure_tcam_trie_all} in the next section.

{\bf Our Contributions.} 
(1) We give new polynomial time algorithms that find partitions that optimally approximate
a desired partition and can be represented with a given memory constraint. To the best of our knowledge, this problem has never been studied, except for simple heuristics in \cite{Niagara,BitMatcher}. We do so for three ``approximation-measures'' between partitions: $\errLinf$, $\errLinfPos$ and $\errLinfPosRel$.
Our algorithm for $\errLinf$ runs in $\runtimeLiftingTwosidedWithT{}$ time and for $\errLinfPos$ and $\errLinfPosRel$ in $\runtimeLiftingOnesidedWithT{}$ time, where $k$ is the number of targets, $W$ is the  TCAM-width, $n$ is the number of allowed rules and $T = \runtimeTc{}$.

We observe that the problem of finding the closest partition that can be represented by at most $n$ rules is polynomially equivalent to the problem in which we want to find a partition that can be implemented by the smallest number of rules among those that have error below a fixed (specified) threshold.

\looseness=-1
To find a partition with bounded error of lowest complexity (fewest rules) we consider a more general formalization of \emph{lifting problems} (see Problem~\ref{problem_lift_x_to_y}). In a lifting problem we specify a lower bound on each coordinate in the partition and a range of legal values. We then ask for a partition of lowest complexity that obeys these constraints.
We reduce the problem of finding a partition with bounded error of lowest complexity  
to particularly simple lifting problems in which the lengths of the ranges have at most $3$ different values for $\errLinf$, and all lower bounds are $0$ for $\errLinfPos$ and $\errLinfPosRel$.
Then we show how to solve these special lifting problems. The lifting formalization can also capture similar approximation-measures.

(2) We exploit our algorithms to perform the following experimental studies: (A) We analyze the distribution of the error when approximating a partition of $2^W$ into $k$ parts, using at most $n$ rules.
We estimate how this error depends on the different parameters by computing the empirical average error and variance over a batch of random partitions with the same parameters.
We change each of the parameters while fixing the others, to 
single out the effect of each of them. We also consider the interesting case where the number of rules is proportional to the number of servers, i.e.\ $n/k=constant$.
(B) We compare the average $\errLinfPos$ and $\errLinf$ errors and conclude that they do not differ by much in practical scenarios. 
(C) We measure the error as a function of the number of rules for partitions derived from real data.
(D) We compare the error of the naive approach that truncates a smallest set of TCAM rules. It turns out that although sub-optimal, the error of this simpler approach is not much higher, and we conjecture that it is a $2$-approximation of the optimal error. 
(E) We measure the running time of our algorithms and compare it to the running time of the simpler truncation scheme. Our results suggest a trade-off between extra approximation error and the slightly faster running time of the Niagara algorithm (compared to the optimal algorithms that we develop here) which is \runtimeNiagaraWithN{}.

Note that all of these evaluations were not possible prior to our work as there was no efficient way to compute the closest partition which can be represented with at most $n$ rules.

\looseness=-1
The structure of the rest of the paper is as follows.
Section~\ref{section_model} formally defines our approximate traffic splitting problems, and the lifting problem that we mentioned before. Section~\ref{section_characterizing_lifting_cases} gives a reduction of the problem of finding the closest partition in $\errLinf$ distance of complexity at most $n$  to particular restricted  lifting problems. Section~\ref{section_solving_lifting} gives algorithms for the special lifting problems that arise from the reduction in  Section~\ref{section_characterizing_lifting_cases}. Section~\ref{section_one_sided_metric} reduces the problems of finding the closest partition in $\errLinfPos$ and $\errLinfPosRel$ distances, of complexity at most $n$, to yet another special lifting problem, and gives an efficient algorithm for this lifting problem.
Section~\ref{section_non_integer_partitions_approximation} extends all the results to non-integer input partitions.
Our experiments are described in Section~\ref{section_experiments}, due to lack of space some are provided as supplementary material. We review additional related work and we conclude with some open problems in Sections~\ref{section_related_work}~and~\ref{section_conclusions}, respectively. Finally, we provide three supplementary sections and an  implementation of our algorithms.
Section~\ref{section_review_bitmatcher_and_niagara} briefly reviews the existing Bit Matcher and Niagara algorithms, Section~\ref{section_code_info} describes the code, and Section~\ref{section_supplementary_experiments} contains more experiments.

\section{Traffic Splitting Problem}
\label{section_model}

A Ternary Content Addressable Memory (TCAM) of width $W$ is a table of entries, or \emph{rules}, each containing a \emph{pattern} and a \emph{target}. (We assume each target is an integer in $[1,\ldots,k]$.) Each pattern is of length $W$ and consists of bits (0 or 1) and wildcards ($*$).
An \emph{address} is said to match a pattern if all of the specified bits of the pattern (ignoring wildcards) agree with the corresponding bits of the address. If several rules fit an address, the first rule applies, and recall that this resolution is done by the hardware itself (in parallel). An  address $v$ is associated with the target of the rule that applies to $v$. We consider only minimal sets of rules, i.e.\ such that we cannot remove a rule without changing the mapping defined by the set. In the following we refer to a set of TCAM rules simply as a TCAM.

\looseness=-1
In this paper we assume the Longest Prefix Match (LPM) model in which wildcards appear only as a consecutive suffix of the pattern. This  model is motivated by specialized TCAM hardware such as \cite{LPM_TCAM}, and has been studied intensively \cite{WangBR11, Niagara, AccurateExp, BitMatcher}. As detailed in~\cite{NSDIJose15}, common programmable switch
architectures such as the RMT and Intel’s FlexPipe have tables of different types and in particular tables dedicated to longest prefix matching~\cite{Forwarding13, FlexPipe}.
In this setting a pattern is in fact a \emph{prefix} of bits that matches all addresses that start with this prefix. The set of addresses $S(p_1)$ and $S(p_2)$ that match different prefixes $p_1$ and $p_2$ (assume without loss of generality that $p_1$ is not longer than $p_2$) form a laminar set system. That is $S(p_1) \supseteq S(p_2)$ if $p_1$ is a prefix of $p_2$ and 
$S(p_1) \cap S(p_2) = \emptyset$ otherwise. Furthermore, in case $p_1$ is a prefix of $p_2$  then we may assume that $p_2$ appears before $p_1$ (otherwise $p_2$ is redundant). It follows that we may assume that the prefix rules are sorted in a non-increasing order of their lengths.
Finally, we assume that the set of rules ends with a rule consisting of a prefix of length 0 that matches all addresses.

A set of prefix rules $T$ corresponds to a subset of the nodes of the full binary trie (see Fig.~\ref{figure_tcam_trie_all}).\footnote{Our arguments and results cannot be generalized to the case of non-prefix rules since the correspondence with the full binary trie breaks, similar to previous works~\cite{WangBR11, Niagara, AccurateExp, BitMatcher}.}
In particular, the match-all prefix corresponds to the root, and any other nonempty prefix $p$ corresponds to the node whose path to the root gives $p$ if we change an edge to a left child to $0$ and an edge to a right child to $1$. The rule which applies to an address $v$ (which is a leaf) corresponds to the closest ancestor which represents a prefix in $T$. This is the longest prefix of the address $v$ in $T$.

\begin{figure}[t!]
    \centering
    \subfigure[Optimal exact representation. Corresponds to the sequence {$[4 \to_0 3][5 \to_0 2][3 \to_1 2][2 \to_2 1][1 \to_3 \bot]$}.]{
    \includegraphics[width=0.75\linewidth]{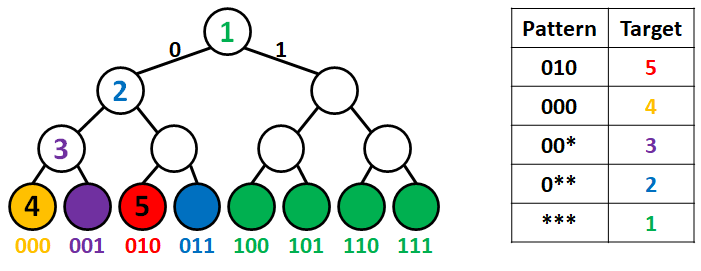}
    \label{figure_tcam_trie_exact}
    }
    
    \subfigure[Approximation via truncation (error of $3$). Corresponds to the sequence {$[2 \to_2 1][1 \to_3 \bot]$}.]{
    \includegraphics[width=0.75\linewidth]{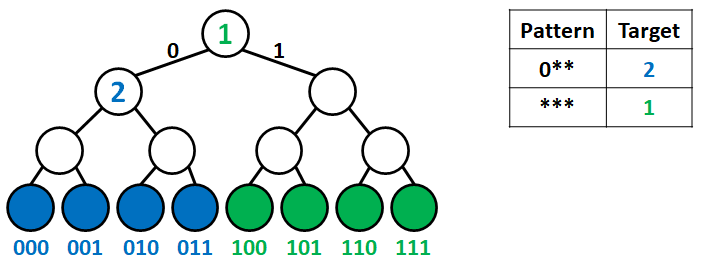}
    \label{figure_tcam_trie_truncated}
    }
    
    \subfigure[Optimal $\errLinf$,$\errLinfPos$ approximation: error of $2$. Corresponds to the sequence {$[2 \to_1 1][1 \to_3 \bot]$}.]{
    \includegraphics[width=0.75\linewidth]{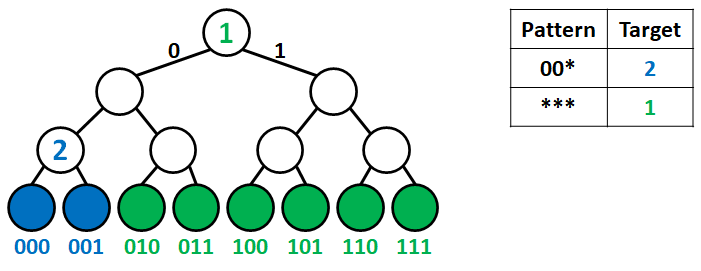}
    \label{figure_tcam_trie_opt}
    }

    \subfigure[Optimal $\errLinfPosRel$ approximation: error of $\frac{3}{4}$. Corresponds to the sequence {$[2 \to_1 0][1 \to_3 \bot]$}.]{
    \includegraphics[width=0.75\linewidth]{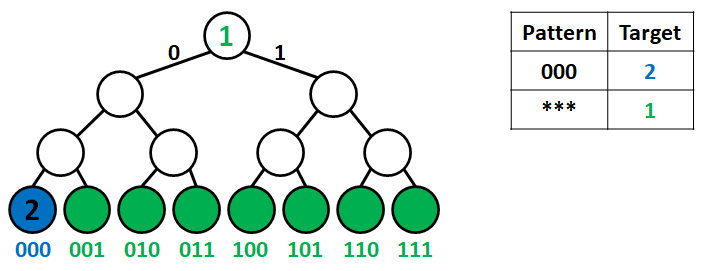}
    \label{figure_tcam_trie_optRel}
    }

  \caption{Illustration of TCAM rules, their trie representation and sequences generating the corresponding partitions. \ref{figure_tcam_trie_exact} A shortest exact representation of $P=[4,1,1,1,1]$ with $5$ rules. \ref{figure_tcam_trie_truncated}, \ref{figure_tcam_trie_opt}, \ref{figure_tcam_trie_optRel} give different approximations of $P$ with $2$ rules. \ref{figure_tcam_trie_truncated} An approximation of $P$ with the last two rules of the optimal exact representation. This induces the partition $[4,4,0,0,0]$. \ref{figure_tcam_trie_opt} A representation with two rules of $[6,2,0,0,0]$ which is closest to $P$ in $\errLinf$ and $\errLinfPos$ among all partitions that can be represented with $2$ rules. \ref{figure_tcam_trie_optRel} A representation with two rules of $[7,1,0,0,0]$ which is closest to $P$ in $\errLinfPosRel$ among all partitions that can be represented with $2$ rules. The figure presents toy-examples with very few rules, which is why the approximation ends up with less than $k$ reachable targets. When there are on average $2.5$ or more rules per target this is unlikely to happen, see Section~\ref{subsection_experiments_degeneracy} for more details.}
  \label{figure_tcam_trie_all}
\end{figure}

A TCAM of width $W$ with $k$ targets induces a partition of the address space of $2^W$ binary strings  into $k$ parts. Each address is associated with the target of the rule that applies to it. In this paper, we consider only partitions to $k$ parts which are non-negative integers that sum to $2^W$. The following natural problem has been solved in \cite{BitMatcher}.

\begin{problem}
\label{problem_tcam_exact}
Given a partition $P$ of $2^W$ into $k$ parts, find the shortest TCAM that realizes $P$, i.e.\ a TCAM that  partitions the addresses space to exactly these parts.
\end{problem}

We assume implicitly that every address is equally likely to arrive, therefore the TCAM implementation only requires each target to receive a certain number of addresses. This assumption might not hold in practice, but it can be mitigated by ignoring bits which are mostly fixed like subnet masks etc. For example, \cite{Kang2014NiagaraSL} analyzed some real-data traces and concluded that for those traces about $6{-}8$ bits out of the client's IPv4 address are ``practically uniform''. The non-uniform case in which addresses have different weights is not considered in this paper.\footnote{The exact non-uniform problem is NP-complete, since if addresses are weighted arbitrarily we can reduce the Partition problem~\cite{garey1979computers} to Problem~\ref{problem_tcam_exact}. Given an input $S=\{s_1,\ldots, s_m\}$ for Partition, consider a space of $2^W \ge m$ addresses. Give $m$ addresses weights $s_1, \ldots, s_m$, and weight $0$ to all others. Let $P = [p,p]$ for $p \equiv \frac{1}{2}\sum_{i=1}^{m}{s_i}$ be input for Problem~\ref{problem_tcam_exact}. $S$ can be partitioned into two subsets of equal weight if and only if $P$ has an exact TCAM representation.}
Moreover, independently of the uniformity assumption, the network designer may choose to reduce the granularity of the partition in order to smooth random fluctuations (even if address uniformity holds in general). For instance, if $W=32$ is over-granular, one can choose $W=10$ to begin with.

Throughout the paper we  denote vectors, such as partitions, with a capital letter such as $P$ or $X$. Each coordinate of such vector is denoted by a corresponding non-capital letter with an index as a subscript, such as $p_i$ and $x_i$, respectively.

Observe that a set of prefix rules $T$ that induces a partition $P$ into parts of sizes $p_1,\ldots, p_k$ defines (non-uniquely) a sequence $s$ of \textit{transactions} between pairs $(i,j)\in [1,\ldots,k]\times [1,\ldots,k]$, $i\not= j$ as defined below, such that after ``executing'' $s$ on $P$, all the $p_i$ are $0$ except for one $p_j$ that equals $2^W$.

\begin{definition}[Transactions]
\label{def_transaction}
We denote a transaction of size $2^\ell$ from $p_i$ to $p_j$ by $[i \to_\ell j]$. We also refer to $i$ (the index) as the ``sender'' and to $j$ as the ``receiver''. Applying a transaction means that we update the values as follows: $p_i \leftarrow p_i - 2^\ell$, $p_j \leftarrow p_j + 2^\ell$. A sequence $s$ is a collection of transactions. Fig.~\ref{figure_tcam_trie_all} provides four examples for describing sequences of transactions, each with its corresponding trie and TCAM table.
\end{definition}

\looseness=-1
Indeed, think of the representation of $T$ as a binary trie. Consider a prefix $p$ of length $W-\ell$ with target $A$ such that no descendant of $p$ is in $T$. Let $p'$ be the closest ancestor of $p$ which is also in $T$. Let $A'$ be the target of $p'$.  By our minimality assumption $A' \not= A$, so we add to $s$ a transaction moving $2^\ell$ from $A$ to $A'$ and remove $p$ from $T$. Then we iterate this step until only the match-all rule is left in $T$.
For example if we start from the original trie in Fig.~\ref{figure_tcam_trie_exact} then after adding to $s$ three transactions corresponding to the three longest prefix-rules, i.e.\ the first rules in this TCAM, the trie is as shown in Fig.~\ref{figure_tcam_trie_truncated}.

For uniformity, we add a dummy-transaction $[j \to_W \bot]$ where $\bot$ is a canonical symbol that has initial value $0$, and $j$ is the index that is the target of the match-all rule.
This transaction corresponds to deleting the match-all rule in $T$ in the process above, such that  all weights end up  zero. For a given partition $P$, we consider only sequences of transactions  that zero all the  weights (into $\bot$).

Given a sequence of transactions $s$ obtained from a TCAM $T$ as above we can reconstruct $T$ as follows: We start by a rule with a prefix of length 0 (match-all) whose target is $j$ where $[j \to_W \bot] \in s$. Then we traverse the transactions of $s$ in non-increasing sizes and act as follows. When we encounter a transaction moving $2^\ell$ from $i$ to $j$, we add a prefix rule of length $W-\ell$ with target $i$ which extends a previous prefix rule of target $j$ (i.e.\ the specified prefix of the rule of $j$, is a prefix of the new rule of $i$). In the trie representation, the prefix rule of $i$ which we add has a subtree of size $2^\ell$ and the prefix of rule $j$ which we extend has a larger subtree size.

\begin{definition}[Induced Partition]
\label{definition_induced_partition}
Let $s$ be any sequence of transactions. We say that $P$ is the \textbf{partition induced by $s$} if when we apply $s$ to $P$ all the weights become zero. Note that a sequence $s$ induces a unique partition, however different sequences may induce the same partition. When $s$ is ``read backwards'', it constructs $P$.
\end{definition}

Note that an arbitrary sequence $s$ that induces a partition $P$, may not correspond to
a TCAM $T$ that realizes $T$ by the mapping described above. Indeed, there are infinitely many sequences, but only a finite number of marked tries. Nevertheless \cite{BitMatcher} proved that any shortest sequence that induces $P$ can be converted to a sequence of the same length that corresponds to a set of prefix rules of a TCAM that realizes $P$.\footnote{In \cite{BitMatcher} the sequences are defined without the $\bot$-transaction, causing an off-by-one mismatch between TCAM size and sequence length.} 
Moreover, the previous paragraphs demonstrate that any TCAM with $n$ rules that realizes a partition $P$ also yields a sequence of length $n$ that induces $P$. Therefore, if a shortest sequence $s$ that induces $P$ is of length $n$ then this is also the size of the smallest TCAM that realizes $P$, and essentially Problem \ref{problem_sequence_exact} below is equivalent to 
Problem~\ref{problem_tcam_exact}.

\begin{problem}
\label{problem_sequence_exact}
Given a partition $P$, compute a shortest sequence that induces it.
\end{problem}

\begin{definition}[Complexity of a Partition]
\label{definition_length_of_partition}
Let $P$ be a partition. We define $\lambda(P)$ as the length of a shortest sequence that induces $P$. We  say that $\lambda(P)$ is \textbf{the complexity of $P$}. This value describes how large a TCAM realization of $P$ must be.
\end{definition}

Problems~\ref{problem_tcam_exact} and \ref{problem_sequence_exact} were studied in \cite{BitMatcher}, which showed that a shortest sequence that induces a partition $P$ can be found in $O(Wk)$ time, by their Bit Matcher algorithm or in $O(Wk \lg k)$ time by the Niagara algorithm of \cite{Niagara}. Therefore, we can compute $\lambda(P)$ for any desired partition $P$, efficiently. However, $\lambda(P)$ might be too large for a particular application. To address this issue,
we study here the problem of finding a best approximate partition of a given complexity.
 
\begin{definition}
\label{definition_partition_balls}
We define  $\Lambda_n$ to be the set of all partitions of complexity at most $n$. That is, $\Lambda_n \equiv \{ P: \lambda(P) \le n \}$. Note that $\Lambda_{n_1} \subseteq \Lambda_{n_2}$ whenever $n_1 \le n_2$.
\end{definition}

\begin{definition}
\label{definition_distances}
Let $P,P'$ be two partitions. Let $\Delta = P'-P$ be the difference vector. We denote by $D(P',P)$ the distance of $P'$ from $P$. We consider four different distances in this paper which are defined as follows:
\begin{enumerate}[leftmargin=*, noitemsep]
    \item Max two-sided: $\errLinf(P',P) = \max_i{|\Delta_i|}$
    \item Max two-sided relative: $\errLinfRel(P',P) = \max_i{\frac{|\Delta_i|}{p_i}}$
    \item Max one-sided (positive): $\errLinfPos(P',P) = \max_i{\Delta_i}$
    \item Max one-sided (positive) relative: $\errLinfPosRel(P',P) = \max_i{\frac{\Delta_i}{p_i}}$
\end{enumerate}
\end{definition}

In simple words, $\errLinf$ penalizes for the maximum deviation, while $\errLinfRel$ normalizes this relative to the capacity of each target. $\errLinfPos$ penalizes only for overloaded targets, and $\errLinfPosRel$ measures the overload relative to the capacity of each target. It is possible to define similar one-sided distances for underloads, but the practical interest in it is questionable. One-sided measures may be of interest if the network designer worries about overload but can accept underutilized resources.

\begin{problem}[Bounded-Length Approximation]
\label{problem_tcam_approx_by_rules}
Given a partition $P$, a distance $D$, and an integer $n$, find a partition $P' \in \Lambda_n$ such that $D(P',P)$ is minimized.
\end{problem}

In TCAM terminology, Problem~\ref{problem_tcam_approx_by_rules} looks for a TCAM with at most $n$ rules that induces a partition that best approximates  a desired partition $P$. We solve Problem~\ref{problem_tcam_approx_by_rules} by reducing it to the following ``dual'' problem.

\begin{definition}
\label{definition_partition_ball}
Let $P$ be a partition, $D$ a distance, and let $e > 0$. We define the open ball around $P$ with radius $e$ as the set of all partitions with distance less than $e$ from $P$. That is, $B_e(P) \equiv \{ P': D(P',P) < e\}$.
\end{definition}

\begin{problem}[Bounded-Error Approximation]
\label{problem_tcam_approx_by_error}
Let $P$ be a partition, $D$ a distance, and let $e>0$ be an error bound. Find a partition $P' \in B_e(P)$ that minimizes $\lambda(P')$.
\end{problem}

\begin{theorem}
\label{theorem_equivalent_rule_bound_error_bound1}
Let $A$ be an algorithm for solving Problem~\ref{problem_tcam_approx_by_error} for $\errLinf$, $\errLinfPos$ or $\errLinfPosRel$ that requires $T_A$ time, and let $C$ be an algorithm that checks for an integer partition $P$ and an integer $n$ whether $\lambda(P) \le n$ in $T_C$ time. Then Problem~\ref{problem_tcam_approx_by_rules} is solvable in $\runtimeReductionTaTc{}$ time. 
\end{theorem}

\begin{proof}
First consider $\errLinf$ or $\errLinfPos$. If $e_2 \ge e_1$ then $B_{e_1}(P) \subseteq B_{e_2}(P)$ and therefore  $\lambda(P'_{e_2})\le \lambda(P'_{e_1})$. By this monotonicity we can use binary search to find the smallest ball around $P$ that contains a partition $P'$ of complexity at most $n$. Since $\errLinf(P',P)$ and $\errLinfPos(P',P)$ are integers no larger than $2^W$, the number of iterations of this binary search is at most $W$.

In each iteration of this binary search we call $A$ once, and then post-process its output to check whether $\lambda(P'_e) \le n$. This totals to $\runtimeReductionTaTc{}$ time. Note that by \cite{BitMatcher} we know that $T_C$ is polynomial.

For $\errLinfPosRel$, since different distance values are separated by at least $\epsilon = 2^{-2W}$, we can apply the search down to a resolution of $\epsilon$ instead of $1$, which still takes $\runtimeReductionTaTc{}$ time.
\end{proof}

\begin{theorem}
\label{theorem_equivalent_rule_bound_error_bound2}
Let $A$ be an algorithm for solving Problem~\ref{problem_tcam_approx_by_rules} for $\errLinf$, $\errLinfPos$ or $\errLinfPosRel$ that requires $T_A$ time. Then Problem~\ref{problem_tcam_approx_by_error} is solvable in $O((\log W + \log k) \cdot (T_A + k))$ time.
\end{theorem}

The proof of Theorem~\ref{theorem_equivalent_rule_bound_error_bound2} is similar to that of Theorem~\ref{theorem_equivalent_rule_bound_error_bound1} and hence omitted. We remark that for any integer partition $P$, $0 \le \lambda(P) \le Wk$. Problem~\ref{problem_lift_x_to_y} below generalizes Problem~\ref{problem_tcam_approx_by_error}.

\begin{definition}[Lifting Vector]
Let $W$ be an integer and let $X \equiv [x_1,\ldots,x_k]$ be a vector of integer \textbf{weights} such that $\sum_i x_i \le 2^{W}$.
Let $C \equiv [c_1,\ldots,c_k]$ be a vector of integer \textbf{capacities} such that $\sum_i (x_i+c_i) \ge 2^{W}$.
We say that a vector $Y$ is a \textbf{lifting} of $X$ with respect to $C$ and $W$ if
(1) $\forall i: x_i \le y_i \le x_i + c_i$, and
(2) $\sum_i y_i = 2^W$.
\end{definition}

\begin{problem}[Lifting Problem]
\label{problem_lift_x_to_y}
Given an input $(X,C,W)$, find a lifting vector $Y$ with minimum $\lambda(Y)$.
\end{problem}

Problem~\ref{problem_lift_x_to_y} looks for a shortest partition in a $k$-dimensional box: $X$ specifies the coordinate-wise smallest corner of the box, and $C$ specifies the width of the box in each dimension,\footnote{Of course only partitions are considered, so the points of interest in the box lie in its intersection with the hyperplane $\sum_i p_i = 2^W$.} see Fig.~\ref{figure_lifting_visualization} for a 2-dimensional illustration. This generalizes Problem~\ref{problem_tcam_approx_by_error} for each of the distances $\errLinf$, $\errLinfRel$, $\errLinfPos$ and $\errLinfPosRel$. In the case of $\errLinf$ this box is in fact a cube, defined by its center $P$ and radius (half-width) $e$.\footnote{Since we do not allow negative values in partitions, there could be  cases where the box $B_e(P)$ in Problem~\ref{problem_tcam_approx_by_error} is not a cube. This happens 
when $p_i < e$ for some $i$, so the cube is trimmed through intersecting with the positive orthant. Consider $P=[1,5]$ with $e=4$ as an example. The cube around $P$ is $[-2,4]\times [2,8]$ but its intersection with the positive orthant is $[0,4]\times [2,8]$.} In general, Problem~\ref{problem_lift_x_to_y} captures any definition of distance in which the error defines a box around $P$. This box is not required to be centered at $P$.

\begin{figure}[!t]
	\centering
	
	\includegraphics[width=1.0\linewidth]{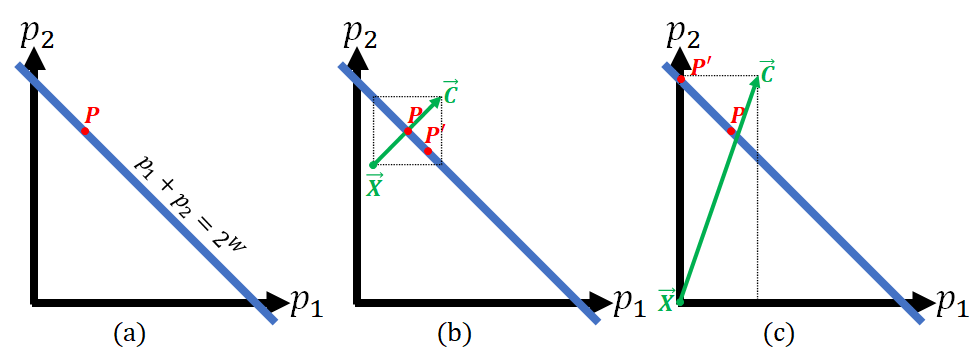}
	
	\looseness=-1
    \caption{A visualization of a lifting problem in $2$ dimensions. (a) Every partition resides on the line (hyper-plane) $p_1 + p_2 = 2^W$. (b) For $\errLinf$ the box is a square ($2$D cube) around $P$. (c) For $\errLinfPos$ and $\errLinfPosRel$ the rectangle ($2$D box) around $P$ is typically not centered at $P$ but instead anchored at the origin. We look for a lifting vector $\vec{Y}$ by starting from the corner with smallest coordinates, $\vec{X}$, and slowly increasing the coordinates inside the box without passing the corner $\vec{X}+\vec{C}$, until ``hitting'' the hyper-plane of partitions at some $P'$ that minimizes $\lambda(P')$ inside the box.}
    \label{figure_lifting_visualization}
\end{figure}

Our focus is on Problem~\ref{problem_tcam_approx_by_error}.
In Section~\ref{section_characterizing_lifting_cases} we identify  
 two special cases of Problem~\ref{problem_lift_x_to_y} that will allow us to solve Problem~\ref{problem_tcam_approx_by_error}  for $\errLinf$. We solve each of these special cases in Section~\ref{section_solving_lifting} and explain how together they suffice for solving Problem~\ref{problem_tcam_approx_by_error}. In Section~\ref{section_one_sided_metric} we identify another special case of Problem~\ref{problem_lift_x_to_y} that we use to solve Problem~\ref{problem_tcam_approx_by_error} for $\errLinfPos$ and $\errLinfPosRel$ and explain how to solve it. While it is simple to reduce Problem~\ref{problem_tcam_approx_by_error} for $\errLinfRel$ to an instance of Problem~\ref{problem_lift_x_to_y}, it does not result in a special case which we know how to solve, so the problem remains open for $\errLinfRel$.

At this point, we provide Fig.~\ref{figure_roadmap} as a summary of the reductions and special cases. The details will be clarified in the subsequent sections.

\begin{figure}[t!]
    \subfigure[Solving for closest-partition via reduction to lifting problems. Adding a scalar to a vector means adding it to each coordinate.]{
    \includegraphics[width=0.95\linewidth]{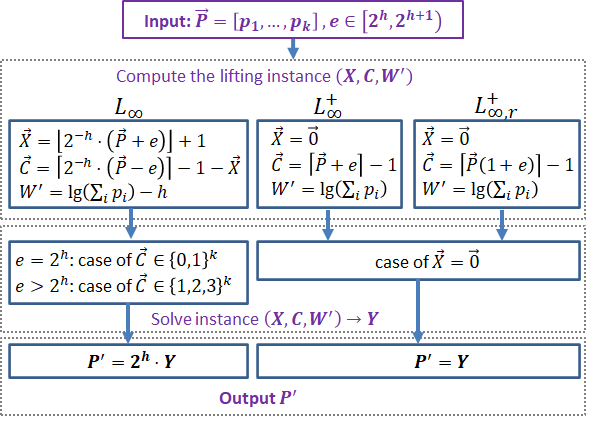}
    \label{figure_roadmap_reduction}
    }
    
    \subfigure[Solving special cases of lifting problems (Problem~\ref{problem_lift_x_to_y}).]{
    \includegraphics[width=0.95\linewidth]{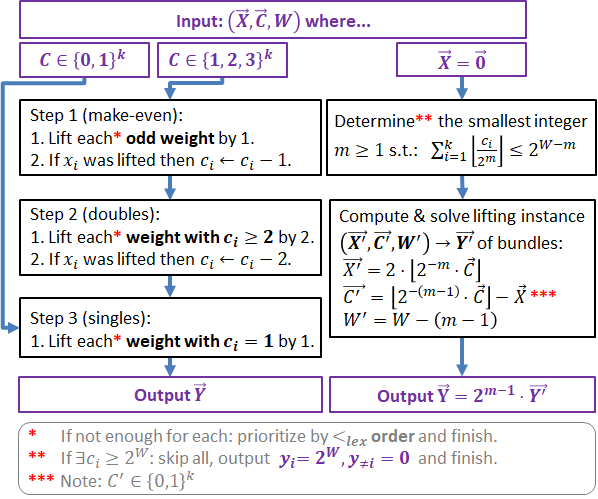}
    \label{figure_roadmap_solvers}
    }

  \caption{High-level overview of the algorithms for computing a closest partition for the different distance functions, by translating the problem to special lifting problems. \ref{figure_roadmap_reduction} is an overview of the reduction (e.g.\ Section~\ref{section_characterizing_lifting_cases} for $\errLinf$), \ref{figure_roadmap_solvers} overviews the solution steps (e.g.\ Section~\ref{section_solving_lifting} for $\errLinf$).}
  \label{figure_roadmap}
\end{figure}

\section{Characterizing Special Lifting Problem Cases}
\label{section_characterizing_lifting_cases}

In this section we characterize special cases of the lifting problem (Problem~\ref{problem_lift_x_to_y}) that are sufficient for solving Problem~\ref{problem_tcam_approx_by_error} for $\errLinf$. The first Lemma states that in order to find a partition $P' \in B_e(P)$ that minimizes $\lambda(P')$ it is sufficient to consider only the subset of $B_e(P)$ consisting of all partitions whose entries are multiples of the largest power of $2$ smaller or equal to the approximation error  $e$.

\begin{lemma}
\label{lemma_fine_rules_unnecessary}
Let $P$ be a desired partition and let $2^h \le e < 2^{h+1}$ for a non-negative integer $h$. Let some $P' \in B_e(P)$. Then there exists $P^* \in  B_e(P)$ such that $\lambda(P^*) \le \lambda(P')$ and $2^h$ divides every coordinate of $P^*$.
\end{lemma}

\begin{proof}
If every coordinate of $P'$ is divisible by $2^h$ then $P^*=P'$ and we are done. Otherwise, let $s'$ be a shortest sequence that induces $P'$. Because $P'$ has elements that are not divisible by $2^h$, some of the transactions of $s'$ must be of size smaller than $2^h$. We show how to modify $s'$
to a sequence $s''$ that has one less 
 transaction of size smaller than $2^h$, $|s''|\le |s'|$,
 and $s''$ induces a partition 
 $P'' \in  B_e(P)$. The lemma then follows by applying this argument iteratively until we get a partition that is induced by a sequence without small transactions.
 
We eliminate a single small transaction as follows. Denote the first transaction in $s'$ which is of size smaller than $2^h$ by $[a \to_m b]$ (for some indices $a,b$ and $m<h$). We delete the transaction $[a \to_m b]$ from $s'$ and  denote the new temporary sequence by $s$. The sequence $s$ induces a partition $P^s$, in which $p^s_a = p'_a+2^m$ and $p^s_b = p'_b-2^m$. The rest of the proof splits into the following four possible cases.

\begin{enumerate}[leftmargin=*]
    \item Both $|p^s_a-p_a| < e$ and $|p^s_b-p_b| < e$: We set
    $s''=s$.
    \label{case_1}
    
    \item Both $|p^s_a-p_a| \ge e$ and $|p^s_b-p_b| \ge e$: because $|p'_a-p_a| < e$, and because $p'_a < p^s_a$ and $m<h$, we get that $p_a \le p'_a < p^s_a$ (if $p'_a < p_a$ then $|p^s_a-p_a| < e$). Thus $p^s_a - 2^h > p_a - 2^h \ge p_a - e$, 
    so by reducing
    $2^h$ from $p^s_a$
    we would bring it back to within distance $e$ from $p_a$.
    For similar reasons, we find that $p_b \ge p'_b > p^s_b$ and by adding $2^h$ to $p^s_b$ we would bring it back to within distance $e$ from $p_b$. We set $s'' = s\cup \{ [a \to_h b]\}$.
    \label{case_2}
    
    \item $|p^s_a-p_a| \ge e$ and $|p^s_b-p_b| < e$ (i.e.\ need to fix $a$): The same reasoning as in case~\ref{case_2} above implies that $p^s_a > p_a$, and that we can fix $p^s_a$ by adding a transaction $[a \to_h c]$  for some index $c$. It remains to show that there exists some weight $p^s_c$ that can accommodate such an increase without violating the constraint, i.e.\ that $|p^s_c + 2^h -p_c| < e$.
   
   Notice that for any $c \ne a,b$ we have $p^s_c = p'_c$, and therefore $|p^s_c - p_c| = |p'_c - p_c| < e$ because $|P'-P|_\infty < e$. Furthermore, in this case $|p^s_b - p_b| < e$ by assumption, hence for any $c \ne a$ (including $b$) we have $|p^s_c - p_c| < e$.
  
   Assume by contradiction that $\forall c \ne a$ we also have $|p^s_c + 2^h - p_c| \ge e$. Then it means that $p^s_c \ge p_c$ (if $p^s_c < p_c$ then $|p^s_c + 2^h - p_c| < e$). But then, because $p^s_a > p_a$, we get that $\sum p^s_i > \sum p_i = 2^W$, which is a contradiction because $P^s$ is also a partition and its sum is $2^W$. Therefore such index $c$ does exist, and we set $s''$ to  be $s \cup \{ [a \to_h c] \}$.
    \label{case_3}
    
    \item $|p^s_a-p_a| < e$ and $|p^s_b-p_b| \ge e$ (i.e.\ need to fix $b$): This case is symmetric to case~\ref{case_3}. The same reasoning as in case~\ref{case_2} implies that $p^s_b < p_b$. We find some index $c$ such that we can add the transaction \TRANS{c}{h}{b} while still maintaining $|p^s_c - 2^h -p_c| < e$. Notice that for any $c \ne b$ we have $|p^s_c - p_c| < e$ (for $c \ne a$ this is because $p^s_c = p'_c$ and for $c=a$ we assume that $|p^s_a - p_a| < e$ in this case). Assuming, by contradiction, that $\forall c \ne b: |p^s_c - 2^h - p_c| \ge e$, we deduce that $p^s_c \le p_c$. From this we get that $\sum p^s_i < \sum p_i = 2^W$, which is a contradiction, and we conclude that such index $c$ as required does exist. Accordingly, we set $s''$ to  be $s \cup \{ \TRANS{c}{h}{b} \}$.
    \label{case_4} \qedhere
\end{enumerate}
\end{proof}

\begin{example}  
Let $W=4$, $k=3$, $P = [11,4,1]$, $P' = [6,9,1]$, and $e=6$. Indeed, $P' \in B_e(P)$ since $|P'-P|_\infty = 5 < e$. One can verify that  $\lambda(P') = 4$ since the sequence $s' = [3 \to_0 2][2 \to_1 1][1 \to_3 2][2 \to_4 \bot]$ is a shortest sequence that induces $P'$. By Lemma~\ref{lemma_fine_rules_unnecessary} we can find 
$P^*\in B_e(P)$ such that $\lambda(P^*)\le 4$ and each coordinate of $P^*$ is divisible by $4$. We can find $P^* = [8,4,4]$, which is induced by $s^* = [3 \to_2 2][1 \to_3 2][2 \to_4 \bot]$. Each weight of $P^*$ is divisible by $4$, $\lambda(P^*) = 3 \le \lambda(P') = 4$, and $|P^*-P|_\infty = 3 < e$.
\end{example}

Based on Lemma \ref{lemma_fine_rules_unnecessary}, when we look for the shortest partition  $P' \in B_e(P)$, it is sufficient to consider only the lattice points consisting of coordinates which are integer multiples of $2^h$ (where we recall that $h$ is such that  $2^h \le e < 2^{h+1}$). This  reduces the search space substantially.

Lemma~\ref{lemma_spread_values} below characterizes more precisely the number of multiples of $2^h$ in the vicinity of any given value. Fig.~\ref{figure_number_line_examples} illustrates the number of non-negative multiples, in each of the three cases.

\begin{lemma}
\label{lemma_spread_values}
Let $x \ge 0$, $2^h\le e< 2^{h+1}$ for a non-negative integer $h$. Let $A(x) = \{ x' \mid |x' -x| < e \wedge x' \in \mathbb{Z} \wedge (2^h\ divides\ x')\}$ and define the functions $\Delta_e(x) = |A(x)|$ and $\Delta^{+}_e(x) = |A(x) \cap \{ x' \mid x'\ge 0\}|$. Then the possible values of
$\Delta_e(x)$ and $\Delta^{+}_e(x)$ depend on $e$ as follows:
\begin{enumerate}[label=(\arabic*),leftmargin=*,noitemsep]
    \item $e=2^h$: both $\Delta_e(x),\Delta^+_e(x) \in \{1,2\}$
    \label{lemma_regret_size_case1}
    
    \item $2^h < e  \le \frac{3}{2} \cdot 2^h$: both $\Delta_e(x),\Delta^+_e(x) \in \{2,3\}$
    \label{lemma_regret_size_case2}
    
    \item $\frac{3}{2} \cdot 2^h < e< 2^{h+1}$: $\Delta_e(x) \in \{3,4\}, \Delta^{+}_e(x) \in \{2,3,4\}$
    \label{lemma_regret_size_case3}
\end{enumerate}
\end{lemma}

\begin{figure}[ht!]
    \centering
    \subfigure[The three cases of Lemma~\ref{lemma_spread_values}.]{
    \includegraphics[width=0.99\linewidth]{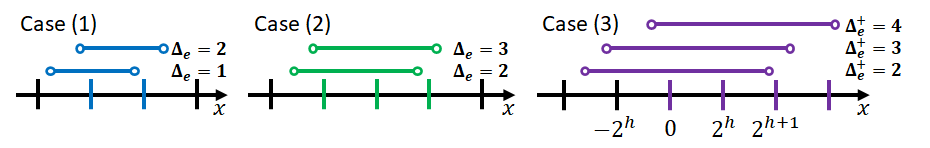}
    \label{figure_number_line_cases}
    }
    
    \subfigure[Breaking a segment into units of $2^h$.]{
    \includegraphics[width=0.75\linewidth]{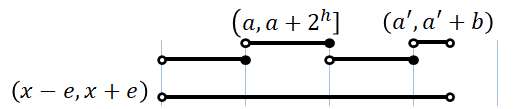}
    \label{figure_number_line_breakdown}
    }

  \caption{Lemma~\ref{lemma_spread_values}: Counting non-negative multiples of $2^h$ covered by a range. \ref{figure_number_line_cases} shows the three possible cases of Lemma~\ref{lemma_spread_values}, \ref{figure_number_line_breakdown} shows a break-down of a segment, for counting the multiples it covers.}
  \label{figure_number_line_examples}
\end{figure}

\begin{proof}
\looseness=-1
The length of the segment $|x-x'| < e$ is exactly $2e$, therefore its length in units of $2^h$ is $r = \frac{2e}{2^h}$. Each full-unit $(a,a+2^h]$ contains exactly one integer multiple of $2^h$, and a partial-unit $(a,a+b)$ where $0 < b \le 2^h$ may or may not contain a multiple (see Fig.~\ref{figure_number_line_breakdown}). Since we can break the segment $(x-e,x+e)$ to $\ceil{r} - 1$ consecutive full-units and another partial-unit, we get that the segment covers either $\ceil{r} - 1$ or $\ceil{r}$ multiples of $2^h$.
We divide to cases depending on $e$:
\begin{enumerate}[label=(\arabic*),leftmargin=*,noitemsep]
    \item $e=2^h$: Then $r = \frac{2e}{2^h} = 2$, and therefore $\ceil{r} = 2$.
    
    \item $2^h < e \le \frac{3}{2} \cdot 2^h$: Then $2 < r = \frac{2e}{2^h} \le 3$, and $\ceil{r} = 3$.
    
    \item $\frac{3}{2} \cdot 2^h < e < 2^{h+1}$: Then $3 < r = \frac{2e}{2^h} < 4$, and $\ceil{r} = 4$.
\end{enumerate}

When adding the requirement that $x' \ge 0$,
negative multiples of $2^h$ in the corresponding interval  become illegal. Since $x \ge 0$ and $e < 2^{h+1}$, we have that $-2^{h+1} \notin A(x)$ and if $-2^h \in A(x)$ then $0,2^h \in A(x)$.
This means that even if we lose a multiple of $2^h$ due to non-negativity constraint, then there are at least two non-negative multiples of $2^h$ within the range. In Case~\ref{lemma_regret_size_case1}
the range does not contain $-2^h$ so we do not lose any multiple. In Case~\ref{lemma_regret_size_case2} we may lose one multiple if we originally had $3$. In Case~\ref{lemma_regret_size_case3} we may lose one if we originally had $3$ or $4$. In the former case, $\Delta_e(x)$ goes down from $3$ to $2$ as stated (see Fig.~\ref{figure_number_line_cases}).
\end{proof}

Lemma~\ref{lemma_fine_rules_unnecessary} shows that
it suffices to work with partitions consisting of elements of sizes which are integer
multiples of $2^h$, and Lemma~\ref{lemma_spread_values} characterizes how many such multiples we have to consider per coordinate.
Therefore, we can now simplify the presentation by 
dividing all values by $2^h$.
 
The benefit of this simplification is that it allows us to reduce our problem to
a simple lifting problem with
small capacities. Notice that although the capacities are small (at most $3$), we still cannot 
solve the resulting lifting problem efficiently by traversing all possible partitions in $B_e(P)$ as their number is still exponential in $k$.
The following theorem gives the details of this reduction. In the next section we  show how to efficiently solve the lifting problem.

\begin{theorem}[Reduction from Bounded-Error to Lifting]
\label{theorem_special_case_lifting_reduction}
Let $(P,e)$ be the input to Problem~\ref{problem_tcam_approx_by_error} for $\errLinf$ where $P$ is a partition and $2^h \le e < 2^{h+1}$ is the error bound.
Then the problem reduces to one of the following two types of lifting problems, depending on the value of $e$:
\begin{enumerate}
    \item $e = 2^h \Rightarrow$ Capacities $C$ such that  $\forall i: c_i \in \{0,1\}$.
    \item $e > 2^h \Rightarrow$ Capacities $C$ such that  $\forall i: c_i \in \{1,2,3\}$.
\end{enumerate}
\end{theorem}

\begin{proof}
By Lemma~\ref{lemma_fine_rules_unnecessary} it suffices to look for a partition $P'$ such that $\forall i \in [k]$, $p'_i$ is a multiple of $2^h$ and $|p_i - p'_i| < e$. Let $x_i$ be the lowest non-negative possible multiple, and let $z_i$ be the largest possible multiple. Any multiple in between is also possible, so we define the lifting problem with the base vector $X = [x_1/2^h, \ldots, x_k/2^h]$ and capacities $C = [c_1, \ldots, c_k]$ where $c_i = (z_i - x_i)/2^h$. The third parameter to the lifting problem is $W = W_P - h$ where $W_P$ is the original ``width'' of $P$, that is $\sum_{i\in[k]} p_i = 2^{W_P}$.

Lemma~\ref{lemma_spread_values} characterizes $\Delta_e^+(p_i)$, the number of multiples of $2^h$ between $x_i$ and $z_i$ (including both), and it follows  that $c_i + 1 = \Delta_e^+(p_i)$.
We get  two types of lifting problems. If  $e=2^h$ then by Case~\ref{lemma_regret_size_case1} of Lemma~\ref{lemma_spread_values} we get $c_i \in \{0,1\}$.
If  $e>2^h$ then by Cases \ref{lemma_regret_size_case2} and  \ref{lemma_regret_size_case3}
we get $c_i \in \{1,2,3\}$.

As a final remark, note that in the lifting problem we get a width of $W_P-h$ rather than $W_P$ (the width of $P$). In other words, since we divided all values  by $2^h$ the sum of the parts  of the lifting should be $2^{W_P-h}$. The solution of the lifting problem is scaled-up by multiplying each coordinate by $2^h$ to get a solution to the original problem. The scaling doesn't change the complexity of the partition.
\end{proof}

\section{Solving the  Lifting Problems}
\label{section_solving_lifting}

In this section we consider the special cases of Problem \ref{problem_lift_x_to_y} derived in Theorem \ref{theorem_special_case_lifting_reduction}. The algorithm for the second case uses the algorithm for the first case as a sub-routine.

\subsection{Lifting Problems with \texorpdfstring{$C \in \{0,1\}^k$}{Ci in (0,1)}}
\label{subsection_lifting_01}
In this section we give an algorithm for lifting instances $(X,C,W)$ such that $\forall i\in [k]: c_i \le 1$. We do so by ranking the weights with capacity $1$ in bit-lexicographic order (defined below), and lift the first $2^W-\sum_i x_i$ largest weights in this order. See Algorithm~\ref{alg_lift_max_capacity_1}.

\begin{definition}[Bit-Lexicographic Order]
\label{definition_lexicorder}
Let $x,y$ be non-negative integers. We say that $\boldsymbol{x <_{lex} y}$ or that $x$ is \textbf{bit-lexicographic smaller than} $y$ if at the lowest bit $\ell$ such that $x[\ell]\ne y[\ell]$ we have $x[\ell]=0<y[\ell]=1$. 
Note that $x <_{lex} y$ if and only if
$x^r < y^r$ where $x^r$ and $y^r$ are obtained from $x$ and $y$ by reversing their binary representations (with respect to a fixed word size), respectively.
\end{definition}
\begin{example}
\label{example_lexicorder}
$5 <_{lex} 3$ ($5=101b$, $3 = 011b$). Also, any even number is lexicographically smaller than any odd number.
\end{example}

\begin{algorithm}[tb!]
    \DontPrintSemicolon
    \KwIn{
         Non-negative integer weights $X = [x_1,\ldots,x_k]$ and capacities $C \in \{0,1\}^k$ such that $\sum_i x_i \le 2^W \le \sum_i (x_i+c_i)$.
    }
    \KwOut{
        An optimal lifting $Y = [y_1,\ldots,y_k]$.
    }
    
    \begin{enumerate}
        \item Let $A = \{x_i \mid c_i = 1 \} \subseteq X$. Let $n = 2^W - \sum_i x_i$.
        
        \item Set $y_i = x_i + 1$ for the $n$ highest weights of $A$ by $<_{lex}$ order (Definition~\ref{definition_lexicorder}), otherwise $y_i=x_i$. Return $Y$.
    \end{enumerate}
    \caption{Lifting with max capacity of 1}
    \label{alg_lift_max_capacity_1}
\end{algorithm}

In order to prove the correctness of Algorithm~\ref{alg_lift_max_capacity_1}, we introduce additional notation.

\begin{definition}[Excess]
\label{definition_excess}
Let $(X,C,W)$ be the input to a lifting problem. We define the \textbf{excess} of $X$ as $e(X) = 2^W - \sum_{i}{x_i}$. Note that $e(X)$ depends only on $X$ and $W$, and is exactly the total amount required to be added to get a lifting of $X$.
\end{definition}

\begin{definition}[Lifting Notation]
\label{definition_lifting}
Let $Y$ be some valid lifting of a vector $X$. If $y_i > x_i$ we say that $x_i$, or the index $i$, has been lifted $y_i-x_i$ times, and if $y_i = x_i$ we say that it was not lifted. We denote a single lift of $x_i$ by  $[i+]$. 
\end{definition}

\begin{definition}[Lifting Sequence]
\label{definition_sequence_of_lifting}
Let $Y$ be a lifting of $X$. Then $\sum_i y_i = 2^W$, and therefore $\lambda(Y)$ is well defined. We define a \textbf{lifting sequence} of $Y$ to be a sequence 
consisting of the lifting from $X$ to $Y$ concatenated to a shortest sequence of transactions that induces $Y$. We denote by $s(Y)$ a lifting sequence of $Y$ in which the subsequence that induces $Y$ is generated by the Bit Matcher algorithm.
\end{definition}

\begin{example}
\label{example_lift_sequence}
Let $X=[1,2,3]$, $W=3$, $C=[1,2,3]$. The lifting $Y=[2,2,4]$ has $s(Y) = [1+][3+][1 \to_1 2][2 \to_2 3][3 \to_3 \bot]$. Its lifting part is $[1+][3+]$, and the transactions part is $[1 \to_1 2][2 \to_2 3][3 \to_3 \bot]$, which is  a Bit Matcher sequence that induces $Y$.
\end{example}

\begin{lemma}
\label{lemma_prioritize_when_capacity_1_E_2h}
Let $(X,C,W)$ be an input to a lifting problem such that $\sum_i x_i = 2^W - 1$, and all capacities are $0$ except for $c_a$ and $c_b$ that equal $1$. 
Denote by $Y^a$ the lifting in which $x_a$ is lifted, and similarly denote by $Y^b$ the lifting in which $x_b$ is lifted. If $x_a <_{lex} x_b$ then $\lambda(Y^b) \le \lambda(Y^a)$.
\end{lemma}

\begin{proof}
Consider the  partition $Z$ of $2^{W+1}$ defined as follows: $\forall i \ne a,b: z_i = 2x_i$, $z_a = 2x_a + 1$, and $z_b = 2x_b + 1$. This is indeed a partition of $2^{W+1}$  since $\sum_i z_i = 2\sum_i x_i + 2 = 2^{W+1}$.
We establish the lemma by proving that $\lambda (Y^b) = \lambda (Z) - 1 \le \lambda (Y^a)$. 

First we show that $ \lambda (Z)-1 \le \lambda (Y^a)$. 
Consider a sequence $s^z$ that induces $Z$, which we define as follows.
It starts with the transaction
$[b \to_0 a]$. After applying this transaction, we have $Z=2Y^a$. 
So we complete $s^z$ by taking a shortest sequence $s^a$ that induces $Y^a$ and adding a transaction $[i \to_{\ell+1} j]$ for every transaction 
$[i \to_\ell j]$ in $s^a$ (double size).
This gives that $\lambda (Z) \le \lambda (Y^a) + 1$ as required.

 Note that a similar argument shows that  $\lambda (Z)\le \lambda (Y^b) + 1 $.

To conclude the proof, we show that $\lambda (Y^b) + 1 \le \lambda (Z)$. Among all shortest sequences of $Z$, consider a Bit Matcher sequence $s$ for $Z$. This sequence begins by matching odd weights according to bit-lexicographic order. Since only $z_a$ and $z_b$ are odd and $z_a <_{lex} z_b$ (because $x_a <_{lex} x_b$), the first transaction in $s$ is $[a \to_0 b]$. After applying this transaction we get that $\forall i \ne b: z_i = 2x_i$ and $z_b = 2x_b + 2$ so at this point $Z = 2 Y^b$, and all the remaining transactions are of size at least $2$.
It follows that we can convert the suffix of $s$ following its first transaction into a sequence of the same length that induces $Y^b$, by replacing every transaction 
$[i \to_{\ell} j]$ by
$[i \to_{\ell-1} j]$.
 Hence we conclude that $\lambda (Y^b) \le \lambda (Z) - 1$.
\end{proof}

Note the importance of the fact that $x_a$ and $x_b$ in Lemma~\ref{lemma_prioritize_when_capacity_1_E_2h} have capacity of $1$.
Consider for example the lifting problem $(X,C,W)$, where
$X=[2,4,6,1]$ with capacities $C=[1,3,0,0]$ and $W=4$. Although $x_2 <_{lex} x_1$, the optimal lifting is $Y=[2,7,6,1]$
with $\lambda(Y)=4$ and  $s(Y)=[2+][2+][2+][4 \to_0 2][1 \to_1 3][3 \to_3 2][2 \to_4 \bot]$.
If we shift one lifting from
$x_2$ to $x_1$ we get the lifting
 $Y = [3,6,6,1]$ with $\lambda(Y)=5$. Note also that if we lift $x_2$ twice so its remaining capacity is $1$, then $x_1=2,x_2=6$ and the lemma holds for the remaining lift because $2 <_{lex} 6$. However, the lemma does not help us foresee this.

\begin{theorem}
\label{theorem_opt_lifting_max_capacity_1}
Algorithm~\ref{alg_lift_max_capacity_1} generates an optimal lifting for Problem~\ref{problem_lift_x_to_y} for any input $(X,C,W)$ where  $C \in \{0,1\}^k$.
\end{theorem}

\begin{proof}
For a lifting $Y$ of $(X,C,W)$ we denote by
$A(Y)$ the set of lifted indices. That is, the indices $i$ such that $y_i=x_i+1$.

Let $Y^1$ denote the lifting produced by Algorithm~\ref{alg_lift_max_capacity_1}. Assume by contradiction that $Y^1$ is not optimal and let $Y^*$ be an optimal lifting of $(X,C,W)$ that maximizes $|A(Y^*)\cap A(Y^1)|$.

Since $Y^* \ne Y^1$, $A(Y^*) \ne A(Y^1)$.
Furthermore, since all liftings of
$(X,C,W)$ contain the same number of lifts, then $|A(Y^*)| = |A(Y^1)|$.
This means that 
$ A(Y^*) \setminus A(Y^1) \ne \emptyset$ and
$ A(Y^1) \setminus A(Y^*) \ne \emptyset$.
We pick $a \in A(Y^*) \setminus A(Y^1)$ and $ b \in A(Y^1) \setminus A(Y^*)$. By the definition of Algorithm~\ref{alg_lift_max_capacity_1}, we have that $x_a <_{lex} x_b$.

We now consider the lifting problem $(X',C',W)$ derived from $Y^*$ as follows. We set $x'_a = y^*_a - 1$ and $x'_i = y^*_i$ for every $i\not= a$. We also set $c'_a=1$, $c'_b =1$ and $c'_i=0$ for $i\not=a,b$.

By Lemma~\ref{lemma_prioritize_when_capacity_1_E_2h}, there is an optimal lifting $Y'$ of $(X',C',W)$ that lifts $x'_b$. It follows that $Y'$ is at least as good as $Y^*$ and hence also optimal for $(X,C,W)$. But $|A(Y')\cap A(Y^1)| > |A(Y^*)\cap A(Y^1)|$ which contradicts the choice of 
$Y^*$.
\end{proof}

\subsection{Lifting Problems with \texorpdfstring{$C \in \{1,2,3\}^k$}{Ci in (1,2,3)}}
\label{subsection_lifting_123}

In this section we present Algorithm~\ref{alg_lift_capacity_1_2_3} that computes
optimal liftings for instances $(X,C,W)$ such that $C \in \{1,2,3\}^k$. We break the excess assignment into phases, and maintain the guarantee that at least one optimal solution survives following each phase. More formally, a partial assignment of the excess induces a new problem $(X',C',W)$ where $x_i \le x'_i \le x_i + c_i$, $c'_i+x'_i=c_i+x_i$, and the guarantee is that an optimal solution $Y'$ for $(X',C',W)$ is an optimal solution for $(X,C,W)$. The algorithm consists of three phases as follows:
\begin{enumerate}
    \item The first phase makes all weights even by lifting by one each odd weight.
    \item The second phase makes pairs of lifts of weights with capacity at least two (keeping the weights even).
    \item The last phase allocates the remaining excess by single lifts (the remaining capacities are at most $1$).
\end{enumerate}

\noindent
When the remaining excess at a phase is too small we use Algorithm~\ref{alg_lift_max_capacity_1} on an appropriately defined lifting problem to prioritize the allocation of the remaining excess and finish. We establish the correctness of each phase of the algorithm separately (Lemma~\ref{lemma_excess_too_little}, Lemma~\ref{lemma_odd_or_all_even}, Theorem~\ref{theorem_opt_lifting_capacity_1_2_3}). 

\begin{algorithm}[tb!]
    \DontPrintSemicolon
    \KwIn{
         Non-negative integer weights $X = [x_1,\ldots,x_k]$ and capacities $C \in \{1,2,3\}^k$ such that $\sum_i x_i \le 2^W \le \sum_i (x_i+c_i)$
    }
    \KwOut{
        An optimal lifting  $Y = [y_1,\ldots,y_k]$.
    }
    
    \If {$e(X) \le |\{i \mid  x_i\ is\ odd\}|$ (see Definition~\ref{definition_excess})}{
    Return Algorithm\ref{alg_lift_max_capacity_1}$(X,[1,\ldots,1],W)$}
    \Else{
        Define $X',C'$ s.t.:
        \begin{itemize}
            \item if $x_i$ is even then $x'_i = x_i,c'_i=c_i$
            \item if $x_i$ is odd then $x'_i = x_i+1,c'_i=c_i-1$
        \end{itemize}
        
        \If{$e(X') < 2 \cdot |\{i \mid  c'_i \ge 2\}|$}{
            Return $2 \cdot$Algorithm\ref{alg_lift_max_capacity_1}$(X'/2,C'/2,W-1)$
        
            \looseness=-1 \emph{$/{*}$ Division and multiplication of vectors are coordinate-wise, over integers (e.g.\ $3/2 = 1$)\ ${*}/$ }
        }
        \Else{
        Define $X'',C''$ s.t.:
        \begin{itemize}
            \item if $c'_i \le 1$ then $x''_i = x'_i,c''_i=c'_i$
            \item if $c'_i \ge 2$ then $x''_i = x'_i+2,c''_i=c'_i-2$
        \end{itemize}
        Return
        Algorithm\ref{alg_lift_max_capacity_1}$(X'',C'',W)$
        }
    }
    \caption{Lifting with capacities in range \{1,2,3\}}
    \label{alg_lift_capacity_1_2_3}
\end{algorithm}

\begin{lemma}
\label{lemma_make_more_even_weights}
Let $Y$ be an optimal lifting of $X$, for capacities such that  $c_i > 0$ for all $i$.
If $\exists i\ne j$ such that $y_i$ and $y_j$ are odd, 
$y_i < x_i + c_i$, and $y_j > x_j$ then there exists another optimal lifting $Y'$ in which both $y'_i,y'_j$ are even.
\end{lemma}

\begin{proof}
\looseness=-1
Since $y_j > x_j$ it follows that $s(Y)$ contains at least one lift $[j+]$ (it could contain more such lifts). Moreover, because $y_i$ and $y_j$ are odd, $s(Y)$ contains transactions of size 1 in which $y_i$ and $y_j$ participate. Note that both $y_i$ and $y_j$ must participate in such a transaction, because at the end of $s(Y)$ they are both $0$. We split the rest of the proof into cases 
according to the roles of $i$ and $j$ in these transactions of size $1$. 
\begin{enumerate}
    \item If $j$ is a sender, $[j \to_0 a] \in s(Y)$ for some coordinate $a$:
    Then we replace $[j+][j \to_0 a]$ by $[i+][i \to_0 a]$.
    
    \item If $i$ is a receiver, $[a \to_0 i] \in s(Y)$ for some coordinate $a$:
    Then we replace $[j+][a \to_0 i]$ by $[i+][a \to_0 j]$.
    
    \item If $j$ is a receiver, and $i$ is a sender: Then
    without loss of generality we may assume $[i \to_0 j] \in s(Y)$, and
    we replace $[j+][i \to_0 j]$ by $[i+][i \to_1 j]$.
\end{enumerate}
In all the cases above, we obtain a lifting $Y'$
 from $Y$ by trading a single  lift
$[j+]$ for the lift $[i+]$. The argument shows
that $\lambda(Y')\le \lambda(Y)$.
In $Y'$
both $y_i'$ and $y_j'$ are even.
\end{proof}

\begin{corollary}
\label{corollary_minimal_odds}
Let $(X,C,W)$ be an input of a lifting problem. Then among the set of optimal lifting vectors, there exists $Y$ such that either every odd $y_i$ satisfies $y_i=x_i$ (none are lifted), or every odd $y_i$ satisfies $y_i > x_i$ (all are lifted).
\end{corollary}

Lemma \ref{lemma_excess_too_little} below proves the correctness of the first ``if'' in Algorithm \ref{alg_lift_capacity_1_2_3} and motivates the definition of $X'$ and $C'$ in the ``else'' part of this ``if''.

\begin{lemma}
\label{lemma_excess_too_little}
Let $(X,C,W)$ be an input to the lifting problem where $C \in \{1,2,3\}^k$. Denote $X^{odd} = \{i \mid x_i \ is \ odd\}$. If $e(X) \le |X^{odd}|$ then
an optimal lifting for $(X,[1,\ldots,1],W)$ 
is also optimal for $(X,C,W)$.
If $e(X) > |X^{odd}|$
then
let $(X',C',W)$ be defined such that 
 if $x_i$ is even
then $x'_i =x_i$ and  $c'_i=c_i$ and if $x_i$ is odd then  $x'_i =x_i+1$ and $c'_i=c_i-1$. An optimal solution to $(X',C',W)$
is also optimal for $(X,C,W)$.
\end{lemma}

\begin{proof}
Let $Y$ be an optimal lifting for the input $(X,C,W)$ and let $Y^{odd} = \{i \mid y_i \ is \ odd \}$. By Corollary \ref{corollary_minimal_odds}, we choose $Y$ such that either all members of $Y^{odd}$ are lifted (Lifted), or none are (Not). We consider the case of $Y^{odd} = \emptyset$ as one in which all members of $Y^{odd}$ are not lifted. We also split based on whether $e(X) \le |X^{odd}|$ (Less) or $e(X) > |X^{odd}|$ (More). The proof splits to 
cases: (1) Lifted-Less; (2) Lifted-More; (3) Not-Less; (4) Not-More.

Case 1 is impossible, because there is not enough excess to make all $X^{odd}$ even and then also lift some weights such that $Y^{odd}$ is not empty (recall that we consider the case $Y^{odd} = \emptyset$ as one in which all members of $Y^{odd}$ are not lifted).

Case 2 is easy. Indeed, if $Y^{odd}$ are all lifted, this means that any weight in $X^{odd}$ has been lifted (either becoming even or remaining odd, but lifted). 
It follows that $Y$ is a also an optimal lifting of $(X',C',W)$ as stated.

To deal with Cases 3 and 4 we describe a process that eliminates weights from $Y^{odd}$ if  there exists an index $j$ that has been lifted more than once, $Y^{odd} \ne \emptyset$, and all members of $Y^{odd}$ are not lifted.
Assume that indeed $j$ has been lifted more than once, $Y^{odd} \ne \emptyset$, and all members of $Y^{odd}$ are not lifted. Since $Y^{odd} \ne \emptyset$ and $\sum_i y_i = 2^W$
is even, we must have that $|Y^{odd}| \ge 2$. Since  $s(Y)$ consists of a sequence of Bit Matcher, there exist $a,b \in Y^{odd}$ such that $[a \to_0 b] \in s(Y)$ (Bit Matcher first matches odd weights together). 

We modify $Y$ to $Y'$ by replacing $[j+][j+][a \to_0 b]$ in $s(Y)$ by $[a+][b+][a \to_1 j]$ (we can lift 
$a$ and $b$ since all members of $Y^{odd}$ are not lifted and all capacities are at least $1$). The resulting lifting $Y'$ is also optimal since $\lambda(Y') \le \lambda(Y)$.

\looseness=-1
Observe that $a,b \notin ({Y'})^{odd}$ and that $j$ is not lifted more than once in $Y'$, since its capacity was initially $c_j \le 3$.
Moreover, $j \notin ({Y'})^{odd}$. This follows since $y_j$ was lifted in $Y$, and therefore $j \notin Y^{odd}$,  $y_j$ is even, and ${y'}_j = y_j - 2$ is also even.
 
We handle Cases 3 and 4 by repeatedly applying the process to the optimal lifting at hand, until it cannot be applied any more. Denote the final lifting that we get by $Y'$.

In Case 3 there is barely or not enough excess to make all the odd weights of $X$ even. Therefore in $Y'$ there is no index that is lifted more than once (if $e(X) = |X^{odd}|$ we will also get that $({Y'})^{odd} = \emptyset$). It follows that $Y'$ is also a  solution of the more restricted lifting problem $(X,\vec{1},W)$ ($\vec{1} \equiv [1,\ldots,1]$). Since $Y'$ is optimal for $(X,C,W)$  this implies that any optimal lifting of
$(X,\vec{1},W)$ is also an optimal lifting of  $(X,C,W)$.

In Case 4 there is enough excess to eliminate $Y^{odd}$, so  $({Y'})^{odd} = \emptyset$ and we get that $Y'$ is a solution of $ (X',C',W)$. Since $Y'$ is optimal for $(X,C,W)$ this implies that any optimal lifting of $(X',C',W)$ is also an optimal lifting of $(X,C,W)$.

To conclude this proof, observe that when $e(X) > |X^{odd}|$
(Cases 2 and 4) it is safe to apply a partial lifting of all $X^{odd}$ such that they become even, and that when $e(X) \le |X^{odd}|$ (Case 3) the problem $(X,C,W)$ reduces to $(X,\vec{1},W)$.
\end{proof}

Lemma \ref{lemma_excess_too_little} guarantees that given a lifting problem $(X,C,W)$,
if the  excess is large compared to the number of odd weights in $X$ then we can reduce it to the lifting problem $(X',C',W)$ as defined in Algorithm~\ref{alg_lift_capacity_1_2_3}. After this partial lifting, all the weights $x'_i$ are even, and  $\forall i: c'_i \in \{0,1,2,3\}$. The range of capacities now includes $0$ since some capacities of $1$ may have been fully used to make $x'_i$ even, while others may not have been used at all.
The following lemma implies the correctness of the second phase of Algorithm~\ref{alg_lift_capacity_1_2_3} (the second ``if''). 
 
\begin{lemma}
\label{lemma_odd_or_all_even}
Let $(X,C,W)$ be an input of a lifting problem such that $\forall i: x_i$ is even. Define $D = 2(C/2)$, that is, we obtain $D$ by subtracting $1$ from every odd capacity.
If $e(X) \le \sum_i d_i$ then there is an optimal lifting $Y^*$ such that $\forall i: y^*_i$ is even, and if $e(X) \ge \sum_i d_i$ then there is an optimal lifting $Y^*$ for $(X,C,W)$ such that $\forall i: y^*_i \ge x_i + d_i$.
\end{lemma}

\begin{proof}
In order to prove the cases, we describe a process which we will use to take an optimal lifting $Y$ and convert it to another optimal lifting $Y'$ with less odd weights. For this process, we require the existence of $j$ such that $y_j$ is odd, and also an index $i$ such that $y_i < x_i + d_i$.

Assuming $Y$ has odd weights, there must be at least two (because the sum over $Y$ is even). Denote by $y_a$ and $y_b$ two weights such that $[b \to_0 a] \in s(Y)$. Such a transaction exists because the transactions are generated by Bit Matcher which first makes transactions between pairs of odd weights. Moreover, if $y_i$ is odd, we pick a transaction in which either $i=a$ or $i=b$. Because $x_a,x_b$ are even and $y_a,y_b$ are odd, then $[a+],[b+] \in s(Y)$.

(1) If $i=a$ then we can replace $[a+][b+][b \to_0 a]$ in $s(Y)$ by $[a+][a+]$ which gives a lifting $Y'$ with a shorter lifting sequence in contradiction to the optimality of $Y$. So we must have that $i \ne a$.
    
(2) If $i=b$ then we can replace $[a+][b+][b \to_0 a]$ in $s(Y)$ by $[b+][b+][b \to_1 a]$, and get a new lifting $Y'$ with a lifting sequence of the same length as the lifting sequence of $Y$. In $Y'$, $y'_a$ and $y'_b$ are even so the number of odd weights in $Y'$ is two less than in $Y$.

(3) If $i \ne a,b$, then $y_i$ is even (if it was odd, then we must have chosen $i=a$ or $i=b$ by definition). Since $x_i$ and $d_i$ are also even, and $y_i < x_i + d_i$ we get $y_i \le x_i + d_i - 2$. Thus $i$ has capacity for two extra lifts. So we replace $[a+][b+][b \to_0 a]$ by $[i+][i+][i \to_1 a]$, and get a new lifting $Y'$ 
with a lifting sequence of the same length as the lifting sequence of $Y$.
In $Y'$, $y'_a$ and $y'_b$ are even so the number of odd weights in $Y'$ is two less than in $Y$.

Let $Y$ be an optimal lifting, and apply the process we described repeatedly, until we can no longer continue.

Consider the case of $e(X) \le \sum_i d_i$. In this case, as long as there are odd weights, there must be $i$ such that $y_i < x_i + d_i$ as required by the process. Indeed, either an odd weight satisfies that, or if every odd weight satisfies $y_a = x_a + d_a + 1$ then there is not enough excess left for the even weights, and there must be some other (even) weight that satisfies $y_i < x_i + d_i$. Therefore, we conclude that the process stops at an optimal lifting $Y^*$ such that $\forall i: y^*_i$ is even as the lemma claims.

Consider the case when $e(X) \ge \sum_i d_i$. In this case, the process stops either if $Y^*$ has no odd weights, or when there is no $i$ such that $y^*_i < x_i + d_i$. In the latter case, we get $\forall i: y^*_i \ge x_i + d_i$ as stated by the lemma. In the former case, since $x_i + d_i$ is the highest even value $y^*_i$ can reach, we get: $e(X) = \sum_i{(y^*_i-x_i)} \le \sum_i {d_i} \le e(X)$, and we conclude that $\sum_i d_i = e(X)$ and that $\forall i: y^*_i = x_i + d_i$.
\end{proof}

\begin{theorem}
\label{theorem_opt_lifting_capacity_1_2_3}
Algorithm \ref{alg_lift_capacity_1_2_3} generates an optimal lifting for Problem \ref{problem_lift_x_to_y} if $C \in \{1,2,3\}^k$.
\end{theorem}

\begin{proof}
We use the terminology from Algorithm \ref{alg_lift_capacity_1_2_3} throughout the proof. If $e(X) \le |\{i \mid x_i\ is\ odd\}|$ then by Lemma~\ref{lemma_excess_too_little} the algorithm returns a correct result. 

If $e(X) > |\{i \mid x_i\ is\ odd\}|$ then 
 by Lemma~\ref{lemma_excess_too_little}  an optimal lifting of $(X',C',W)$ is also an optimal lifting of $(X,C,W)$.
So it remains to show that in this case Algorithm \ref{alg_lift_capacity_1_2_3} computes an optimal lifting of 
$(X',C',W)$.

In $(X',C',W)$ all weights are even. By Lemma~\ref{lemma_odd_or_all_even} if $e(X')\ge 2 \cdot |\{i \mid  c'_i \ge 2\}|$ then an optimal solution of $(X'',C'',W)$ (in which each weight $x_i'$ with $c'_i\ge 2$ is lifted twice) is indeed an optimal solution of $(X',C',W)$. All the capacities in $(X'',C'',W)$ are either $0$ or $1$ so we can solve it optimally using Algorithm~\ref{alg_lift_max_capacity_1}.
 
If $e(X') < 2 \cdot |\{i \mid  c'_i \ge 2\}|$ then also by Lemma~\ref{lemma_odd_or_all_even} we are guaranteed that an optimal solution $Y^*$ exists in which all lifted weights are even. Since in both $X'$ and $Y^*$ all the weights are even, $Y^*/2$ must be an optimal solution of $(X'/2,C'/2,W-1)$. And since $c_i \le 3 \Rightarrow \floor{c'_i/2} \le 1$, this is a problem with maximum capacity of $1$ which can be solved optimally using Algorithm~\ref{alg_lift_max_capacity_1}.
\end{proof}

\begin{theorem}
\label{theorem_runtime_l_infinity}
Assuming $W$-bit word operations take $O(1)$ time, Problem~\ref{problem_tcam_approx_by_error} can be solved in $O(k)$ time, and Problem~\ref{problem_tcam_approx_by_rules} can be solved in $\runtimeLiftingTwosided{}$ time, for $\errLinf$.
\end{theorem}

\begin{proof}
Recall  the statement of Theorem~\ref{theorem_equivalent_rule_bound_error_bound1}, in which we define the solution time of Problem~\ref{problem_tcam_approx_by_error} as $T_A$, and argue that Problem~\ref{problem_tcam_approx_by_rules} can be solved in $O(W \cdot (T_A+ T_C))$ time.
Our bounds follow by 
 bounding $T_A$ and $T_C$ as follows.

We begin with $T_A$. We reduce an input $(P,e)$ to a lifting problem $(X,C,W')$ in $O(k)$ time since we compute each coordinate of $X$ and $C$ in $O(1)$ time.
Then, we use either Algorithm~\ref{alg_lift_max_capacity_1} or Algorithm~\ref{alg_lift_capacity_1_2_3}  to solve $(X,C,W')$. Algorithm~\ref{alg_lift_max_capacity_1} requires $O(k)$ time to  identify the weights which should receive a lift. Algorithm~\ref{alg_lift_capacity_1_2_3} has a constant number of phases, in each phase it does some processing that takes $O(k)$ time, and then applies Algorithm~\ref{alg_lift_max_capacity_1}. It follows that  Algorithm~\ref{alg_lift_capacity_1_2_3} also requires $O(k)$ running time. We conclude that $T_A = O(k)$.

We now bound $T_C$ which we recall is the time required to compute for a partition $P$ and an integer $n$ whether $\lambda(P) \le n$. This can be done by finding a shortest sequence that induces $P$
and comparing the length of the resulting sequence to $n$.

We can find a shortest sequence that induces $P$
by running either Bit Matcher or Niagara algorithms on $P$. Bit Matcher runs in  $O(Wk)$ time. If $n$ is small, it is possible to modify Niagara to stop after $n$ steps and conclude that $\lambda(P) > n$. It follows that we can decide with Niagara whether $\lambda(P) \le n$ in $\runtimeNiagaraWithN{}$ time (the initialization is linear in $k$). We conclude that $T_C = \runtimeTc{}$.
\end{proof}

\section{One-sided approximations}
\label{section_one_sided_metric}

In this section we solve Problem~\ref{problem_tcam_approx_by_rules} for $\errLinfPos$ and $\errLinfPosRel$.

\begin{lemma}
\label{lemma_reduction_to_one_sided_lifting}
We can reduce 
Problem~\ref{problem_tcam_approx_by_error} with $\errLinfPos$ or $\errLinfPosRel$ to a lifting problem restricted to inputs in which all weights are  $0$. 
\end{lemma}

\begin{proof}
Let $(P,e)$ be the input to Problem~\ref{problem_tcam_approx_by_error}. Set $x_i = 0$ for all $i$. For $\errLinfPos$, set $\forall i: c_i = \ceil{p_i + e} - 1$ (largest integer strictly-smaller than $p_i + e$),
and we have that $Y$ is a lifting of $X$ (with respect to $C$) iff $\forall i: y_i \le c_i$, which happens iff $\forall i: y_i - p_i < e$.
For $\errLinfPosRel$, set $\forall i: c_i = \ceil{(e+1) \cdot p_i} - 1$, and we have that $Y$ is a lifting of $X$ iff $\forall i: y_i \le c_i$, which happens iff $\forall i: \frac{y_i - p_i}{p_i} < e$.
\end{proof}

\begin{algorithm}[t]
    \DontPrintSemicolon
    \KwIn{
         Non-negative integers $W$ and capacities $C = [c_1 \ldots, c_k]$ such that $2^W \le \sum_i c_i$.
    }
    \KwOut{
        An optimal lifting  $Y$ for $(X = \vec{0},C,W)$.
    }
    
    \looseness=-1
    If $\exists c_i \ge 2^W$: return a lifting $Y$ s.t.\ $y_i = 2^W, y_{\ne i} = 0$. Else:
    
    \textbf{Recursive Version:}
    \begin{enumerate}
    
    \item Let $D = 2(C/2)$.
    
    \item If $\sum_i d_i \le 2^W$:
    return Algorithm\ref{alg_lift_max_capacity_1}$(D,C-D,W)$\\ $/{*}$ \emph{weights vector $D$ and capacities vector $C-D$} ${*}/$
    
    \item Else, $\sum_i d_i > 2^W$:
    recursion: \\
    return $2 \cdot$Algorithm\ref{alg_lift_one_sided_dualversion}$(\vec{0},C/2,W-1)$
    
    \end{enumerate}
    
    \textbf{Equivalent Iterative Version:}
    \begin{enumerate}
    
    \item Find the smallest integer $m \ge 1$ such that $\sum_{i=1}^{k}{\floor{C_i/2^m}} \le 2^{W-m}$. \label{condition_iterative}
    
    \item Set $D^*=2 \cdot (C_i/2^m)$ and $C^* =(C_i/2^{m-1})$. Return $2^{m-1} \cdot$Algorithm\ref{alg_lift_max_capacity_1}$(D^*,C^* - D^*,W-(m-1))$
    
    \end{enumerate}

    \caption{Lifting for one-sided error (2 versions)}
    \label{alg_lift_one_sided_dualversion}
\end{algorithm}

\begin{theorem}
\label{theorem_one_sided_lifting}
Algorithm \ref{alg_lift_one_sided_dualversion} returns an optimal lifting for an input consisting of arbitrary capacities and all weights zero.
\end{theorem}

\begin{proof}
We prove for the recursive version. The lifting problem begins from weights that are all zeros, so all weights are even and $e(X) = 2^W$. Let $D = 2(C/2)$. By Lemma~\ref{lemma_odd_or_all_even}, if $\sum d_i \le 2^W$, then there exists an optimal solution $Y^* \ge D$. So $Y^*$ is also optimal for a new lifting problem $(X',C',W)$ where $X' = D, C' = C-D$. It follows that any optimal solution of $(X',C',W)$ is optimal for the original problem $(\vec{0},C,W)$. But since $\forall i: c'_i = c_i-d_i \le 1$, the new problem can be solved optimally by Algorithm~\ref{alg_lift_max_capacity_1}.

In the other case, when $\sum d_i > 2^W$, Lemma~\ref{lemma_odd_or_all_even} guarantees an optimal solution $Y^*$ such that all of its weights are even. But since $Y^*$ has only even weights, it is not only optimal for $(\vec{0},C,W)$ but also for $(\vec{0},D,W)$. Since all the quantities are even here, then it is optimal iff $Y^*/2$ is optimal for $(\vec{0}/2,D/2,W-1) = (\vec{0},C/2,W-1)$, which proves the correctness of the recursive step.
\end{proof}

\begin{theorem}
\label{theorem_runtime_one_sided}
Assuming $W$-bit word operations take $O(1)$ time, Problem~\ref{problem_tcam_approx_by_error} can be solved in $O(k \cdot \lg W)$ time and Problem~\ref{problem_tcam_approx_by_rules} can be solved in $\runtimeLiftingOnesided{}$ time, for both $\errLinfPos$ and $\errLinfPosRel$.
\end{theorem}

\begin{proof}
As in the proof of Theorem~\ref{theorem_runtime_l_infinity}, we just need to plug-in values for $T_A$ and $T_C$ into the bound of Theorem~\ref{theorem_equivalent_rule_bound_error_bound1}. $T_C = \runtimeTc{}$ as argued in Theorem~\ref{theorem_runtime_l_infinity}.

Although the recursive version of Algorithm~\ref{alg_lift_one_sided_dualversion} may require $O(W)$ recursive calls, its iterative version allows to binary-search for the value of $m$, requiring only $O(\lg W)$ iterations, each consisting of $O(k)$ division and sum operations. We can use binary search since the condition is monotone: $\sum_{i=1}^{k}{\floor{C_i/2^m}}$ decreases at least as fast as $2^{W-m}$ due to potential losses by the integer division.
For the correct $m$, executing Algorithm \ref{alg_lift_max_capacity_1} takes another $O(k)$ time as explained in Theorem~\ref{theorem_runtime_l_infinity}. Therefore we conclude that $T_A = O(k \cdot \lg W)$.

Observe that unlike in Theorem~\ref{theorem_runtime_l_infinity}, in this case $T_A$ is not necessarily dominated by $T_C$, and we get $O(W \cdot (T_A + T_C)) = \runtimeLiftingOnesided{}$.
\end{proof}

\section{Approximations for Non-integer Partitions}
\label{section_non_integer_partitions_approximation}

Our discussion so far assumed that a partition is specified by a vector of $k$ integers that sum to $2^W$ for some width $W$. However, in practice it could be that the desired partition is given by an arbitrary vector of positive numbers such as 1:2:3, or 0.1:0.1:0.8. We may still normalize this vector to sum to $2^W$ (and indeed we assume that this is the case in this section), but the weights are no longer integers. The natural problem that arises is a simple generalization of Problem~\ref{problem_tcam_approx_by_rules}:

\begin{problem}[Length-Bounded Approximation - Non-integer Partition]
\label{problem_tcam_approx_by_rules_non_integer}
Given a partition $P = [p_1, \ldots, p_k]$ where $\forall i: 0 < p_i \in \mathbb{R}$ and $\sum_{i=1}^{k}{p_i}=2^W$, an integer $n$, and a distance-function $D$, find a partition $P' \in \Lambda_n$ such that $D(P',P)$ is minimized, $\sum_{i=1}^{k}{p'_i}=2^W$ and $p'_i$ are non-negative integers for all $1 \le i \le k$.
\end{problem}

The obvious heuristic to solve this generalized problem is simply to round some of the input weights $p_i$ up and some down such that we get an integer partition, and then solve the problem with the rounded partition as an input. While this rounding is likely to produce small error, we can in fact solve this problem optimally with minor modifications to our algorithms, for $\errLinf$, $\errLinfPos$ and $\errLinfPosRel$. In all three cases, we require an adaptation of Theorem~\ref{theorem_equivalent_rule_bound_error_bound1}: While the claim remains true, the details of the reduction slightly change.

\textbf{First, consider $\errLinfPos$.} In this case, we need to slightly extend the reduction described in Theorem~\ref{theorem_equivalent_rule_bound_error_bound1} from Problem~\ref{problem_tcam_approx_by_rules_non_integer} to Problem~\ref{problem_tcam_approx_by_error}. The binary search on the value of the error starts with integers, and proceed as usual until we narrowed the search to an interval of length one, say $e \in (a,a+1]$ for some integer $a$. If the input partition was integer, then the $\errLinfPos$ distance is integer, and since only $a+1$ is possible we could finish the search. However, in the non-integer case, there may still be possible values within this range that produce different lifting problems. For example, if $P = [5.7,5.2,5.1]$, there is a difference between $e=6$ which results in a lifting problem with capacities of $C = [11,11,11]$, $e=5.5$ yielding $C = [11,10,10]$ and $e=5.1$ yielding $C=[10,10,10]$.

Each weight $p_i$ has a single threshold $t_i \in (a,a+1]$ such that $e \le t_i$ gives a different lifting problem than $e > t_i$. Specifically, $t_i = \ceil{p_i} - p_i + a$. We get at most $k$ additional critical values in the range $(a,a+1]$, and we continue the binary search with additional $O(\lg k)$ steps over these thresholds. This results in a total running time of $O((W+\lg k) \cdot (T_A + T_C))$ compared to $\runtimeReductionTaTc{}$ in Theorem~\ref{theorem_equivalent_rule_bound_error_bound1}. Since $k \le 2^W$ to begin with, this means we still get $\runtimeReductionTaTc{}$.

The reduction to a lifting problem is exactly according to Lemma~\ref{lemma_reduction_to_one_sided_lifting}, this is why we needed the floor function, and the lifting problem itself is agnostic to whether the original input partition was integer or not.

\textbf{Next, consider $\errLinf$.} The reduction from Problem~\ref{problem_tcam_approx_by_rules_non_integer} to Problem~\ref{problem_tcam_approx_by_error} is extended just like for $\errLinfPos$ as described above, in two phases: first we narrow-down the error to $e \in (a,a+1]$ for an integer $a$, and then we focus on the critical thresholds within this range. In the two-sided case each coordinate introduces two thresholds, $t_i = \ceil{p_i} - p_i + a$ and $t'_i = p_i - \floor{p_i} + a$. This still only adds $O(\lg k)$ iterations to the binary search. Notice that if $a > 0$ then there exists a non-negative integer $h$ such that $2^h \le e < 2^{h+1}$, and so we can apply Theorem~\ref{theorem_special_case_lifting_reduction} verbatim. We emphasize that Theorem~\ref{theorem_special_case_lifting_reduction}, and the two lemmas it relies upon (Lemma~\ref{lemma_fine_rules_unnecessary} which states that we do not require transactions smaller than $2^h$, and Lemma~\ref{lemma_spread_values} which counts the number of multiples of $2^h$ within range), do not assume that $P$ or $e$ are integer, and therefore still apply. Otherwise, $a=0$ means that $e \le 1$. In this case, for each coordinate $i$ we have at most two integers in the range $(p_i-e,p_i+e)$, simply because the length of this range is $(p_i+e)-(p_i-e) < 2$. Thus, we get a lifting problem with capacities vector $C \in \{0,1\}^k$, and solve it using Algorithm~\ref{alg_lift_max_capacity_1}.

\textbf{Last, consider $\errLinfPosRel$.} The reduction is affected because the separation between different values of distance could be smaller than $2^{-2W}$, and we address that in Remark~\ref{remark_improved_reduction_for_relative_error} below. The second step is the reduction from Problem~\ref{problem_tcam_approx_by_error} to a lifting problem is done exactly as in Lemma~\ref{lemma_reduction_to_one_sided_lifting}. Finally, the lifting problem itself has the form of $(\vec{0},C,W)$ and Algorithm~\ref{alg_lift_one_sided_dualversion} is agnostic to the the original input partition.

Note that in all three adaptations, we may have scenarios that reduce to an infeasible lifting problem, such as $P = [1.7,1.7,0.6]$ or $P = [1.3,1.3,1.4]$ with $e = 0.5$ for $\errLinf$, or do not reduce at all, such as $P = [1.7,1.7,0.6]$ with $e = 0.35$ for $\errLinf$. It means that the bounding error is too small, and we interpret these cases in the binary search the same as ``too many rules are required''.

\begin{remark}[Robust Proof of Theorem~\ref{theorem_equivalent_rule_bound_error_bound1} for $\errLinfPosRel$]
\label{remark_improved_reduction_for_relative_error}
First, observe that there are at most $k \cdot 2^W$ possible values in the range of $\errLinfPosRel(\cdot,P)$, because the distance is a fraction $\frac{a-b}{b} = \frac{a}{b} - 1$ where $a$ is an integer and is at most $2^W$, and $b$ is one of $p_1,\ldots,p_k$. If $p_i$ are all very different we should not expect much fewer distinct values than $k \cdot 2^W$, so the best we should aim for is a binary search with $\lg (k \cdot 2^W) = W + \lg k$ steps. This is achievable:

\begin{enumerate}[label=(\arabic*),topsep=0pt,nosep]
    \item Let $p_j$ be the largest weight. We begin the binary search over error thresholds of the form $\frac{a}{p_j}$ for integer $a$. That is, we apply binary search over the integers $a \in [p_j,2^W]$, and associate them with the error thresholds $\frac{a}{p_j} - 1$. This stage ends when the range containing the target distance is of the form $[\frac{a^*}{p_j}-1,\frac{a^* + 1}{p_j}-1]$.
    
    Because $p_j$ is the largest weight, $\forall i \ne j: \frac{1}{p_j} \le \frac{1}{p_i}$, and therefore the open segment $(\frac{a^*}{p_j}-1,\frac{a^* + 1}{p_j}-1)$ contains at most a single value $\frac{a_i}{p_i}-1$ where $a_i \in [0,2^W]$ is an integer (regardless of $p_i,p_j$ being integers or not). This means that there are at most $k-1$ values in the image of $\errLinfPos(\cdot,P)$ in $(\frac{a^*}{p_j}-1,\frac{a^* + 1}{p_j}-1)$. These values can be computed efficiently.
    
    \item Now, at every iteration we first find the median value out of the remaining candidates, which takes time linear in the number of values (we did not sort them). Then we use this median value for our binary search step. We started this stage with $O(k)$ values, so there are $O(\lg k)$ iterations until we determine the exact error threshold. Each step takes $T_A + T_C$ time, and finding the median in all steps takes $O(k)$ time.
\end{enumerate}

The total running time of the fine-grained binary-search is therefore $(W + \lg k) \cdot (T_A + T_C) + O(k)$, and since $O(k)$ is dominated by both $T_A$ and $T_C$, and $\lg k \le W$, we get total running time of $\runtimeReductionTaTc{}$.
\end{remark}

\section{Experimental Results}
\label{section_experiments}
In this section we analyze the expected error resulting from approximating a partition by a fixed number of rules, where partitions are sampled uniformly from the set of ordered-partitions with $k$ positive parts that sum to $2^W$ (ordered means $[1,3] \ne [3,1]$). That is, the number of targets, $k$, and the TCAM width, $W$, define a distribution over partitions of $2^W$ into $k$ positive parts. Given a maximum number of rules $n$, each partition can be approximated up to some error. Thus, each triplet $(n,k,W)$ defines a distribution over the error. For example, for $k=2,W=2$ the partitions are $[1,3]$,$[2,2]$,$[3,1]$. For $n=1$ the $\errLinf$ and $\errLinfPos$ errors are $1$, $2$, and $1$, respectively, and the $\errLinfPosRel$ errors are $\frac{1}{3}$,$1$ and $\frac{1}{3}$, respectively, while for $n=2$ all the errors are $0$.\footnote{As explained in Section~\ref{section_model}, the algorithms assume all addresses to be equally likely, and we will evaluate them in this way. In the non-uniform case, the resulting error may be larger than expected.}

For a given triplet $(n,k,W)$ we estimated the expectation and the standard deviation of the error as follows. We sampled 1000 partitions uniformly from all ordered-partitions of $2^W$ into $k$ parts, and used our algorithm to compute $n$ rules that induce a partition closest in $\errLinf$ to each of the sampled partitions. Then, we calculated the empirical expectation and variance of
the list of outcomes. We sampled the random ordered-partitions uniformly using the following technique: choose uniformly a subset of $k-1$ different values $b_1,\ldots,b_{k-1} \in \{1,\ldots,2^W-1\}$, and define $b_0 = 0$ and $b_k = 2^W$. Then the $i^{th}$ part is $b_{i} - b_{i-1}$, and it is positive.

In Section~\ref{section_experiments_var_n} we evaluate the dependence of the error on $n$, while fixing $k$ and $W$. Due to lack of space, the following additional experiments are provided in the supplementary material. In Sections~\ref{section_experiments_var_k} and \ref{section_experiments_var_W} we evaluate the dependence of the error on $k$, and $W$, respectively. That is, we fix $n$ and the other parameter ($W$ or $k$, respectively), and analyze how the expected error changes as a function of the third parameter. In Section~\ref{subsection_experiments_ratio_nk} we test how the error depends on $n$ when we fix $W$ and the ratio between $n$ and $k$; In Section~\ref{subsection_experiments_onesided_vs_twosided} we compare the expected $\errLinfPos$ error vs.\ the expected $\errLinf$ error; In Section~\ref{subsection_experiments_niagara_ratio} we compare the error of a heuristic that truncates Niagara's TCAM,  to the optimal error as computed by our algorithms; In Section~\ref{subsection_experiments_running_time} we measure the running times of our algorithms, and  compare it to the time it takes to compute the truncated Niagara TCAM; In Section~\ref{subsection_experiments_degeneracy} we check when we may get a degenerate approximate partition that does not assign any addresses to one or more targets; In Section~\ref{subsection_experiments_real_data_partitions} we analyze the approximation error of ``real partitions'' that we derive from real data.

In all the sections we present only the values of the expectation $\mathbb{E}$. Regarding the variance we note that in almost all our experiments the standard deviation was such that $\sigma / \mathbb{E} \le 1.2$. Exceptions occur in simulations with parameters that result in many partitions that could be represented exactly. In these simulations $\mathbb{E}$ was close to zero and consequently the ratio $\sigma / \mathbb{E}$ blows up. However, in these cases the ratio is no longer meaningful.

Table~\ref{fig_parameters} summarizes the values of $n$, $k$ and $W$ which we used in our experiments. The value of $W=32$ corresponds to the width of an IPv4 address. The value of $k$ depends on the physical setup. For instance, \cite{Niagara} used $k\in\{8,16\}$. We chose the values for $n$ such that exact representation
of a random partition by $n$ rules is unlikely (or impossible), in order to have a nonzero approximation error to analyze.

The implementations of our algorithms for computing the error, as well as an implementation of sampling random ordered partitions and other utilities, are provided in the supplementary material. See it for more details.

\begin{table}[ht]
    \begin{center}
        \begin{tabular}{|c||c|c|c|} 
        \hline
        Figure & $W$ (\#bits) & $k$ (\#targets) & $n$ (\#max rules) \\ 
        \hline
        Fig.~\ref{figure_var_n} & 32 & 5,10,11,20 & [1,70] \\
        \hline
        Fig.~\ref{figure_var_k} {*} & 32 & [4,100] & 25,50,100  \\
        \hline
        Fig.~\ref{figure_var_W} {*} & [10,50] & 10 & 25,50 \\
        \hline
        Fig.~\ref{figure_var_ratio_nk} {*} & 32 & [4,100] & $\frac{ck}{2}, c\in[1,11]$ \\
        \hline
        Fig.~\ref{figure_var_ratio_nk_large} {*} & 32 & $1000x, x\in[1,20]$ & $ck, c \in [1,5]$ \\
        \hline
        Fig.~\ref{figure_oneisded_two_sided_ratio} {*} & 32 & 10 & [10,65] \\
        \hline
        Fig.~\ref{figure_ratio_niagara_all} {*} & 16,32 & 16 & [1,100] \\
        \hline
        Fig.~\ref{figure_running_time_all} {*} & 16,32 & 16 & [1,100] \\
        \hline
        Fig.~\ref{figure_degeneracy} {*} & 32 & $5x, x\in[1,20]$ & $5y, y\in[1,120]$ \\
        \hline
        \end{tabular}
    \end{center}
    \caption{Summary of all simulations' parameters. A universe of $2^W$ addresses is partitioned to $k$ targets. We approximate the partition with $n$ TCAM rules. Items marked with * are in the supplementary material.} \label{fig_parameters}
\end{table}

\subsection{Error as a function of \texorpdfstring{$n$}{n} available Rules}
\label{section_experiments_var_n}
In this section we show the expected $\errLinf$ approximation error, as a function of the number $n$ of TCAM rules. 
The expected error drops exponentially with $n$. This happens since we need a similar number of rules to represent each bit-level. So adding a constant number of rules allows the rules to represent one additional bit-level which in turn reduces the error by a factor of two.
Fig.~\ref{figure_var_n} shows this data in logarithmic scale, for fixed values $W=32$ and $k=5,10,11,20$. The error decreases slower for larger $k$, because more targets require more rules to achieve the same error.

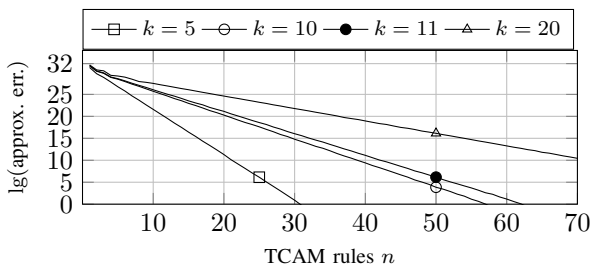
\begin{figure}[t]
	\centering
    \begin{tikzpicture}
        \begin{axis}[
            width=0.45\textwidth,
    		height=0.2\textwidth, 
    		xmin=0,
    		xmax=70,
    		xtick={10,20,30,40,50,60,70,80,90},
    		xlabel=TCAM rules $n$,
    		ymin=0,
            ymax=35,
            ytick={-10,0,5,10,15,20,25,32},
            ylabel=$\lg(\text{approx. err.})$,
            legend style={at={(0.0,1.27)}, anchor=north west,font=\footnotesize,legend columns=-1,},
    		label style={font=\footnotesize},
    		grid=both
        ]
        \addplot [black,mark=square] coordinates{ (25, 6.137) }; \addlegendentry{$k=5$}
        \addplot [black,mark=o] coordinates{ (50, 3.854) }; \addlegendentry{$k=10$}
        \addplot [black,mark=*] coordinates{ (50, 6.158) }; \addlegendentry{$k=11$}
        \addplot [black,mark=triangle] coordinates{ (50, 16.135) }; \addlegendentry{$k=20$}
        
        \addplot [black,mark=.] coordinates{
           (1, 31.126) (2, 29.702) (3, 28.849) (4, 27.817) (5, 26.730) (6, 25.765) (7, 24.648) (8, 23.664) (9, 22.601) (10, 21.583) (11, 20.553) (12, 19.531) (13, 18.460) (14, 17.481) (15, 16.404) (16, 15.420) (17, 14.344) (18, 13.367) (19, 12.352) (20, 11.300) (21, 10.138) (22, 9.136) (23, 8.137) (24, 7.103) (25, 6.137) (26, 5.068) (27, 4.041) (28, 3.015) (29, 1.956) (30, 0.902) (31, -0.265) (32, -1.466) (33, -3.540) (34, -5.573) (35, -6.381)
        }; 
        
        \addplot [black,mark=.] coordinates{
            (1, 31.505) (2, 30.260) (3, 29.788) (4, 28.868) (5, 28.452) (6, 27.916) (7, 27.332) (8, 26.824) (9, 26.261) (10, 25.703) (11, 25.161) (12, 24.628) (13, 24.040) (14, 23.552) (15, 22.950) (16, 22.461) (17, 21.905) (18, 21.311) (19, 20.829) (20, 20.284) (21, 19.730) (22, 19.147) (23, 18.600) (24, 18.131) (25, 17.612) (26, 17.034) (27, 16.490) (28, 15.865) (29, 15.398) (30, 14.786) (31, 14.261) (32, 13.727) (33, 13.159) (34, 12.733) (35, 12.161) (36, 11.536) (37, 10.984) (38, 10.486) (39, 9.926) (40, 9.401) (41, 8.804) (42, 8.314) (43, 7.779) (44, 7.231) (45, 6.679) (46, 6.160) (47, 5.515) (48, 5.041) (49, 4.516) (50, 3.854) (51, 3.425) (52, 2.897) (53, 2.288) (54, 1.835) (55, 1.139) (56, 0.647) (57, 0.014) (58, -0.732) (59, -1.531) (60, -2.322) (61, -3.120)
        }; 
     		
        \addplot [black,mark=.] coordinates{
           (1, 31.540) (2, 30.321) (3, 29.890) (4, 28.953) (5, 28.570) (6, 28.081) (7, 27.536) (8, 27.053) (9, 26.576) (10, 26.048) (11, 25.537) (12, 25.039) (13, 24.594) (14, 24.079) (15, 23.523) (16, 23.043) (17, 22.620) (18, 22.025) (19, 21.571) (20, 21.137) (21, 20.567) (22, 20.076) (23, 19.574) (24, 19.057) (25, 18.583) (26, 18.143) (27, 17.565) (28, 17.100) (29, 16.529) (30, 16.044) (31, 15.596) (32, 15.058) (33, 14.586) (34, 14.146) (35, 13.614) (36, 13.114) (37, 12.592) (38, 12.130) (39, 11.596) (40, 11.127) (41, 10.586) (42, 10.051) (43, 9.635) (44, 9.134) (45, 8.623) (46, 8.140) (47, 7.570) (48, 7.154) (49, 6.649) (50, 6.158) (51, 5.588) (52, 5.089) (53, 4.635) (54, 4.159) (55, 3.702) (56, 3.177) (57, 2.600) (58, 2.260) (59, 1.646) (60, 1.169) (61, 0.602) (62, 0.079) (63, -0.639) (64, -1.311) (65, -2.035) (66, -2.806) (67, -4.059) (68, -5.718) (69, -6.796) (70, -6.966)
        }; 
     	
     	\addplot [black,mark=.] coordinates{
           (1, 31.714) (2, 30.566) (3, 30.360) (4, 29.336) (5, 29.179) (6, 28.925) (7, 28.456) (8, 27.920) (9, 27.741) (10, 27.501) (11, 27.194) (12, 26.915) (13, 26.619) (14, 26.322) (15, 26.060) (16, 25.764) (17, 25.465) (18, 25.185) (19, 24.888) (20, 24.616) (21, 24.327) (22, 24.064) (23, 23.774) (24, 23.495) (25, 23.198) (26, 22.934) (27, 22.617) (28, 22.353) (29, 22.051) (30, 21.774) (31, 21.505) (32, 21.232) (33, 20.931) (34, 20.655) (35, 20.383) (36, 20.089) (37, 19.735) (38, 19.496) (39, 19.266) (40, 18.953) (41, 18.647) (42, 18.350) (43, 18.146) (44, 17.802) (45, 17.484) (46, 17.280) (47, 17.001) (48, 16.651) (49, 16.411) (50, 16.135) (51, 15.828) (52, 15.567) (53, 15.275) (54, 14.994) (55, 14.703) (56, 14.405) (57, 14.138) (58, 13.854) (59, 13.567) (60, 13.271) (61, 12.971) (62, 12.718) (63, 12.408) (64, 12.189) (65, 11.810) (66, 11.616) (67, 11.312) (68, 10.995) (69, 10.736) (70, 10.443) (71, 10.154) (72, 9.868) (73, 9.557) (74, 9.307) (75, 9.043) (76, 8.762) (77, 8.495) (78, 8.188) (79, 7.904) (80, 7.589) (81, 7.340) (82, 7.029) (83, 6.711) (84, 6.459) (85, 6.244) (86, 5.871) (87, 5.643) (88, 5.402) (89, 5.085) (90, 4.781) (91, 4.561) (92, 4.189) (93, 3.964) (94, 3.751) (95, 3.395) (96, 3.145) (97, 2.850) (98, 2.472) (99, 2.271) (100, 1.955)
        }; 
     	
        \end{axis}
    \end{tikzpicture}
    \caption{\looseness=-1 Expected $\errLinf$ approximation error (log-scale) as a function of the number of TCAM rules $n \in [1,70]$, for address-width $W=32$ and $k=5,10,11,20$ targets.\label{figure_var_n}}
\end{figure}

\section{Related Work}
\label{section_related_work}

\textbf{Matching-based Implementations:}
\looseness=-1
The work of \cite{AccurateExp} showed that for two possible targets, the complexity of a partition $[x,2^W-x]$  is exactly $\min(\phi(x), \phi(2^W-x)) + 1$ rules, where $\phi(y)$  is the number of powers  in a signed bit representation without adjacent powers with a non-zero coefficient~\cite{Encyclopedia}.  An earlier work \cite{WangBR11} considered only restricted TCAM encodings in which rules are \emph{disjoint}. For instance, the partition $P=[4,3,1]$ for $W=3$ is implemented with the four rules $(\textsc{0**} \to 1, \textsc{10*} \to 2, \textsc{110} \to 2, \textsc{111} \to 3)$. Since TCAMs allow overlapping rules and resolve overlaps by ordering the rules, this early approach does not take full advantage of them. For example $P$ can also be implemented by prioritizing longer prefix rules as $(\textsc{0**} \to 1, \textsc{111} \to 3, \textsc{1**} \to 2)$.

\textbf{Hashing-based Implementations:}
Hash-based solutions for load-balancing use an array,
each of its cells contains a target. The fraction of the cells containing a particular target determines the fraction of the addresses that this target gets. 
This solution is also known as WCMP~\cite{WCMP, Zegura} or as 
ECMP~\cite{RFC2992} when traffic is split equally. \cite{WuTang} studies the relation between the size of the array and how good it approximates a desired distribution.
While the above works studied a fixed output distribution, in a dynamic scenario mapping has to be updated following a change in the required  distribution. ~\cite{Chao08, Chim, KandulaKSB07} considered such updates for load balancing over multiple paths. They suggested update schemes that reduce transient negative impact of packet reordering.
A recent approach~\cite{DASH_alg} refrains from memory blowup by comparing the hash to range-boundaries. Since the hash is tested sequentially against each range, it restricts the total number of load-balancing targets.

\textbf{Partitions vs. Functions:} This paper studies efficient representations of partitions. A \emph{partition} specifies the number of addresses that have to be mapped to each possible target but leaves the freedom to choose these addresses. In contrast a \emph{function} specifies exactly the target of each address. Note that there may still be multiple ways to implement a function with a TCAM. Finding the smallest list of prefix rules that realizes a given function can be done in polynomial time with dynamic programming~\cite{DravesKSZ99, suri03}.  When we are not restricted to prefix rules the problem is  NP-hard~\cite{McGeerY09}. The particular family of ``range functions'' where the preimage of each target is an interval was carefully studied due to its popularity in packet classifiers for access control~\cite{Gray, ORange}. Going back to implementing partitions, ~\cite{AccurateExp} proved that any partition to two targets has an optimal realization as a range function.

\section{Conclusions and Future Work}
\label{section_conclusions}
In this paper we gave efficient algorithms to find a TCAM with
$n$ rules that induces a partition that best approximates a
given partition in $\errLinf$, its one-sided variant \errLinfPos, and the relative one-sided variant $\errLinfPosRel$. 
A relative distance measure is desired when the deviation could increase with the load. An absolute distance measure is appropriate when we try to keep all deviations below the same fixed threshold. Finding the closest partition with respect to the two sided relative measure $\errLinfRel$ is left as an open problem.

We did so by formalizing a lifting problem, and solving some special cases of it. We observed experimentally that truncating an optimal TCAM of the exact partition gives a good approximate partition. We conjecture that it is never off by a factor of more than $2$ than the closest partition we can get with $n$ rules. We also observed experimentally that the \errLinfPos and $\errLinf$ errors are similar on average, and derived empirically an expression of the error as a function of $n$, $k$, and $W$.

An interesting direction for future research is to consider
additional distances. For example, we can prove that a truncated Niagara sequence gives a closest partition to the target partition in $\errLsum$ distance. The two-sided maximum relative error is also interesting to study. In the bounded error version of this problem we search in a neighbourhood which is a hyper-box instead of a hyper-cube. The lifting formalization can come handy here: we will get different capacities for different coordinates and will need to solve more general lifting problems than the ones we solved here (Theorem~\ref{theorem_special_case_lifting_reduction} and Lemma~\ref{lemma_reduction_to_one_sided_lifting}).
Other ``more-practical'' interesting questions would be to solve the same problems for non-uniform address-space (when each address may have a different weight), or consider how to update the tables per changing demand on-the-fly, while incurring minimal impact.

\section{Acknowledgment}
The work of Haim Kaplan and Yaniv Sadeh was partially supported by Israel Science Foundation (ISF) grant numbers 1841-14 and 1595-19, German Science Foundation (GIF) grant number 1367 and the Blavatnik research fund at Tel Aviv University. The work of Ori Rottenstreich was partially supported by the Taub Family Foundation as well as by the Technion Hiroshi Fujiwara Cyber Security Research Center and the Israel National Cyber Directorate, by Alon fellowship, by German-Israeli Foundation (GIF) Young Scientists Program and by the Gordon Fund for System Engineering.

\bibliographystyle{IEEEtran}
\bibliography{IEEEabrv, reference}

\begin{thebibliography}{10}
\providecommand{\url}[1]{#1}
\csname url@samestyle\endcsname
\providecommand{\newblock}{\relax}
\providecommand{\bibinfo}[2]{#2}
\providecommand{\BIBentrySTDinterwordspacing}{\spaceskip=0pt\relax}
\providecommand{\BIBentryALTinterwordstretchfactor}{4}
\providecommand{\BIBentryALTinterwordspacing}{\spaceskip=\fontdimen2\font plus
\BIBentryALTinterwordstretchfactor\fontdimen3\font minus
  \fontdimen4\font\relax}
\providecommand{\BIBforeignlanguage}[2]{{%
\expandafter\ifx\csname l@#1\endcsname\relax
\typeout{** WARNING: IEEEtran.bst: No hyphenation pattern has been}%
\typeout{** loaded for the language `#1'. Using the pattern for}%
\typeout{** the default language instead.}%
\else
\language=\csname l@#1\endcsname
\fi
#2}}
\providecommand{\BIBdecl}{\relax}
\BIBdecl

\bibitem{sadehconext}
Y.~{Sadeh}, O.~{Rottenstreich}, and H.~{Kaplan}, ``Optimal approximations for
  traffic distribution in bounded switch memories,'' in \emph{ACM CoNEXT},
  2020.

\bibitem{split_sigcomm05}
S.~Kandula, D.~Katabi, B.~Davie, and A.~Charny, ``Walking the tightrope:
  Responsive yet stable traffic engineering,'' in \emph{ACM SIGCOMM}, 2005.

\bibitem{split_conext07}
J.~He, M.~Suchara, M.~Bresler, J.~Rexford, and M.~Chiang, ``Rethinking internet
  traffic management: From multiple decompositions to a practical protocol,''
  in \emph{ACM CoNEXT}, 2007.

\bibitem{split_sigcomm14}
M.~Alizadeh, T.~Edsall, S.~Dharmapurikar, R.~Vaidyanathan, K.~Chu,
  A.~Fingerhut, V.~T. Lam, F.~Matus, R.~Pan, N.~Yadav, and G.~Varghese,
  ``Conga: Distributed congestion-aware load balancing for datacenters,'' in
  \emph{ACM SIGCOMM}, 2014.

\bibitem{al2008scalable}
M.~Al-Fares, A.~Loukissas, and A.~Vahdat, ``A scalable, commodity data center
  network architecture,'' in \emph{{ACM SIGCOMM}}, 2008.

\bibitem{Ananta}
P.~Patel, D.~Bansal, L.~Yuan, A.~Murthy, A.~G. Greenberg, D.~A. Maltz, R.~Kern,
  H.~Kumar, M.~Zikos, H.~Wu, C.~Kim, and N.~Karri, ``Ananta: {Cloud} scale load
  balancing,'' in \emph{{ACM} {SIGCOMM}}, 2013.

\bibitem{RFC2992}
C.~Hopps, ``Analysis of an equal-cost multi-path algorithm,'' Nov. 2000, {RFC
  2992}.

\bibitem{WCMP}
J.~Zhou, M.~Tewari, M.~Zhu, A.~Kabbani, L.~Poutievski, A.~Singh, and A.~Vahdat,
  ``{WCMP:} {Weighted} cost multipathing for improved fairness in data
  centers,'' in \emph{EuroSys}, 2014.

\bibitem{WangBR11}
R.~Wang, D.~Butnariu, and J.~Rexford, ``Openflow-based server load balancing
  gone wild,'' in \emph{{USENIX Hot-ICE}}, 2011.

\bibitem{Niagara}
N.~Kang, M.~Ghobadi, J.~Reumann, A.~Shraer, and J.~Rexford, ``Efficient traffic
  splitting on commodity switches,'' in \emph{ACM CoNEXT}, 2015.

\bibitem{AccurateExp}
O.~Rottenstreich, Y.~Kanizo, H.~Kaplan, and J.~Rexford, ``Accurate traffic
  splitting on commodity switches,'' in \emph{{ACM SPAA}}, 2018.

\bibitem{appelman2012performance}
M.~Appelman and M.~de~Boer, ``Performance analysis of {OpenFlow} hardware,''
  \emph{University of Amsterdam, Tech. Rep}, 2012.

\bibitem{McKeownABPPRST08}
N.~McKeown, T.~Anderson, H.~Balakrishnan, G.~M. Parulkar, L.~L. Peterson,
  J.~Rexford, S.~Shenker, and J.~S. Turner, ``Openflow: Enabling innovation in
  campus networks,'' \emph{Computer Communication Review}, vol.~38, no.~2, pp.
  69--74, 2008.

\bibitem{BitMatcher}
Y.~{Sadeh}, O.~{Rottenstreich}, A.~{Barkan}, Y.~{Kanizo}, and H.~{Kaplan},
  ``Optimal representations of a traffic distribution in switch memories,''
  \emph{{IEEE/ACM} Trans. Netw.}, vol.~28, no.~2, pp. 930--943, 2020.

\bibitem{LPM_TCAM}
S.~{Kasnavi}, V.~C. {Gaudet}, P.~{Berube}, and J.~N. {Amaral}, ``A
  hardware-based longest prefix matching scheme for {TCAM}s,'' in \emph{IEEE
  International Symposium on Circuits and Systems}, 2005.

\bibitem{NSDIJose15}
L.~Jose, L.~Yan, G.~Varghese, and N.~McKeown, ``Compiling packet programs to
  reconfigurable switches,'' in \emph{{USENIX} {NSDI}}, 2015.

\bibitem{Forwarding13}
P.~Bosshart, G.~Gibb, H.~Kim, G.~Varghese, N.~McKeown, M.~Izzard, F.~A. Mujica,
  and M.~Horowitz, ``Forwarding metamorphosis: {F}ast programmable match-action
  processing in hardware for {SDN},'' in \emph{{ACM} {SIGCOMM}}, 2013.

\bibitem{FlexPipe}
R.~Ozdag, ``{Intel\textregistered Ethernet Switch FM6000 Series-Software
  Defined Networking},'' \emph{Intel Coroporation}, 2012.

\bibitem{Kang2014NiagaraSL}
N.~{Kang}, M.~{Ghobadi}, J.~{Reumann}, A.~{Shraer}, and J.~{Rexford},
  ``Niagara: Scalable load balancing on commodity switches,'' Princeton, Tech.
  Rep. TR-973-14, 2014.

\bibitem{garey1979computers}
M.~R. Garey and D.~S. Johnson, \emph{Computers and Intractability: A Guide to
  the Theory of NP-Completeness}.\hskip 1em plus 0.5em minus 0.4em\relax USA:
  W. H. Freeman \& Co., 1979.

\bibitem{Encyclopedia}
N.~J.~A. Sloane and S.~Plouffe, ``{The Encyclopedia of Integer Sequences},''
  \emph{Academic Press}, 1995.

\bibitem{Zegura}
Z.~Cao, Z.~Wang, and E.~W. Zegura, ``Performance of hashing-based schemes for
  internet load balancing,'' in \emph{{IEEE} {INFOCOM}}, 2000.

\bibitem{WuTang}
N.~Wu, S.~Tseng, and A.~Tang, ``Accurate rate-aware flow-level traffic
  splitting,'' in \emph{Allerton Conference on Communication, Control, and
  Computing}, 2018.

\bibitem{Chao08}
N.~S. Artan, H.~Yuan, and H.~J. Chao, ``A dynamic load-balanced hashing scheme
  for networking applications,'' in \emph{IEEE {GLOBECOM}}, 2008.

\bibitem{Chim}
T.~W. Chim, K.~L. Yeung, and K.~Lui, ``Traffic distribution over
  equal-cost-multi-paths,'' \emph{Computer Networks}, vol.~49, no.~4, pp.
  465--475, 2005.

\bibitem{KandulaKSB07}
S.~Kandula, D.~Katabi, S.~Sinha, and A.~W. Berger, ``Dynamic load balancing
  without packet reordering,'' \emph{Computer Communication Review}, vol.~37,
  no.~2, pp. 51--62, 2007.

\bibitem{DASH_alg}
K.-F. Hsu, P.~Tammana, R.~Beckett, A.~Chen, J.~Rexford, and D.~Walker,
  ``Adaptive weighted traffic splitting in programmable data planes,'' in
  \emph{Proceedings of the Symposium on SDN Research}, 2020.

\bibitem{DravesKSZ99}
R.~Draves, C.~King, S.~Venkatachary, and B.~Zill, ``Constructing optimal {IP}
  routing tables,'' in \emph{IEEE Infocom}, 1999.

\bibitem{suri03}
S.~Suri, T.~Sandholm, and P.~R. Warkhede, ``Compressing two-dimensional routing
  tables,'' \emph{Algorithmica}, vol.~35, no.~4, pp. 287--300, 2003.

\bibitem{McGeerY09}
R.~McGeer and P.~Yalagandula, ``Minimizing rulesets for {TCAM}
  implementation,'' in \emph{IEEE {INFOCOM}}, 2009.

\bibitem{Gray}
A.~Bremler{-}Barr and D.~Hendler, ``Space-efficient {TCAM}-based classification
  using gray coding,'' \emph{{IEEE} Trans. Computers}, vol.~61, no.~1, pp.
  18--30, 2012.

\bibitem{ORange}
L.~Schiff, Y.~Afek, and A.~Bremler{-}Barr, ``Orange: Multi field openflow based
  range classifier,'' in \emph{{ACM/IEEE} {ANCS}}, 2015.

\bibitem{ftp_packet_captures}
R.~Pang and V.~Paxson, ``A high-level programming environment for packet trace
  anonymization and transformation,'' in \emph{ACM SIGCOMM}, 2003.

\end{thebibliography}

\vspace{-0.4in}
\begin{IEEEbiography}[{\includegraphics[width=1.15in,height=1.42in,clip,keepaspectratio]{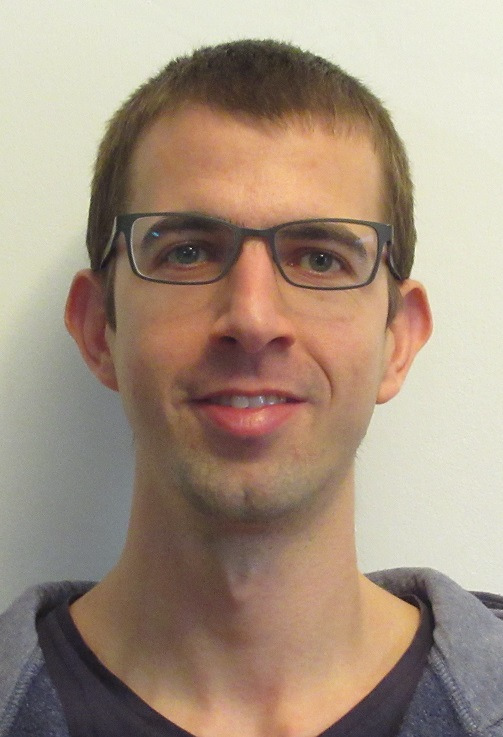}}]
{Yaniv Sadeh} is a PhD student in Computer Science at Tel Aviv University. He received his BSc in Mathematics and Computer Science from the Open University of Israel in 2017 and his MSc in Computer Science from Tel Aviv University in 2021.
\end{IEEEbiography}
\vspace{-0.4in}
\begin{IEEEbiography}[{\includegraphics[width=1.15in,height=1.42in,clip,keepaspectratio]{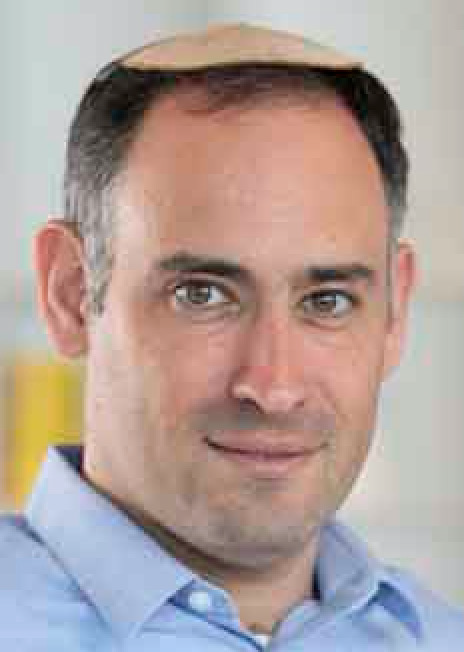}}]
{Ori Rottenstreich} is an assistant professor at the department of Computer Science and the department of Electrical Engineering of the Technion, Haifa, Israel.  In 2015-2017 he was a Postdoctoral Research Fellow at Princeton university. Earlier, he received the BSc in Computer Engineering (summa cum laude), and PhD degree from the Technion in 2008 and 2014, respectively. 
\end{IEEEbiography}
\vspace{-0.4in}
\begin{IEEEbiography}[{\includegraphics[width=1.10in,height=1.42in,clip,keepaspectratio]{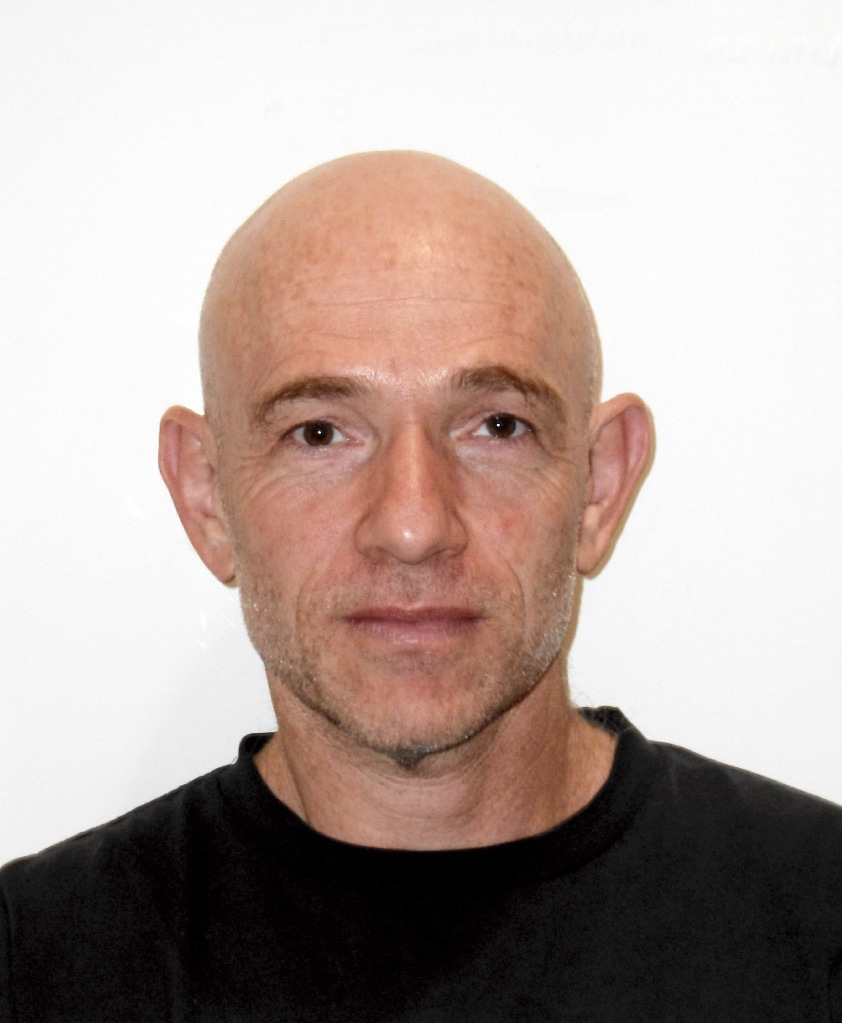}}]
{Haim Kaplan} received his PhD degree from Princeton University at 1997. He was a member of technical stuff at AT\&T research from 1996 to 1999. Since 1999 he is a Professor in the School of Computer Science at Tel Aviv University. His research interests are design and analysis of algorithms and data structures.
\end{IEEEbiography}

{
\newpage 
\setcounter{page}{1}
\section*{Optimal Weighted Load Balancing in TCAMs - SUPPLEMENTARY MATERIALS}


\section{Reviewing Bit Matcher and Niagara}
\label{section_review_bitmatcher_and_niagara}

In this Section we provide pseudo-code (see Algorithm~\ref{algorithms_bitmatcher_and_niagara}) and a short description of the  algorithms Bit Matcher~\cite{BitMatcher} and Niagara~\cite{Niagara}. This description is based on~\cite{BitMatcher}.

Both algorithms are given a list of weights $x_1,\ldots,x_k$ that sum to $2^W$ as input, and compute a sequence of transactions that can be mapped to a TCAM. Bit Matcher can be implemented in $O(Wk)$  time, and Niagara can be implemented in $O(Wk\log k)$  time.

The core idea of Bit Matcher is that any odd weight of the partition must participate in a transaction of size $1$. After performing these transactions all weights whose second least significant bit is $1$ must participate in a transaction of size $2$  and so on. The order $<_{lex}$ (Definition~\ref{definition_lexicorder}) helps to identify which transactions to make at lower levels such that there will be fewer $1$-bits at higher levels. Loosely speaking,
we choose transactions that
 cancel as many $1$-bits as possible in the binary representation of the weights, using the carry caused by each one of these transaction.
 
A high level overview of Niagara is as follows. 
Consider the target partition $P^0 = (x_1,\ldots,x_k)$. The algorithm maintains an implicit partition $P^1$, such that the vector $P$ in the algorithm satisfies $P = P^0 - P^1$. Initially, when the TCAM's default rule is allocated to target $i$, $P^1 = (0, \ldots, 0, 2^W, 0, \ldots, 0)$ (non-zero in index $i$). Since $P = P^0-P^1$ we have $\sum(P)=0$. Thus while $P \ne (0,\ldots,0)$ it must contain a positive and a negative difference. Let $i$ be a coordinate with maximum (positive) value in $P$ and let $j$ be a coordinate with minimum (negative) value in $P$, then Niagara refines $P$ by computing a value  $\lambda = 2^h$  such that when moving $\lambda$ from $p_i$ to $p_j$, the $L_1$ norm of $P$ following this transaction is minimized.

\begin{algorithm}[ht]
    \SetAlgoLined
    \DontPrintSemicolon
    \KwIn{
          A partition $P = [p_1, \ldots, p_k]$, each $p_i$ is a non-negative integer and $\sum_{i=1}^k p_i =2^W$.
    }
    \KwOut{
        Shortest sequence $s$ of transactions (of
        sizes which are powers of $2$), that zeroes $P$.
    }
    
    \SetKwFunction{funcBitMatcher}{BitMatcher}
    \SetKwFunction{funcNiagara}{Niagara}
    \SetKwProg{Fn}{Function}{:}{} \SetKwProg{Fn}{}{:}{} 
    
    Initialize $s$ to be an empty sequence. Then call either \funcBitMatcher{$P$} or \funcNiagara{$P$}, and \Return $s$.
    \;\;
    
    \Fn{\funcBitMatcher{$P$}}{
        \For{level $\ell = 0 \ldots W-1$} {
            \smallskip
     
            Let $A$ be the set of weights whose $\ell^{\text{th}}$ bit is $1$ (i.e. $x_i\in A$ iff bitwise $x_i\&2^\ell \ne 0$).
    
            \looseness=-1
            Let $A_h$ consists of the $|A|/2$ largest weights in $A$ in bit lexicographic order, and let $A_l$ consists of the $|A|/2$ smallest weights in $A$ in bit lexicographic order.
    
            Pair the elements of $A_h$ and $A_l$ arbitrarily and for each pair append to $s$ and apply to $P$ a transaction of value $2^\ell$ from the smallest weight to the largest weight in the pair.
        }
        Let $i$ be such that $p_i = 2^W$. Add to $s$ the transaction $[i \to_W \bot]$.
    }
    \;
    \Fn{\funcNiagara{$P$}}{
        Let $i = argmax(P)$. Subtract $2^W$ from $p_i$, and add to $s$ the transaction $[i \to_W \bot]$.
    
        \While{$P \ne (0, \ldots, 0)$} {
    	    Let $i = argmax(\Delta)$, $j = argmin(\Delta)$, and let $h$ be the largest integer minimizing $\Big(|p_i-2^h| + |p_j+2^h|\Big)$.
        
            Apply a transaction $[i\rightarrow_h j]$ to $P$ and append it to $s$.
        }
    }
    \caption{Bit Matcher and Niagara Algorithms}
    \label{algorithms_bitmatcher_and_niagara}
\end{algorithm}


\section{Supplementary Code Implementation}
\label{section_code_info}
We provide an implementation of our algorithms as supplementary material for two main purposes:
\begin{enumerate}
    \item Clarify low level details that the reader may find unclear or is curious about. For this purpose, the code includes both comments and test-cases which can be used as additional examples.
    
    \item Allowing any interested party to implement any of the algorithms described easily, in their chosen environment and language, and provide tests for verification.
\end{enumerate}

\noindent
The implementation covers the following in the paper:

\begin{enumerate}
    \item Utilities: bit-lexicographic comparison (Definition~\ref{definition_lexicorder}), partition complexity (Definition~\ref{definition_length_of_partition}), computation of the $\errLinf$, $\errLinfPos$ and $\errLinfPosRel$ distances, and sampling ordered-partitions as described in Section~\ref{section_experiments}.

    \item Lifting solvers: Algorithm~\ref{alg_lift_max_capacity_1}, Algorithm~\ref{alg_lift_capacity_1_2_3} and Algorithm~\ref{alg_lift_one_sided_dualversion}.
    
    \item Partition approximators: reductions from rule-based to error-based approximation (according to the reduction of Theorem~\ref{theorem_equivalent_rule_bound_error_bound1}), and from error-based approximation to a lifting problem (according to Lemma~\ref{theorem_special_case_lifting_reduction} and Lemma~\ref{lemma_reduction_to_one_sided_lifting}).
\end{enumerate}

\noindent
We provide python code for clarity. It allows us to deliver the logic clearly and reduce language-specific implementation details.


\section{Supplementary Experiments}
\label{section_supplementary_experiments}

\subsection{Error as a function of \texorpdfstring{$k$}{k} Targets}
\label{section_experiments_var_k}
In this section we examine the expected $\errLinf$ approximation error as a function of the number of targets $k$. There are two opposite effects to consider. First, the number of rules is fixed so the larger $k$ is, the harder it is to approximate the partition. On the other hand, the average weight in the partition is $\frac{2^W}{k}$ so the larger $k$ is, the weights get smaller and therefore the error should get smaller as well. However, since the error is defined by the maximal error of the targets rather than the average, this latter effect should be weak.
Fig.~\ref{figure_var_k} shows the results for  $W=32$ and $n=25,50,100$. For $k \le 100$, we see
that the error grows with $k$. So the increase in error due to the larger number of targets is more dominant than the decrease in error due to the decrease in the average target size.
Interestingly, the growth is approximately linear, except for small values of $k$
for which we have many zero-error partitions. For $n=25$ the growth is linear when $k \ge 20$, for $n=50$ when $k\ge 40$ and for $n=100$ when $k \ge 80$. The case $n=25$ begins to show growth slowdown, since the error approaches its maximum possible value ($W=32$, and the error already passed $2^{27}$).

\begin{figure}[t]
    \subfigure[$n=25$ Rules.] {
        \hspace*{-1.0em}  
        \begin{tikzpicture}
            \begin{axis}[
                width=0.2\textwidth,
        		height=0.2\textwidth,
        		xmin=0,
        		xmax=100,
        		xtick={0,20,40,60,80,100},
        		xlabel=Number of targets $k$,
        		ymin=0, ymax = 160,
                ytick={32,64,96,128,160},
                ylabel=$(\text{approx.\ err.}) / 2^{20}$,
                style={font=\scriptsize}, 
        		legend style={at={(0.01,0.99)}, anchor=north west,font=\footnotesize,},
        		label style={font=\footnotesize},
        		grid=both
            ]
            \addplot [black] coordinates{
            (4, 0.000) (5, 0.000) (6, 0.001) (7, 0.006) (8, 0.025) (9, 0.075) (10, 0.184) (11, 0.376) (12, 0.694) (13, 1.165) (14, 1.744) (15, 2.568) (16, 3.546) (17, 4.660) (18, 5.872) (19, 7.711) (20, 9.292) (21, 10.917) (22, 12.987) (23, 15.124) (24, 17.194) (25, 19.514) (26, 21.270) (27, 23.370) (28, 25.707) (29, 27.906) (30, 30.483) (31, 33.143) (32, 35.725) (33, 38.471) (34, 40.324) (35, 42.541) (36, 44.682) (37, 47.117) (38, 49.099) (39, 51.049) (40, 52.418) (41, 54.987) (42, 56.188) (43, 59.531) (44, 60.893) (45, 63.211) (46, 65.749) (47, 67.510) (48, 68.895) (49, 71.632) (50, 73.455) (51, 75.359) (52, 77.116) (53, 79.149) (54, 81.623) (55, 83.416) (56, 86.457) (57, 88.171) (58, 88.829) (59, 90.995) (60, 93.638) (61, 95.821) (62, 96.697) (63, 98.801) (64, 101.033) (65, 102.585) (66, 104.044) (67, 105.569) (68, 106.887) (69, 108.518) (70, 110.642) (71, 111.676) (72, 112.102) (73, 114.511) (74, 116.100) (75, 116.739) (76, 118.520) (77, 119.024) (78, 121.080) (79, 121.400) (80, 123.217) (81, 123.723) (82, 124.604) (83, 126.482) (84, 127.274) (85, 128.371) (86, 130.427) (87, 130.289) (88, 131.685) (89, 132.892) (90, 133.769) (91, 134.775) (92, 135.716) (93, 136.642) (94, 136.771) (95, 138.313) (96, 139.145) (97, 139.888) (98, 141.111) (99, 141.090) (100, 142.006)
            }; 
            
        \end{axis}
    \end{tikzpicture} 
    \label{figure_var_k_sub1}
    }
    \hspace*{-2.1em}  
    ~
    \subfigure[$n=50$ Rules.] {
        \begin{tikzpicture}
            \begin{axis}[
                width=0.2\textwidth,
        		height=0.2\textwidth,
        		xmin=0,
        		xmax=100,
        		xtick={0,20,40,60,80,100},
        		xlabel=Number of targets $k$,
        		ymin=0, ymax = 40,
                ytick={8,16,24,32,40},
                style={font=\scriptsize}, 
        		label style={font=\footnotesize},
        		grid=both
            ]
            \addplot [black] coordinates{
            (4, 0.000) (5, 0.000) (6, 0.000) (7, 0.000) (8, 0.000) (9, 0.000) (10, 0.000) (11, 0.000) (12, 0.000) (13, 0.001) (14, 0.002) (15, 0.004) (16, 0.008) (17, 0.015) (18, 0.026) (19, 0.042) (20, 0.068) (21, 0.100) (22, 0.148) (23, 0.207) (24, 0.274) (25, 0.354) (26, 0.463) (27, 0.602) (28, 0.757) (29, 0.899) (30, 1.092) (31, 1.314) (32, 1.543) (33, 1.821) (34, 2.086) (35, 2.398) (36, 2.703) (37, 3.053) (38, 3.416) (39, 3.842) (40, 4.184) (41, 4.669) (42, 5.025) (43, 5.621) (44, 6.153) (45, 6.535) (46, 7.120) (47, 7.664) (48, 8.174) (49, 8.748) (50, 9.107) (51, 9.780) (52, 10.364) (53, 10.862) (54, 11.383) (55, 12.044) (56, 12.532) (57, 13.113) (58, 13.578) (59, 14.261) (60, 14.972) (61, 15.560) (62, 16.163) (63, 16.827) (64, 17.266) (65, 17.778) (66, 18.488) (67, 19.328) (68, 19.590) (69, 20.067) (70, 20.668) (71, 21.086) (72, 21.849) (73, 22.314) (74, 22.756) (75, 23.387) (76, 23.828) (77, 24.175) (78, 24.821) (79, 25.354) (80, 25.750) (81, 26.260) (82, 26.818) (83, 26.990) (84, 27.863) (85, 28.348) (86, 28.919) (87, 29.619) (88, 29.814) (89, 30.329) (90, 31.311) (91, 31.317) (92, 32.248) (93, 32.597) (94, 33.454) (95, 33.851) (96, 34.428) (97, 35.158) (98, 35.679) (99, 36.231) (100, 36.777)
            }; 
            
        \end{axis}
    \end{tikzpicture}
    \label{figure_var_k_sub2}
    }
    \hspace*{-2.1em}  
    ~
    \subfigure[$n=100$ Rules.] {
        \begin{tikzpicture}
            \begin{axis}[
                width=0.2\textwidth,
        		height=0.2\textwidth,
        		xmin=0,
        		xmax=100,
        		xtick={0,20,40,60,80,100},
        		xlabel=Number of targets $k$,
        		ymin=0, ymax = 5,
                ytick={1,2,3,4,5},
                style={font=\scriptsize}, 
        		legend style={at={(0.01,0.99)}, anchor=north west,font=\footnotesize,},
        		label style={font=\footnotesize},
        		grid=both
            ]
            \addplot [black] coordinates{
            (4, 0.000) (5, 0.000) (6, 0.000) (7, 0.000) (8, 0.000) (9, 0.000) (10, 0.000) (11, 0.000) (12, 0.000) (13, 0.000) (14, 0.000) (15, 0.000) (16, 0.000) (17, 0.000) (18, 0.000) (19, 0.000) (20, 0.000) (21, 0.000) (22, 0.000) (23, 0.000) (24, 0.000) (25, 0.000) (26, 0.000) (27, 0.000) (28, 0.001) (29, 0.001) (30, 0.001) (31, 0.002) (32, 0.003) (33, 0.004) (34, 0.006) (35, 0.008) (36, 0.010) (37, 0.013) (38, 0.017) (39, 0.022) (40, 0.028) (41, 0.034) (42, 0.043) (43, 0.051) (44, 0.060) (45, 0.075) (46, 0.086) (47, 0.103) (48, 0.120) (49, 0.135) (50, 0.160) (51, 0.182) (52, 0.209) (53, 0.234) (54, 0.264) (55, 0.293) (56, 0.336) (57, 0.374) (58, 0.412) (59, 0.455) (60, 0.502) (61, 0.549) (62, 0.600) (63, 0.657) (64, 0.718) (65, 0.762) (66, 0.844) (67, 0.908) (68, 0.961) (69, 1.053) (70, 1.118) (71, 1.189) (72, 1.248) (73, 1.351) (74, 1.425) (75, 1.517) (76, 1.609) (77, 1.667) (78, 1.801) (79, 1.877) (80, 1.972) (81, 2.072) (82, 2.185) (83, 2.307) (84, 2.421) (85, 2.528) (86, 2.703) (87, 2.818) (88, 2.916) (89, 3.084) (90, 3.151) (91, 3.324) (92, 3.449) (93, 3.542) (94, 3.691) (95, 3.832) (96, 3.956) (97, 4.086) (98, 4.201) (99, 4.388) (100, 4.505)
            }; 
            
        \end{axis}
    \end{tikzpicture}
    \label{figure_var_k_sub3}
    }
    \caption{Expected $\errLinf$ approximation error as a function of the number of targets $k \in [4,100]$, for  $W=32$ and $n=25,50,100$ rules.
    \label{figure_var_k}}
\end{figure}
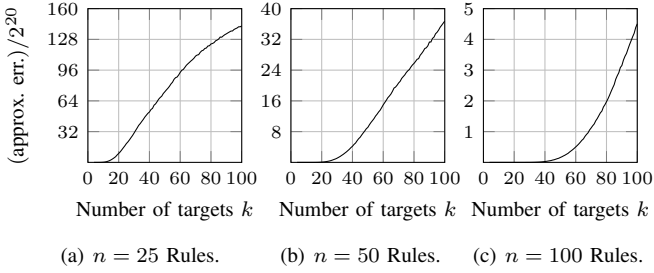

\subsection{Error as a function of Width \texorpdfstring{$W$}{W}}
\label{section_experiments_var_W}
In this section we examine the expected $\errLinf$ approximation error as a function of the address-width $W$. We expect the error to grow exponentially in proportion to $2^W$, because increasing $W$ by $1$ doubles the average weight, and therefore is likely to double the expected error because everything scales-up by a factor of $2$.
Fig.~\ref{figure_var_W} shows the results in logarithmic scale, for fixed values $k=10$ and $n=25,50$. As expected, the slopes are almost $1$: $1.0014$ (for $k=10,n=25$) and $1.0076$ (for $k=10,n=50$). The graphs are almost parallel and $n=50$ is lower, because more rules yield lower error. The slopes are not exactly $1$ since they  also depend  weakly on $n$ and $k$.

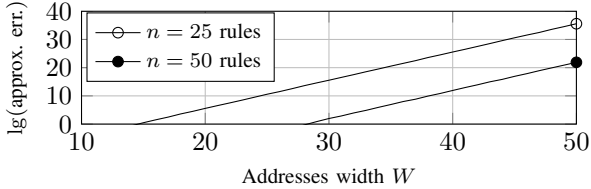
\begin{figure}[t]
	\centering
    \begin{tikzpicture}
        \begin{axis}[
            width=0.45\textwidth,
    		height=0.17\textwidth,
    		xmin=10,
    		xmax=50,
    		xtick={10,20,30,40,50},
    		xlabel=Addresses width $W$,
    		ymin=0,
            ymax = 40,
            ytick={-10,0,10,20,30,40},
            ylabel=$\lg(\text{approx. err.})$,
    		legend style={at={(0.01,0.99)}, anchor=north west,font=\footnotesize,},
    		label style={font=\footnotesize},
    		grid=both
        ]
        
        \addplot [black,mark=o] coordinates{ (50, 35.573) }; \addlegendentry{$n=25$ rules}
        \addplot [black,mark=*] coordinates{ (50, 21.892) }; \addlegendentry{$n=50$ rules}
        
         \addplot [black,mark=.] coordinates{
            (12, -6.644) (13, -2.708) (14, -0.651) (15, 0.485) (16, 1.517) (17, 2.478) (18, 3.588) (19, 4.528) (20, 5.569) (21, 6.561) (22, 7.524) (23, 8.518) (24, 9.547) (25, 10.586) (26, 11.511) (27, 12.536) (28, 13.614) (29, 14.597) (30, 15.550) (31, 16.549) (32, 17.532) (33, 18.561) (34, 19.535) (35, 20.551) (36, 21.527) (37, 22.491) (38, 23.589) (39, 24.543) (40, 25.551) (41, 26.516) (42, 27.525) (43, 28.537) (44, 29.581) (45, 30.592) (46, 31.564) (47, 32.480) (48, 33.577) (49, 34.546) (50, 35.573)
         }; 
     		
        \addplot [black,mark=.] coordinates{
            (25, -6.506) (26, -3.626) (27, -1.680) (28, -0.304) (29, 0.874) (30, 1.980) (31, 2.884) (32, 3.947) (33, 4.872) (34, 5.938) (35, 6.936) (36, 7.994) (37, 8.829) (38, 9.929) (39, 10.863) (40, 11.921) (41, 12.951) (42, 13.921) (43, 14.980) (44, 15.843) (45, 16.925) (46, 17.884) (47, 18.939) (48, 19.925) (49, 20.900) (50, 21.892)
         }; 
     		
        \end{axis}
    \end{tikzpicture}
     \caption{Expected $\errLinf$ approximation error (log-scale) as a function of the address size $W \in [10,50]$, for fixed $k=10$ targets and $n=25,50$ rules. \label{figure_var_W}}
\end{figure}

\subsection{Error vs.\ fixed Rules per Target Ratio}
\label{subsection_experiments_ratio_nk}
In this section we examine the expected $\errLinf$ approximation error as a function of the number of targets $k$, 
where the number of rules is proportional to the number of targets, that is $n=ck$ for some constant $c$.
This setting represents the scenario in which we are willing to allocate a TCAM table of size that is proportional to the number of targets.
We expect the error to decrease as $k$ (and $n$) increases. The reason is that we have freedom to use the additional rules to decrease large weights and reduce error.

Most of the simulations in this paper consider moderate values of $n$ and $k$. In this section we not only consider moderate values (Fig.~\ref{figure_var_ratio_nk}) but also consider the ratio for a very large number of targets (Fig.~\ref{figure_var_ratio_nk_large}), in the order of $1K {-} 10K$, to show that the behaviour of the error is mostly affected by the average number of rules per target rather than some absolute quantity.

Fig.~\ref{figure_var_ratio_nk} shows the data for various ratios $n/k=0.5,\ldots,5.5$ and $W=32$. When $n/k=5$ the error is typically zero, so there was no point to proceed to higher ratios ($n/k=5.5$ emphasizes that). We also plot the error with one rule ($n=1$) to provide a baseline. Although the average weight decreases proportionally to $\frac{2^W}{k}$, a single rule only deals with the maximum weight, and it is not sufficient to provide any significant reduction in the error.

The relation which we see is an inverse-power law, meaning that the error is proportional to $k^{-\alpha}$ for some power $\alpha > 0$. The graph is in log-log scale to emphasize the linearity of the relation $\lg err \approx -\alpha \cdot \lg k + const$.
Although it may look like parallel slopes, it actually steepens slightly as $n/k$ grows, so $\alpha$ depends on $n/k$. Moreover, $const$ also depends on $n/k$. Factoring this into account, the numeric relation we get from this data is close to $\lg err = -(0.16(n/k) + 0.91) \lg k + (34.24-4.92(n/k))$. The constant $34.24$ is most likely in part due to $W$ (which is $32$ in this experiment), since we know the error is exponential in $W$ (Section~\ref{section_experiments_var_W}).
Replacing $34.24$ by $W+ 2.24$ we get:
\begin{equation}
\label{equation_empiric}
	    \mathbb{E}[\errLinf \text{\ error}] = 2^{W - 4.92 \cdot (n/k) + 2.24} \cdot k^{-(0.16 (n/k) + 0.91)}
\end{equation}
The ``magic'' constants in Equation~(\ref{equation_empiric}) probably result from a more intricate dependence of the error on $n$, $k$, and $W$.

The way to interpret Equation~(\ref{equation_empiric}) is as follows: The error grows 
exponentially (base $2$) with the width $W$. For a fixed $k$, it drops exponentially with $n$. The dependence on $k$ is more complex: When $k \le n$ the ratio $n/k$ is rapidly dropping as $k$ grows, causing a quick growth in the error. When $k \ge n$, although the base grows with $k$ and the power is negative, the exponent still drops slowly as $k$ grows ($n/k \in (0,1)$). This produces a moderate error-growth that seems linear in Fig.~\ref{figure_var_k}, until the change in $n/k$ becomes so small that the growth of the base $k$ takes effect and slows the error-growth.

Using Equation~(\ref{equation_empiric}) we can derive the exponential dependence of the error as a function of $n$, per Section~\ref{section_experiments_var_n}: $\lg err = An+B$ where $A(k) = -\frac{4.92 + 0.16 \lg k}{k}$ and $B(k,W) = W+2.24- 0.91 \lg k$. The data in Section~\ref{section_experiments_var_n} for $k=5,10,11,20$ highly agrees with these expressions; $A(k)$ deviates by at most $2.2\%$, and $B(k,32)$ deviates by at most $0.5\%$.

Equation~(\ref{equation_empiric}) can also be used to revisit the graphs in Section~\ref{section_experiments_var_k}. It produces pretty accurate curves for $n=50,100$ (Fig.~\ref{figure_var_k_sub2}-\ref{figure_var_k_sub3}), less than $8\%$ deviation from the computed data. However, for $n=25$ (Fig.~\ref{figure_var_k_sub1}) the 
error predicted by Equation~(\ref{equation_empiric})
deviates from the computed data for $k\ge 2n$ (deviation goes up to $30\%$), though still qualitatively exhibiting the slowdown, when the base $k$ becomes more important than the exponent $n/k$. A visual plot of Equation~(\ref{equation_empiric}) against the graphs from Section~\ref{section_experiments_var_k} is provided in Fig.~\ref{figure_var_k_and_equation_plot}.

Fig.~\ref{figure_var_ratio_nk_large} augments the analysis for large values of $k$, in multiples of $1000$ up to $20000$, for the ratios $n/k=1,2,3,4$. Due to the large values of $k$,  we obtained each data-point by averaging the error for 100 sampled partitions (rather than 1000 as in our other experiments). It is clear that for a very large number of targets, about $4$ rules on average per target are enough for almost no error. We note that the slopes of the trend lines that were discussed in the previous paragraphs for $k \in [4,100]$ are not accurate for larger $k$. For large $k$ the slope is smaller. For instance, for $n/k=2$ the trend-line for large $k$ is $\lg(error) = -1.0635 \lg(k) + 23.434$, compared to $\lg(error) = -1.2259 \lg(k) + 24.285$ when $k \le 100$ is considered. It could be that indeed the slope is not constant and decreases with $k$ or, possibly, the fact that the parameters ($k$, $W$, $n$) of the problem are discrete makes it less suitable to represent this (discrete) function by a (continuous) line.

\begin{figure}[!t]
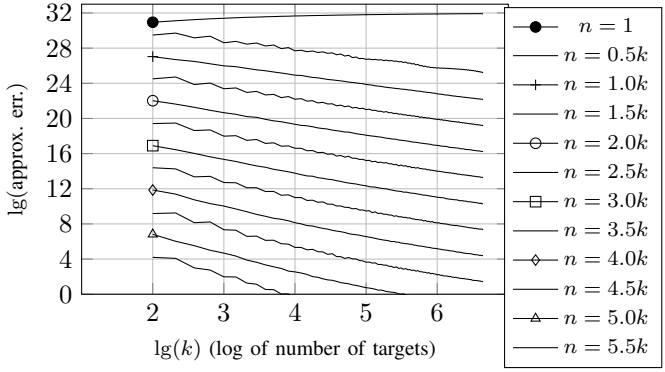

	\centering

    \caption{Expected $\errLinf$ approximation error as a function of the number of targets $k \in [4,100]$, in log-log-scale, for  $W=32$  and $n/k = const$ rules. Half-ratio are more jagged because a new rule is added only on every other increment of $k$.}
    \label{figure_var_ratio_nk}
\end{figure}

\begin{figure}[!t]
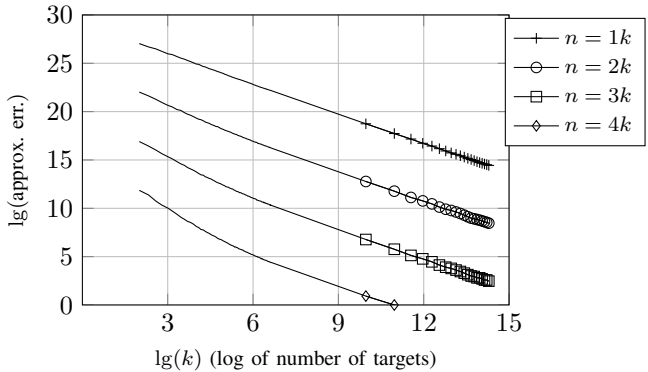

	\centering
    %
    \caption{Expected $\errLinf$ approximation error as a function of the number of targets $k$, for large
     values of $k$, specifically, for $k=1000x$ where $x \in [1,20]$. This graph shows the expected $\errLinf$ approximation error as a function of the number of targets, in log-log-scale, for  $W=32$  and $n/k = const$ rules. For $n/k \ge 4$ there is typically no error for large $k$. (The plot is extended to $\lg k < 7$ according to Fig.~\ref{figure_var_ratio_nk}.)}
    \label{figure_var_ratio_nk_large}
\end{figure}

\begin{figure}[t]
    \subfigure[$n=25$ Rules.] {
        \hspace*{-1.0em}  
        \begin{tikzpicture}
            \begin{axis}[
                width=0.2\textwidth,
        		height=0.2\textwidth,
        		xmin=0,
        		xmax=100,
        		xtick={0,20,40,60,80,100},
        		xlabel=Number of targets $k$,
        		style={font=\scriptsize}, 
        		ymin=0, ymax = 160,
                ytick={32,64,96,128,160},
                ylabel=$(\text{approx.\ err.}) / 2^{20}$,
        		legend style={at={(0.01,0.99)}, anchor=north west,font=\footnotesize,},
        		label style={font=\footnotesize},
        		grid=both
            ]
            \addplot [black] coordinates{
            (4, 0.000) (5, 0.000) (6, 0.001) (7, 0.006) (8, 0.025) (9, 0.075) (10, 0.184) (11, 0.376) (12, 0.694) (13, 1.165) (14, 1.744) (15, 2.568) (16, 3.546) (17, 4.660) (18, 5.872) (19, 7.711) (20, 9.292) (21, 10.917) (22, 12.987) (23, 15.124) (24, 17.194) (25, 19.514) (26, 21.270) (27, 23.370) (28, 25.707) (29, 27.906) (30, 30.483) (31, 33.143) (32, 35.725) (33, 38.471) (34, 40.324) (35, 42.541) (36, 44.682) (37, 47.117) (38, 49.099) (39, 51.049) (40, 52.418) (41, 54.987) (42, 56.188) (43, 59.531) (44, 60.893) (45, 63.211) (46, 65.749) (47, 67.510) (48, 68.895) (49, 71.632) (50, 73.455) (51, 75.359) (52, 77.116) (53, 79.149) (54, 81.623) (55, 83.416) (56, 86.457) (57, 88.171) (58, 88.829) (59, 90.995) (60, 93.638) (61, 95.821) (62, 96.697) (63, 98.801) (64, 101.033) (65, 102.585) (66, 104.044) (67, 105.569) (68, 106.887) (69, 108.518) (70, 110.642) (71, 111.676) (72, 112.102) (73, 114.511) (74, 116.100) (75, 116.739) (76, 118.520) (77, 119.024) (78, 121.080) (79, 121.400) (80, 123.217) (81, 123.723) (82, 124.604) (83, 126.482) (84, 127.274) (85, 128.371) (86, 130.427) (87, 130.289) (88, 131.685) (89, 132.892) (90, 133.769) (91, 134.775) (92, 135.716) (93, 136.642) (94, 136.771) (95, 138.313) (96, 139.145) (97, 139.888) (98, 141.111) (99, 141.090) (100, 142.006)
            }; 
            
            \addplot [black, style=dashed] coordinates{ (4, 7.37745E-07) (5, 4.72442E-05) (6, 0.000753882) (7, 0.00542188) (8, 0.023670763) (9, 0.074061684) (10, 0.183513401) (11, 0.38378457) (12, 0.706824218) (13, 1.18067132) (14, 1.826752096) (15, 2.658680653) (16, 3.68226575) (17, 4.896307865) (18, 6.293809237) (19, 7.8633163) (20, 9.590213977) (21, 11.45787227) (22, 13.4486026) (23, 15.54441738) (24, 17.72760701) (25, 19.98115788) (26, 22.28903842) (27, 24.63637957) (28, 27.00957302) (29, 29.39630767) (30, 31.78556078) (31, 34.1675572) (32, 36.53370717) (33, 38.87653064) (34, 41.1895742) (35, 43.467325) (36, 45.70512487) (37, 47.89908689) (38, 50.04601582) (39, 52.14333332) (40, 54.1890085) (41, 56.18149403) (42, 58.11966772) (43, 60.00277953) (44, 61.8304037) (45, 63.60239573) (46, 65.3188538) (47, 66.98008445) (48, 68.58657195) (49, 70.13895118) (50, 71.63798367) (51, 73.08453635) (52, 74.479563) (53, 75.82408782) (54, 77.11919114) (55, 78.36599687) (56, 79.56566162) (57, 80.71936518) (58, 81.82830236) (59, 82.89367586) (60, 83.9166902) (61, 84.8985465) (62, 85.84043802) (63, 86.74354638) (64, 87.60903846) (65, 88.43806365) (66, 89.23175175) (67, 89.99121114) (68, 90.71752735) (69, 91.41176194) (70, 92.07495165) (71, 92.70810772) (72, 93.31221553) (73, 93.88823426) (74, 94.43709684) (75, 94.95970993) (76, 95.45695409) (77, 95.92968391) (78, 96.37872841) (79, 96.8048913) (80, 97.20895148) (81, 97.59166341) (82, 97.95375768) (83, 98.2959415) (84, 98.61889921) (85, 98.9232929) (86, 99.20976295) (87, 99.47892861) (88, 99.73138856) (89, 99.96772156) (90, 100.188487) (91, 100.3942253) (92, 100.585459) (93, 100.7626926) (94, 100.9264137) (95, 101.0770933) (96, 101.2151864) (97, 101.3411323) (98, 101.4553556) (99, 101.5582663) (100, 101.6502602) }; 
        \end{axis}
    \end{tikzpicture} 
    \label{figure_var_k_and_equation_plot_sub1}
    }
    \hspace*{-2.1em}  
    ~
    \subfigure[$n=50$ Rules.] {
        \begin{tikzpicture}
            \begin{axis}[
                width=0.2\textwidth,
        		height=0.2\textwidth,
        		xmin=0,
        		xmax=100,
        		xtick={0,20,40,60,80,100},
        		xlabel=Number of targets $k$,
        		style={font=\scriptsize}, 
        		ymin=0, ymax = 40,
                ytick={8,16,24,32,40},
        		label style={font=\footnotesize},
        		grid=both
            ]
            \addplot [black] coordinates{
            (4, 0.000) (5, 0.000) (6, 0.000) (7, 0.000) (8, 0.000) (9, 0.000) (10, 0.000) (11, 0.000) (12, 0.000) (13, 0.001) (14, 0.002) (15, 0.004) (16, 0.008) (17, 0.015) (18, 0.026) (19, 0.042) (20, 0.068) (21, 0.100) (22, 0.148) (23, 0.207) (24, 0.274) (25, 0.354) (26, 0.463) (27, 0.602) (28, 0.757) (29, 0.899) (30, 1.092) (31, 1.314) (32, 1.543) (33, 1.821) (34, 2.086) (35, 2.398) (36, 2.703) (37, 3.053) (38, 3.416) (39, 3.842) (40, 4.184) (41, 4.669) (42, 5.025) (43, 5.621) (44, 6.153) (45, 6.535) (46, 7.120) (47, 7.664) (48, 8.174) (49, 8.748) (50, 9.107) (51, 9.780) (52, 10.364) (53, 10.862) (54, 11.383) (55, 12.044) (56, 12.532) (57, 13.113) (58, 13.578) (59, 14.261) (60, 14.972) (61, 15.560) (62, 16.163) (63, 16.827) (64, 17.266) (65, 17.778) (66, 18.488) (67, 19.328) (68, 19.590) (69, 20.067) (70, 20.668) (71, 21.086) (72, 21.849) (73, 22.314) (74, 22.756) (75, 23.387) (76, 23.828) (77, 24.175) (78, 24.821) (79, 25.354) (80, 25.750) (81, 26.260) (82, 26.818) (83, 26.990) (84, 27.863) (85, 28.348) (86, 28.919) (87, 29.619) (88, 29.814) (89, 30.329) (90, 31.311) (91, 31.317) (92, 32.248) (93, 32.597) (94, 33.454) (95, 33.851) (96, 34.428) (97, 35.158) (98, 35.679) (99, 36.231) (100, 36.777)
            }; 
            
            \addplot [black, style=dashed] coordinates{ (4, 1.00076E-16) (5, 5.03202E-13) (6, 1.51351E-10) (7, 9.01226E-09) (8, 1.94059E-07) (9, 2.11556E-06) (10, 1.43012E-05) (11, 6.82377E-05) (12, 0.000250606) (13, 0.000752285) (14, 0.001927011) (15, 0.004347391) (16, 0.008845672) (17, 0.016530634) (18, 0.028777666) (19, 0.047194337) (20, 0.073567242) (21, 0.109797452) (22, 0.157831768) (23, 0.219595915) (24, 0.296934173) (25, 0.391558343) (26, 0.505007468) (27, 0.638618585) (28, 0.793507947) (29, 0.970561637) (30, 1.170434201) (31, 1.39355383) (32, 1.640132649) (33, 1.91018077) (34, 2.203522946) (35, 2.51981681) (36, 2.85857187) (37, 3.21916862) (38, 3.600877234) (39, 4.002875487) (40, 4.424265625) (41, 4.864090012) (42, 5.321345457) (43, 5.794996173) (44, 6.283985379) (45, 6.787245564) (46, 7.303707488) (47, 7.832307983) (48, 8.371996644) (49, 8.921741498) (50, 9.480533735) (51, 10.04739161) (52, 10.62136359) (53, 11.2015308) (54, 11.78700894) (55, 12.37694963) (56, 12.97054131) (57, 13.56700977) (58, 14.16561833) (59, 14.76566773) (60, 15.3664958) (61, 15.96747688) (62, 16.56802114) (63, 17.16757372) (64, 17.76561375) (65, 18.36165332) (66, 18.95523638) (67, 19.54593758) (68, 20.13336106) (69, 20.71713929) (70, 21.2969318) (71, 21.87242403) (72, 22.44332608) (73, 23.00937156) (74, 23.57031643) (75, 24.12593784) (76, 24.67603305) (77, 25.22041834) (78, 25.75892801) (79, 26.29141334) (80, 26.81774163) (81, 27.3377953) (82, 27.851471) (83, 28.35867873) (84, 28.85934106) (85, 29.35339233) (86, 29.84077794) (87, 30.32145361) (88, 30.79538474) (89, 31.26254577) (90, 31.72291957) (91, 32.17649686) (92, 32.62327567) (93, 33.06326087) (94, 33.4964636) (95, 33.92290091) (96, 34.34259523) (97, 34.75557404) (98, 35.16186944) (99, 35.5615178) (100, 35.95455943) }; 
        \end{axis}
    \end{tikzpicture}
    \label{figure_var_k_and_equation_plot_sub2}
    }
    \hspace*{-2.1em}  
    ~
    \subfigure[$n=100$ Rules.] {
        \begin{tikzpicture}
            \begin{axis}[
                width=0.2\textwidth,
        		height=0.2\textwidth,
        		xmin=0,
        		xmax=100,
        		xtick={0,20,40,60,80,100},
        		xlabel=Number of targets $k$,
        		style={font=\scriptsize}, 
        		ymin=0, ymax = 5,
                ytick={1,2,3,4,5},
        		legend style={at={(0.01,0.99)}, anchor=north west,font=\footnotesize,},
        		label style={font=\footnotesize},
        		grid=both
            ]
            \addplot [black] coordinates{
            (4, 0.000) (5, 0.000) (6, 0.000) (7, 0.000) (8, 0.000) (9, 0.000) (10, 0.000) (11, 0.000) (12, 0.000) (13, 0.000) (14, 0.000) (15, 0.000) (16, 0.000) (17, 0.000) (18, 0.000) (19, 0.000) (20, 0.000) (21, 0.000) (22, 0.000) (23, 0.000) (24, 0.000) (25, 0.000) (26, 0.000) (27, 0.000) (28, 0.001) (29, 0.001) (30, 0.001) (31, 0.002) (32, 0.003) (33, 0.004) (34, 0.006) (35, 0.008) (36, 0.010) (37, 0.013) (38, 0.017) (39, 0.022) (40, 0.028) (41, 0.034) (42, 0.043) (43, 0.051) (44, 0.060) (45, 0.075) (46, 0.086) (47, 0.103) (48, 0.120) (49, 0.135) (50, 0.160) (51, 0.182) (52, 0.209) (53, 0.234) (54, 0.264) (55, 0.293) (56, 0.336) (57, 0.374) (58, 0.412) (59, 0.455) (60, 0.502) (61, 0.549) (62, 0.600) (63, 0.657) (64, 0.718) (65, 0.762) (66, 0.844) (67, 0.908) (68, 0.961) (69, 1.053) (70, 1.118) (71, 1.189) (72, 1.248) (73, 1.351) (74, 1.425) (75, 1.517) (76, 1.609) (77, 1.667) (78, 1.801) (79, 1.877) (80, 1.972) (81, 2.072) (82, 2.185) (83, 2.307) (84, 2.421) (85, 2.528) (86, 2.703) (87, 2.818) (88, 2.916) (89, 3.084) (90, 3.151) (91, 3.324) (92, 3.449) (93, 3.542) (94, 3.691) (95, 3.832) (96, 3.956) (97, 4.086) (98, 4.201) (99, 4.388) (100, 4.505)
            }; 
            
            \addplot [black, style=dashed] coordinates{ (4, 1.84154E-36) (5, 5.70862E-29) (6, 6.10026E-24) (7, 2.49001E-20) (8, 1.3043E-17) (9, 1.72618E-15) (10, 8.68525E-14) (11, 2.15724E-12) (12, 3.15031E-11) (13, 3.05414E-10) (14, 2.14434E-09) (15, 1.1624E-08) (16, 5.10461E-08) (17, 1.88422E-07) (18, 6.01644E-07) (19, 1.70004E-06) (20, 4.32909E-06) (21, 1.00825E-05) (22, 2.17385E-05) (23, 4.38252E-05) (24, 8.33068E-05) (25, 0.000150365) (26, 0.000259245) (27, 0.000429112) (28, 0.000684886) (29, 0.001057995) (30, 0.001587016) (31, 0.00231816) (32, 0.003305593) (33, 0.004611566) (34, 0.006306351) (35, 0.008468004) (36, 0.01118195) (37, 0.014540412) (38, 0.018641718) (39, 0.023589496) (40, 0.029491793) (41, 0.036460132) (42, 0.044608551) (43, 0.054052617) (44, 0.064908464) (45, 0.077291856) (46, 0.091317288) (47, 0.107097144) (48, 0.12474092) (49, 0.144354515) (50, 0.166039592) (51, 0.189893018) (52, 0.216006384) (53, 0.244465593) (54, 0.275350524) (55, 0.308734772) (56, 0.344685447) (57, 0.383263038) (58, 0.42452134) (59, 0.468507424) (60, 0.515261665) (61, 0.564817804) (62, 0.61720306) (63, 0.67243826) (64, 0.730538013) (65, 0.791510899) (66, 0.855359683) (67, 0.922081543) (68, 0.991668316) (69, 1.06410675) (70, 1.139378766) (71, 1.217461725) (72, 1.298328696) (73, 1.381948723) (74, 1.468287101) (75, 1.557305634) (76, 1.6489629) (77, 1.743214507) (78, 1.840013342) (79, 1.939309815) (80, 2.041052094) (81, 2.145186332) (82, 2.251656883) (83, 2.360406514) (84, 2.471376604) (85, 2.584507335) (86, 2.69973787) (87, 2.817006532) (88, 2.936250957) (89, 3.057408256) (90, 3.180415155) (91, 3.305208132) (92, 3.431723542) (93, 3.559897738) (94, 3.689667182) (95, 3.820968546) (96, 3.953738807) (97, 4.087915338) (98, 4.223435985) (99, 4.360239146) (100, 4.498263838) }; 
        \end{axis}
    \end{tikzpicture}
    \label{figure_var_k_and_equation_plot_sub3}
    }
    \caption{Expected $\errLinf$ approximation error as a function of the number of targets $k \in [4,100]$, for  $W=32$ and $n=25,50,100$ rules. The dashed line plots Equation~(\ref{equation_empiric}) of Section~\ref{subsection_experiments_ratio_nk} for comparison.
    \label{figure_var_k_and_equation_plot}}
\end{figure}
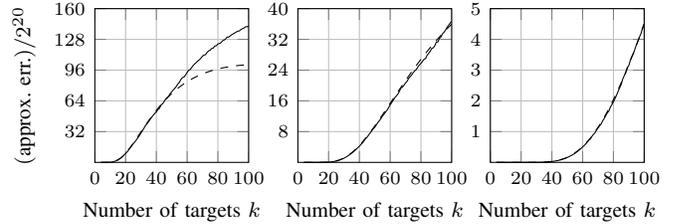

\subsection{\texorpdfstring{$\errLinfPos$}{L-infinity+} error vs. \texorpdfstring{$\errLinf$}{L-infinity} error}
\label{subsection_experiments_onesided_vs_twosided}

In this section we compare the expectation of the $\errLinfPos$ error which we considered in Section~\ref{section_one_sided_metric} to
the expectation of the $\errLinf$ error, as a function of the number of rules $n$.
We note that there are partitions for which this difference is large. For example, let $c > k$ be a power of $2$, and define $P = [k-1, c-1, \ldots, c-1, 2^W - (k-2)c - 1]$, and 
 let $P' = [0, c, \ldots, c, 2^W - (k-2)c]$.  One can verify that the $\errLinf$ error of $P'$ is $k-1$, which is 
 the smallest possible with $k-1$ rules
since at least one weight of $P'$ must be $0$. The $\errLinfPos$ error of $P'$ is only $1$, which is also the smallest possible with $k-1$ rules. The example can be scaled-up by a factor of $2^d$ for any integer $d>0$, though the ratio of the errors will not change.

For a random partition the expected difference between these errors is much less noticeable. Fig.~\ref{figure_oneisded_two_sided_ratio} shows the ratio between the expected $\errLinfPos$ error and the expected $\errLinf$ error for fixed $k=10$, and $W=32$.
The $\errLinfPos$ error is, of course, never larger than the $\errLinf$ error, and we see that in fact the ratio is around $0.89$ for most values of $n$. It starts at $1.0$ since with one rule there is no difference:
The overload of the server that gets the rule is equal to the sum of the underloads of all the other servers.
As long as $n \le k$ it is still likely
for these errors to be equal, until the overloads and underloads spread over enough different weights. Even with two rules there can be a difference, for example $P=[2,3,3]$ is best approximated by $P'=[0,4,4]$ for both errors
(these partitions are obtained by substituting $k=3$, and $c=4$ in the general example at the beginning of this section), which yields $\errLinfPos$ error of $1$ compared to $\errLinf$ error of $2$. However, for $n > k$ (See Fig.~\ref{figure_oneisded_two_sided_ratio} for  $n > k=10$), there is enough freedom for the difference between the errors to show up. The ratio climbs back to  $1.0$ for $n \ge 55$ due to increasing probability to achieve exact representation of the partitions. We conclude that on average, when there are not too many or too few rules ($n$) compared to the number of targets ($k$), the $\errLinfPos$ error is approximately $11\%$ smaller than the $\errLinf$ error.

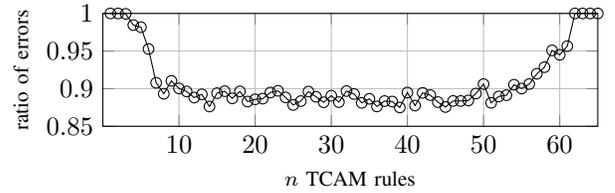
\begin{figure}[!t]
	\centering
    \begin{tikzpicture}
        \begin{axis}[
            width=0.45\textwidth,
    		height=0.17\textwidth,
    		xmin=0,
    		xmax=65,
    		xtick={10,20,30,40,50,60},
    		xlabel=$n$ TCAM rules,
    		ymin=0.85,
            ymax =1.0,
            ytick={0.85,0.9,0.95,1},
            ylabel= ratio of errors, 
    		legend style={at={(0.99,0.99)}, anchor=north east,font=\footnotesize,},
    		label style={font=\footnotesize},
    		grid=both
        ]
          
         \addplot [black,mark=o] coordinates{
            (1, 1.000000) (2, 1.000000) (3, 0.999458) (4, 0.984523) (5, 0.981701) (6, 0.952732) (7, 0.907828) (8, 0.893221) (9, 0.910348) (10, 0.900306) (11, 0.896299) (12, 0.888149) (13, 0.892425) (14, 0.876268) (15, 0.894018) (16, 0.896648) (17, 0.887059) (18, 0.896399) (19, 0.882834) (20, 0.885826) (21, 0.886754) (22, 0.894880) (23, 0.897189) (24, 0.888391) (25, 0.878349) (26, 0.883497) (27, 0.896134) (28, 0.889211) (29, 0.881904) (30, 0.890634) (31, 0.882112) (32, 0.896815) (33, 0.892864) (34, 0.881165) (35, 0.886521) (36, 0.876369) (37, 0.883481) (38, 0.882866) (39, 0.874994) (40, 0.894832) (41, 0.877543) (42, 0.894413) (43, 0.891436) (44, 0.882172) (45, 0.875610) (46, 0.883655) (47, 0.883596) (48, 0.884305) (49, 0.893497) (50, 0.906049) (51, 0.881170) (52, 0.889650) (53, 0.891278) (54, 0.905243) (55, 0.900136) (56, 0.906130) (57, 0.919802) (58, 0.928571) (59, 0.950867) (60, 0.945000) (61, 0.956522) (62, 1.000000) (63, 1.000000) (64, 1.000000) (65, 1.000000)
         };
     		
        \end{axis}
    \end{tikzpicture}
    \caption{Ratio between the expected $\errLinfPos$ error and the expected $\errLinf$ error,  as a function of $n \in [1,65]$  rules, for  $W=32$ and $k=10$ targets. When $n \le k=10$ or $n \ge 60$ the errors are close due to inability to approximate well or the ability to achieve zero-error, respectively. In between, the ratio is around $0.89$. \label{figure_oneisded_two_sided_ratio}}
\end{figure}

\subsection{Approximating with a prefix of Niagara}
\label{subsection_experiments_niagara_ratio}
In this subsection we evaluate the approximation error of following simple heuristic:
Compute a sequence of TCAM rules that induce the target partition exactly and use the last $n$  rules (with most don't-cares) in this sequence. 
\cite{BitMatcher} explains and shows experimentally that applying this heuristic to the sequence produced by Niagara gives better results than when applied to the sequence of Bit Matcher. For this reason, in this paper we only apply truncation to the sequence of Niagara.

Prefixes of the Niagara sequence may be sub-optimal. Fig.~\ref{figure_tcam_trie_all} shows the sub-optimality of a prefix of such a sequence. This example can be generalized by considering the partition $[2^{m+1},1,\ldots,1]$ with $k=2^{m+1}+1$ parts.
For this partition and $n=2$ rules, the ratio between the $\errLinf$ error of the truncated Niagara sequence and the optimal error is $2-\frac{1}{2^m}$. Namely, the error obtained by the truncated Niagara can be close to twice the optimal error. We conjecture that the error of such a truncation is at most twice the optimal error. Specifically, our conjecture is as follows.

\begin{conjecture}
\label{conjecture_approx2_error}
Let $P$ be a partition and let $P'$ be its optimal approximation in $\errLinf$, with $n$ TCAM rules. Choose $h$ such that $2^h \le |P-P'|_\infty < 2^{h+1}$. Then truncating a Niagara TCAM to $n$ rules induces a partition $P''$ such that $|P-P''|_\infty < 2^{h+1}$.
\end{conjecture}

Assuming this conjecture, it is not surprising that Niagara performs extremely well compared to the optimal algorithm, on a random partition.
For instance, Fig.~\ref{figure_ratio_1sided_2sided} shows that the expectation of the ratio between the error of a truncated Niagara sequence and the optimal error as a function of $n$, when $W=16$, and $k=16$, for both $\errLinfPos$ and $\errLinf$, is very close to $1$. The plotted values were averaged over $10,000$ random ordered-partitions. When $n \ge 16$, i.e.\ we have at least 1 rule per target, the ratios do not rise above $1.00015$ (less than $0.015\%$ extra error).\footnote{For $n < 16$, the approximation error is much more sensitive to the choices of Niagara, and as a result the expectation can get as high as $1.015$ ($\times 100$ the scale of $n \ge 16$). For clarity of the graphs in Fig.~\ref{figure_ratio_1sided_2sided} and Fig.~\ref{figure_ratio_1sided_2sided_w32} we limited the $y$-axis such that values of $n<16$ might not be presented.}
In the case of $\errLinfPosRel$, the approximation is not as good. The expectation of the ratio of errors can get up to $\approx 21$, as shown in Fig.~\ref{figure_ratio_relative}. This ratio is worse in $\errLinfPosRel$ than in $\errLinf$ and $\errLinfPos$ since Niagara reduces the sum of the deviations, regardless of the identity of the deviating target. Therefore, it may not differentiate between having some error in a target with small weight while it can have instead the same error in a target with larger weight. We note that while this leads to the conclusion that when approximating the partition for $\errLinfPosRel$ the Niagara-truncation heuristic is much worse, this should be taken with a grain of salt because the relative errors are quite small, even if their ratios are large.

Fig.~\ref{figure_ratio_1sided_2sided_w32} and Fig.~\ref{figure_ratio_relative_w32} shows the results for $W=32$ and $k=16$. 
The behaviour is similar.
When $16 \le n$, the ratios of $\errLinf$ and $\errLinfPos$ do not rise above $1.0001$ (less than $0.01\%$ difference), and the ratios of $\errLinfPosRel$ form a curve similar to Fig.~\ref{figure_ratio_relative}, but more noisy. The curve peaks at about $n=70$ with a ratio of $60$.
For $n=64$ and $n=69$ (excluded from the plot) we got much higher ratios of $119.1$ and $779.6$, respectively.

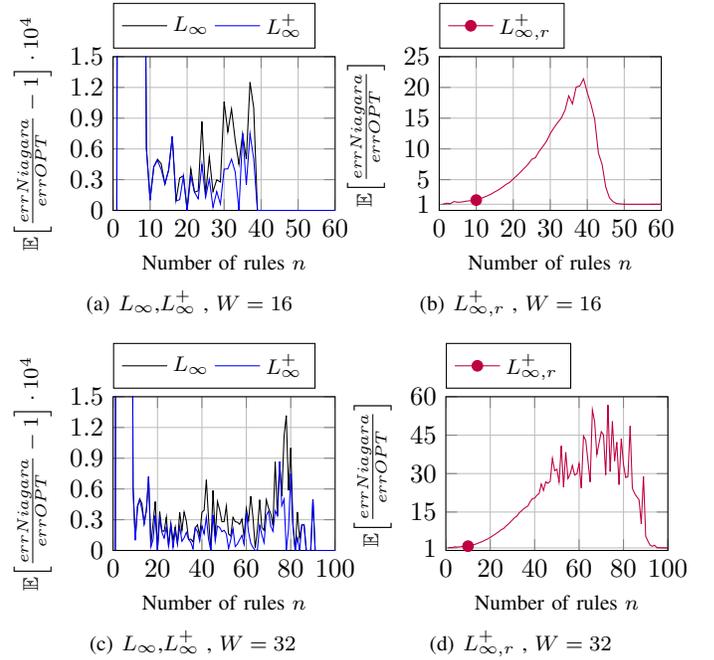
\begin{figure}[t]
    \subfigure[$\errLinf$,$\errLinfPos$ , $W=16$] {
        \hspace*{-1.0em}  
        \begin{tikzpicture}
            \begin{axis}[
                width=0.25\textwidth,
        		height=0.2\textwidth,
        		xmin=0, xmax=60,
        		xtick={0,10,20,30,40,50,60},
        		xlabel=Number of rules $n$,
        		ymin=0, ymax = 1.5,
                ytick={0,0.3,0.6,0.9,1.2,1.5},
                ylabel={$\mathbb{E} \Big [\frac{errNiagara}{errOPT}  - 1 \Big ] \cdot {10}^4$},
        		legend style={at={(0.00,1.35)}, anchor=north west,font=\footnotesize,legend columns=-1,},
        		label style={font=\footnotesize},
        		grid=both
            ]
            \addplot [black] coordinates{ (32,0.989379085) }; \addlegendentry{$\errLinf$}
            \addplot [blue] coordinates{ (32,0.5) }; \addlegendentry{$\errLinfPos$} 
            
            \addplot [black] coordinates{ (1,0) (2,16.81446118) (3,2.28862606) (4,154.3576891) (5,56.80217883) (6,23.26155748) (7,17.82473539) (8,6.738601267) (9,0.607719695) (10,0.10081556) (11,0.426579191) (12,0.499660696) (13,0.455474639) (14,0.256931535) (15,0.392360121) (16,0.721879735) (17,0.098984826) (18,0.317075452) (19,0.337564276) (20,0.007042254) (21,0.404968281) (22,0.177504837) (23,0.191018836) (24,0.866314733) (25,0.159346847) (26,0.522330693) (27,0.184659091) (28,0.298411914) (29,0.27662656) (30,1.059659091) (31,0.766414141) (32,0.989379085) (33,0.683823529) (34,0.45) (35,0.75) (36,0.5) (37,1.25) (38,1) (39,0) (40,0) (41,0) (42,0) (43,0) (44,0) (45,0) (46,0) (47,0) (48,0) (49,0) (50,0) (51,0) (52,0) (53,0) (54,0) (55,0) (56,0) (57,0) (58,0) (59,0) (60,0) (61,0) (62,0) (63,0) (64,0) (65,0) (66,0) (67,0) (68,0) (69,0) (70,0) (71,0) (72,0) (73,0) (74,0) (75,0) (76,0) (77,0) (78,0) (79,0) (80,0) (81,0) (82,0) (83,0) (84,0) (85,0) (86,0) (87,0) (88,0) (89,0) (90,0) (91,0) (92,0) (93,0) (94,0) (95,0) (96,0) (97,0) (98,0) (99,0) (100,0) };
            
            \addplot [blue] coordinates{ (1,0) (2,16.81446118) (3,2.28862606) (4,154.3576891) (5,56.80217883) (6,23.26155748) (7,17.82473539) (8,6.738601267) (9,0.607719695) (10,0.10081556) (11,0.426579191) (12,0.487002468) (13,0.399008337) (14,0.256931535) (15,0.392360121) (16,0.721879735) (17,0.092470168) (18,0.105130156) (19,0.337564276) (20,0) (21,0.332200423) (22,0.177504837) (23,0.11375) (24,0.447280642) (25,0.128096847) (26,0.304979947) (27,0.03030303) (28,0.183712121) (29,0.0625) (30,0.403409091) (31,0.40530303) (32,0.5) (33,0.375) (34,0) (35,0.75) (36,0.25) (37,0.75) (38,0.5) (39,0) (40,0) (41,0) (42,0) (43,0) (44,0) (45,0) (46,0) (47,0) (48,0) (49,0) (50,0) (51,0) (52,0) (53,0) (54,0) (55,0) (56,0) (57,0) (58,0) (59,0) (60,0) (61,0) (62,0) (63,0) (64,0) (65,0) (66,0) (67,0) (68,0) (69,0) (70,0) (71,0) (72,0) (73,0) (74,0) (75,0) (76,0) (77,0) (78,0) (79,0) (80,0) (81,0) (82,0) (83,0) (84,0) (85,0) (86,0) (87,0) (88,0) (89,0) (90,0) (91,0) (92,0) (93,0) (94,0) (95,0) (96,0) (97,0) (98,0) (99,0) (100,0) }; 
            
        \end{axis}
    \end{tikzpicture} 
    \label{figure_ratio_1sided_2sided}
    }
    \hspace*{-2.1em}  
    ~
    \subfigure[$\errLinfPosRel$ , $W=16$] {
        \begin{tikzpicture}
            \begin{axis}[
                width=0.25\textwidth,
        		height=0.2\textwidth,
        		xmin=0, xmax=60,
        		xtick={0,10,20,30,40,50,60},
        		xlabel=Number of rules $n$,
        		ymin=0, ymax = 25,
                ytick={1,5,10,15,20,25},
                ylabel={$\mathbb{E} \Big [\frac{errNiagara}{errOPT} \Big ]$},
        		legend style={at={(0.00,1.35)}, anchor=north west,font=\footnotesize,legend columns=-1,},
        		label style={font=\footnotesize},
        		grid=both
            ]
            \addplot [purple,mark=*] coordinates{ (10,1.680585683) }; \addlegendentry{$\errLinfPosRel$}
            \addplot [purple] coordinates{ (1,1) (2,1.175133409) (3,1.09787264) (4,1.445871405) (5,1.323695651) (6,1.364065373) (7,1.479845948) (8,1.556778501) (9,1.578111293) (10,1.680585683) (11,1.882460916) (12,2.119955358) (13,2.369031145) (14,2.649321035) (15,3.013036285) (16,3.335528792) (17,3.745033534) (18,4.212176787) (19,4.544360526) (20,5.11561828) (21,5.670586241) (22,6.255260377) (23,6.879805921) (24,7.497152222) (25,8.311632908) (26,8.581045961) (27,9.613406763) (28,10.40185351) (29,11.24131202) (30,12.50416239) (31,13.43521012) (32,14.27429981) (33,15.02846765) (34,16.32102287) (35,18.62867119) (36,17.32594048) (37,20.05695193) (38,20.248786) (39,21.40273218) (40,19.1568527) (41,17.26511767) (42,14.83354594) (43,9.350585875) (44,7.499807427) (45,3.779141556) (46,2.045280685) (47,1.399050802) (48,1.095376756) (49,1.059314038) (50,1.00393177) (51,1.000785228) (52,1) (53,1) (54,1) (55,1) (56,1) (57,1) (58,1) (59,1) (60,1) (61,1) (62,1) (63,1) (64,1) (65,1) (66,1) (67,1) (68,1) (69,1) (70,1) (71,1) (72,1) (73,1) (74,1) (75,1) (76,1) (77,1) (78,1) (79,1) (80,1) (81,1) (82,1) (83,1) (84,1) (85,1) (86,1) (87,1) (88,1) (89,1) (90,1) (91,1) (92,1) (93,1) (94,1) (95,1) (96,1) (97,1) (98,1) (99,1) (100,1) }; 
        \end{axis}
    \end{tikzpicture}
    \label{figure_ratio_relative}
    }
    
        \subfigure[$\errLinf$,$\errLinfPos$ , $W=32$] {
        \hspace*{-1.0em}  
        \begin{tikzpicture}
            \begin{axis}[
                width=0.25\textwidth,
        		height=0.2\textwidth,
        		xmin=0, xmax=100,
        		xtick={0,20,40,60,80,100},
        		xlabel=Number of rules $n$,
        		ymin=0, ymax = 1.5,
                ytick={0,0.3,0.6,0.9,1.2,1.5},
                ylabel={$\mathbb{E} \Big [\frac{errNiagara}{errOPT}  - 1 \Big ] \cdot {10}^4$},
        		legend style={at={(0.00,1.35)}, anchor=north west,font=\footnotesize,legend columns=-1,},
        		label style={font=\footnotesize},
        		grid=both
            ]
            \addplot [black] coordinates{ (78, 1.316176471) }; \addlegendentry{$\errLinf$}
            \addplot [blue] coordinates{ (78, 0) }; \addlegendentry{$\errLinfPos$} 
            
            \addplot [black] coordinates{ (1, 0) (2, 16.83024335) (3, 2.291326631) (4, 154.2957953) (5, 56.89059852) (6, 23.09783296) (7, 17.87218466) (8, 6.673676571) (9, 0.606859014) (10, 0.100867171) (11, 0.430582273) (12, 0.503625708) (13, 0.457641601) (14, 0.261931168) (15, 0.391989152) (16, 0.721288104) (17, 0.046039738) (18, 0.261315932) (19, 0.476529595) (20, 0.010125895) (21, 0.37767638) (22, 0.191299246) (23, 0.186775588) (24, 0.324713227) (25, 0.042712734) (26, 0.265800336) (27, 0.101370739) (28, 0.228617828) (29, 0.066929484) (30, 0.372287525) (31, 0.19551523) (32, 0.374904443) (33, 0.287694659) (34, 0.10818096) (35, 0.111055101) (36, 0.083967211) (37, 0.220137624) (38, 0.250354849) (39, 0.140993878) (40, 0.371349626) (41, 0.393106034) (42, 0.688427429) (43, 0.308995225) (44, 0) (45, 0.584679065) (46, 0.141219549) (47, 0.47826228) (48, 0.39531394) (49, 0.28077723) (50, 0.281639918) (51, 0.432719338) (52, 0.161429384) (53, 0.139755223) (54, 0.322998626) (55, 0.280536581) (56, 0.275829303) (57, 0.255733625) (58, 0.304342158) (59, 0.134745121) (60, 0.408936866) (61, 0.28894911) (62, 0.528445391) (63, 0.321366327) (64, 0.026315789) (65, 0.198339495) (66, 0.49456345) (67, 0.224592075) (68, 0.19693498) (69, 0.493137828) (70, 0.413074406) (71, 0.358400178) (72, 0.247214795) (73, 0.863450683) (74, 0.338235294) (75, 0.867647059) (76, 0.53125) (77, 1.117647059) (78, 1.316176471) (79, 0.586111111) (80, 1) (81, 0.25) (82, 0) (83, 0.375) (84, 0.125) (85, 0.25) (86, 0.25) (87, 0) (88, 0) (89, 0) (90, 0.5) (91, 0) (92, 0) (93, 0) (94, 0) (95, 0) (96, 0) (97, 0) (98, 0) (99, 0) (100, 0) };
            
            \addplot [blue] coordinates{ (1, 0) (2, 16.83024335) (3, 2.291326631) (4, 154.2957953) (5, 56.89059852) (6, 23.09783296) (7, 17.87218466) (8, 6.673676571) (9, 0.606859014) (10, 0.100867171) (11, 0.430582273) (12, 0.491105423) (13, 0.399147414) (14, 0.261931168) (15, 0.391989152) (16, 0.721288104) (17, 0.034595424) (18, 0.102460129) (19, 0.341768665) (20, 0) (21, 0.312768164) (22, 0.175794963) (23, 0.121188593) (24, 0.227309297) (25, 0.042241358) (26, 0.242263316) (27, 0.029459733) (28, 0.19195041) (29, 0.066929484) (30, 0.234005913) (31, 0.089559593) (32, 0.123328281) (33, 0.175576473) (34, 0.10818096) (35, 0.080047162) (36, 0.017351244) (37, 0.220137624) (38, 0.185049954) (39, 0.091010828) (40, 0.238575321) (41, 0.162646127) (42, 0.324723826) (43, 0.109322475) (44, 0) (45, 0.346128078) (46, 0.091794262) (47, 0.237063209) (48, 0.224460696) (49, 0.183514968) (50, 0.196115426) (51, 0.162629068) (52, 0.019055845) (53, 0.106331555) (54, 0.211260791) (55, 0.038017932) (56, 0.062264151) (57, 0.106497953) (58, 0.146041) (59, 0.059273423) (60, 0.372384444) (61, 0.182485231) (62, 0.071263728) (63, 0.02962963) (64, 0) (65, 0.005042017) (66, 0.245225701) (67, 0.151515152) (68, 0.075722859) (69, 0.340384615) (70, 0.382771376) (71, 0.235294118) (72, 0.1875) (73, 0.388368984) (74, 0.338235294) (75, 0.867647059) (76, 0.4375) (77, 0.5) (78, 0) (79, 0.211111111) (80, 0.75) (81, 0.25) (82, 0) (83, 0.125) (84, 0) (85, 0.25) (86, 0.25) (87, 0) (88, 0) (89, 0) (90, 0.5) (91, 0) (92, 0) (93, 0) (94, 0) (95, 0) (96, 0) (97, 0) (98, 0) (99, 0) (100, 0) }; 
            
        \end{axis}
    \end{tikzpicture} 
    \label{figure_ratio_1sided_2sided_w32}
    }
    \hspace*{-2.1em}  
    ~
    \subfigure[$\errLinfPosRel$ , $W=32$] {
        \begin{tikzpicture}
            \begin{axis}[
                width=0.25\textwidth,
        		height=0.2\textwidth,
        		xmin=0, xmax=100,
        		xtick={0,20,40,60,80,100},
        		xlabel=Number of rules $n$,
        		ymin=0, ymax = 60,
                ytick={1,15,30,45,60},
                ylabel={$\mathbb{E} \Big [\frac{errNiagara}{errOPT} \Big ]$},
        		legend style={at={(0.00,1.35)}, anchor=north west,font=\footnotesize,legend columns=-1,},
        		label style={font=\footnotesize},
        		grid=both
            ]
            \addplot [purple,mark=*] coordinates{ (10,1.680585683) }; \addlegendentry{$\errLinfPosRel$}
            \addplot [purple] coordinates{ (1, 1) (2, 1.175164476) (3, 1.097899161) (4, 1.445881236) (5, 1.323740684) (6, 1.36351253) (7, 1.479867646) (8, 1.556875277) (9, 1.578023944) (10, 1.68066007) (11, 1.882895476) (12, 2.121086636) (13, 2.369401995) (14, 2.650004957) (15, 3.009939943) (16, 3.330833599) (17, 3.736448557) (18, 4.200721194) (19, 4.545212224) (20, 5.113381647) (21, 5.653496631) (22, 6.207050042) (23, 6.836661955) (24, 7.417870722) (25, 8.167364103) (26, 8.51603331) (27, 9.476146536) (28, 10.15305622) (29, 11.03931887) (30, 11.94881966) (31, 12.84944513) (32, 13.3704769) (33, 14.68533011) (34, 15.38970141) (35, 16.24415257) (36, 16.8590372) (37, 17.38504957) (38, 19.12658901) (39, 20.23722953) (40, 20.41472256) (41, 21.54809782) (42, 22.19974332) (43, 25.92378207) (44, 23.44496755) (45, 26.76068775) (46, 26.21649914) (47, 26.97246857) (48, 35.93956281) (49, 31.13677335) (50, 31.69647671) (51, 26.75689406) (52, 40.82819799) (53, 24.74606973) (54, 38.30670286) (55, 28.07839508) (56, 29.87242114) (57, 33.19727138) (58, 29.42125203) (59, 29.91650249) (60, 33.8348624) (61, 24.37774657) (62, 44.68629415) (63, 42.98214063)
            (65, 26.73362518) (66, 54.8812091) (67, 50.35539076) (68, 37.44272298)
            (70, 46.38467605) (71, 45.14134299) (72, 29.48596058) (73, 56.76632559) (74, 30.78663094) (75, 50.37807742) (76, 34.9300725) (77, 42.25694647) (78, 25.67906442) (79, 43.66805864) (80, 33.56095856) (81, 28.46198605) (82, 28.82184288) (83, 48.55026991) (84, 24.25952216) (85, 22.28838939) (86, 20.95419562) (87, 19.82224672) (88, 10.44928207) (89, 28.86727057) (90, 5.565863852) (91, 3.659893351) (92, 2.056476414) (93, 1.628858106) (94, 2.041008787) (95, 1.153413757) (96, 1.098018857) (97, 1.022209649) (98, 1.019415077) (99, 1.001772014) (100, 1.000273283) }; 
        \end{axis}
    \end{tikzpicture}
    \label{figure_ratio_relative_w32}
    }

    \caption{Comparing the error of heuristic of truncating a Niagara sequence to the optimal error, for $k=16$ and $W=16$ (a),(b) or $W=32$ (c),(d). The cases of $\errLinf$ and $\errLinfPos$ are very similar in both (a) and (c), and the ratio is almost $1$ except for when the number of rules is very small. In the relative error (b),(d), the ratio is not negligibly close to $1$, and it keeps rising until $n$ is sufficiently large to allow exact representation of the partitions. The slow drop in (b),(d) is because of an increasing percentage of partitions where both errors are $0$, and we consider the ratio to be $1$ in this case. In (d) two outlier points are omitted from the graph: (69, 779.6) and (64, 119.1)}
    \label{figure_ratio_niagara_all}
\end{figure}

\subsection{Running Time}
\label{subsection_experiments_running_time}
Theorem~\ref{theorem_runtime_l_infinity} and Theorem~\ref{theorem_runtime_one_sided} state an analytic bound on the running time. To put matters in practical perspective, we implemented the algorithms in C++ and measured their actual running time. The code was compiled by Visual Studio 2019 with Ox optimization, and ran on a Windows-10 computer with 64bit core i5-9300H processor. We emphasize that we did not try to optimize the code to improve the performance, so the resulting data should be regarded as a good estimation. 

Fig.~\ref{figure_running_time_micros_w16} shows the running time of each of the three algorithms for $\errLinf$, $\errLinfPos$ and $\errLinfPosRel$, in micro-seconds, for $W=16$ and $k=16$. The times were averaged over $10,000$ random ordered-partitions, the same partitions for each of the algorithms. The running time grows approximately linearly in $n$, until $n = 40$ where it rapidly drops.
Fig.~\ref{figure_running_time_relative_w16} shows the ratio between the average running time of each of the algorithms, versus the average time of the Niagara truncation heuristic.
We get similar results for $W=32$ (with $k=16$), as shown in Fig.~\ref{figure_running_time_micros_w32} and Fig.~\ref{figure_running_time_relative_w32}: for a fixed $n$ the running time is about twice slower, and the ratio compared to Niagara is approximately doubled as well.

To understand the graphs, we note that the rapid drop at about $n=40$ for $W=16$ and $n=80$ for $W=32$ happens because these values of $n$ are large enough to allow an exact representation of most of the partitions, and no lifting or binary-search happens (recall the reduction in Theorem~\ref{theorem_equivalent_rule_bound_error_bound1}). When $n$ is too small for an exact representation, then we execute a fixed number of binary-search iterations (for example, $W$ iterations for $\errLinf$). Each such iteration consists of solving a lifting problem which is independent of $n$, and then testing the solution by running $n$ steps of Niagara.
These lifting problems cause the running time to be far from $0$ even for a very small value of $n$. The subsequent tests make it increase linearly in $n$ for $n<40$ when $W=16$ and for $n<80$ when $W=32$. For the same reasons, if we divide the running time of our algorithm by the running time of Niagara for the same $n$ (Fig.~\ref{figure_running_time_relative_w16} and Fig.~\ref{figure_running_time_relative_w32}), the ratio is large for very small $n$, but stabilizes quickly. As expected, the stable ratio is around the number of binary-search iterations, in fact it is slightly larger because in each iteration we solve a lifting problem in addition to running Niagara.

It is important to note that our algorithms are  meant to run once in a period of time when either the traffic pattern or the load-balancing topology change (e.g. link failure or new server added). Such changes do not happen in the order of milliseconds, so the fact that the computation time of our algorithms takes in the order of a millisecond means that we should not be concerned of delay which they may cause. The only concern is due to remapping of addresses, which may require graceful transitioning of connections. This issue is orthogonal to the running-time question, and out of the scope of this paper.

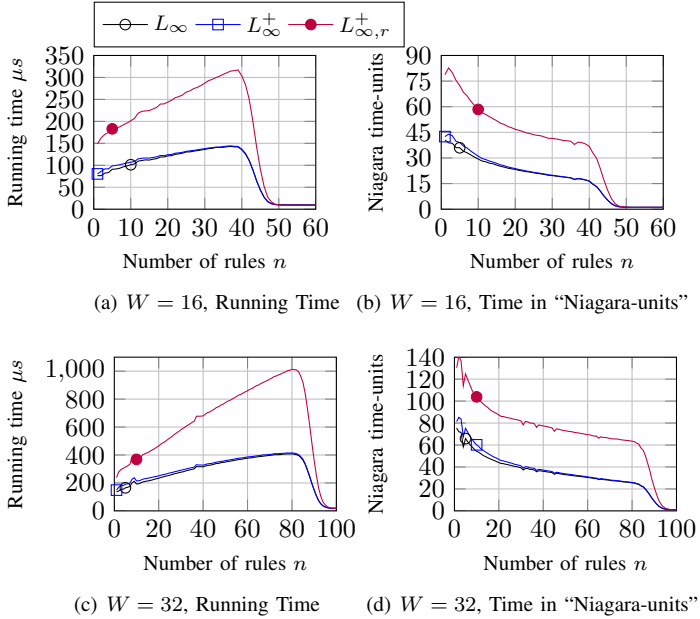
\begin{figure}[t]
    \subfigure[$W=16$, Running Time] {
        \hspace*{-1.5em}  
        \begin{tikzpicture}
            \begin{axis}[
                width=0.25\textwidth,
        		height=0.2\textwidth,
        		xmin=0, xmax=60,
        		xtick={0,10,20,30,40,50,60},
        		xlabel=Number of rules $n$,
        		ymin=0, ymax = 350,
                ytick={0,50,100,150,200,250,300,350},
                ylabel=Running time ${\mu}s$,
        		legend style={at={(0.00,1.35)}, anchor=north west,font=\footnotesize,legend columns=-1,},
        		label style={font=\footnotesize},
        		grid=both
            ]
            \addplot [black,mark=o] coordinates{ (10, 101.2432) }; \addlegendentry{$\errLinf$}
            \addplot [blue,mark=square] coordinates{ (1, 80.33169) }; \addlegendentry{$\errLinfPos$} 
            \addplot [purple,mark=*] coordinates{ (5, 182.99078) }; \addlegendentry{$\errLinfPosRel$} 
            
            \addplot [black] coordinates{ (1, 76.47261) (2, 76.29653) (3, 82.19798) (4, 83.37688) (5, 90.87292) (6, 91.53037) (7, 93.24698) (8, 96.51762) (9, 98.54793) (10, 101.2432) (11, 103.06503) (12, 109.45569) (13, 110.90822) (14, 112.29704) (15, 111.85443) (16, 113.81238) (17, 115.12773) (18, 118.81202) (19, 118.90799) (20, 121.38164) (21, 121.88423) (22, 124.74284) (23, 126.05962) (24, 127.61507) (25, 130.27323) (26, 130.54176) (27, 132.09013) (28, 133.39199) (29, 134.72207) (30, 136.15734) (31, 137.98953) (32, 138.53399) (33, 140.57383) (34, 140.64258) (35, 140.99751) (36, 142.09965) (37, 143.01682) (38, 141.91114) (39, 141.666) (40, 135.45549) (41, 125.4202) (42, 110.41623) (43, 88.96924) (44, 64.6365) (45, 43.9541) (46, 27.80583) (47, 17.3867) (48, 12.86432) (49, 10.77993) (50, 10.17731) (51, 9.97713) (52, 10.10959) (53, 9.85213) (54, 9.8682) (55, 10.03652) (56, 9.84802) (57, 9.82582) (58, 9.91262) (59, 9.73791) (60, 9.87208) (61, 9.89761) (62, 9.84588) (63, 9.75879) (64, 9.78252) (65, 9.86931) (66, 9.7726) (67, 9.87376) (68, 9.83454) (69, 9.83996) (70, 10.04253) (71, 9.72376) (72, 9.83232) (73, 10.05945) (74, 9.75966) (75, 9.84539) (76, 9.88262) (77, 10.15176) (78, 9.80105) (79, 9.79192) (80, 9.74715) (81, 9.88924) (82, 10.09338) (83, 9.83252) (84, 9.88088) (85, 9.76539) (86, 9.73948) (87, 9.80944) (88, 9.8753) (89, 9.85203) (90, 10.0982) (91, 9.81429) (92, 9.80795) (93, 9.89037) (94, 9.82327) (95, 9.86911) (96, 10.19399) (97, 9.85723) (98, 9.83723) (99, 9.80518) (100, 10.09856) };
            
            \addplot [blue] coordinates{ (1, 80.33169) (2, 86.18684) (3, 91.74168) (4, 91.93455) (5, 98.60003) (6, 99.75227) (7, 101.1638) (8, 103.34259) (9, 105.30903) (10, 106.97091) (11, 108.67537) (12, 114.6284) (13, 116.11259) (14, 116.48741) (15, 116.34678) (16, 117.80613) (17, 119.11079) (18, 122.33181) (19, 121.88421) (20, 123.98048) (21, 124.67411) (22, 127.4545) (23, 128.2983) (24, 129.61233) (25, 131.85396) (26, 132.87081) (27, 133.9316) (28, 135.28131) (29, 136.757) (30, 137.74777) (31, 139.32981) (32, 140.1199) (33, 141.39385) (34, 141.46333) (35, 142.59748) (36, 143.61403) (37, 144.1905) (38, 143.12614) (39, 142.73116) (40, 136.50446) (41, 127.07156) (42, 111.61012) (43, 89.82241) (44, 64.86775) (45, 44.34084) (46, 27.54614) (47, 16.97474) (48, 12.43558) (49, 10.37145) (50, 9.71895) (51, 9.54373) (52, 9.64554) (53, 9.36988) (54, 9.40049) (55, 9.55126) (56, 9.35716) (57, 9.3609) (58, 9.444) (59, 9.26629) (60, 9.34745) (61, 9.40387) (62, 9.36213) (63, 9.26825) (64, 9.34311) (65, 9.39266) (66, 9.2996) (67, 9.69713) (68, 9.33992) (69, 9.36284) (70, 9.56166) (71, 9.28774) (72, 9.3285) (73, 9.60434) (74, 9.25864) (75, 9.31008) (76, 9.42321) (77, 9.29923) (78, 9.3288) (79, 9.32637) (80, 9.26843) (81, 9.38669) (82, 9.61333) (83, 9.34444) (84, 9.4079) (85, 9.32933) (86, 9.65958) (87, 9.30317) (88, 9.43503) (89, 9.4368) (90, 9.31633) (91, 9.33709) (92, 9.33487) (93, 9.46744) (94, 9.33889) (95, 9.39934) (96, 9.39257) (97, 9.35473) (98, 9.3113) (99, 9.3643) (100, 9.57866) }; 
            
            \addplot [purple] coordinates{ (1, 148.85042) (2, 162.37942) (3, 170.38004) (4, 174.46248) (5, 182.99078) (6, 184.48581) (7, 187.43386) (8, 192.13722) (9, 196.16432) (10, 201.38137) (11, 206.57956) (12, 218.64197) (13, 222.78024) (14, 223.82574) (15, 224.55309) (16, 228.16538) (17, 231.89434) (18, 239.3537) (19, 239.82676) (20, 243.76086) (21, 247.75107) (22, 253.75717) (23, 256.49019) (24, 260.1583) (25, 266.48172) (26, 269.61051) (27, 274.51824) (28, 277.41204) (29, 281.81643) (30, 287.00059) (31, 292.6331) (32, 295.6521) (33, 299.46579) (34, 303.00377) (35, 307.09835) (36, 310.745) (37, 314.50736) (38, 315.47501) (39, 316.98145) (40, 303.65622) (41, 283.71724) (42, 250.12824) (43, 200.00097) (44, 141.85184) (45, 92.08632) (46, 52.83467) (47, 27.56744) (48, 16.50733) (49, 11.33285) (50, 10.13134) (51, 9.35342) (52, 9.3051) (53, 8.9964) (54, 9.05847) (55, 9.19489) (56, 9.30532) (57, 9.00077) (58, 9.0904) (59, 8.94894) (60, 8.98874) (61, 9.05121) (62, 9.05575) (63, 8.88787) (64, 8.932) (65, 9.03284) (66, 8.96375) (67, 9.02396) (68, 8.97127) (69, 9.00532) (70, 9.2026) (71, 8.88389) (72, 8.99647) (73, 9.27393) (74, 8.91391) (75, 8.93572) (76, 9.07681) (77, 8.91519) (78, 8.95951) (79, 8.99897) (80, 8.87123) (81, 9.01856) (82, 9.24062) (83, 8.96668) (84, 9.44579) (85, 8.91932) (86, 8.94957) (87, 8.93952) (88, 8.98448) (89, 9.01503) (90, 8.96317) (91, 8.976) (92, 8.95763) (93, 9.05418) (94, 8.95606) (95, 9.02755) (96, 9.03569) (97, 8.9828) (98, 8.94892) (99, 8.98945) (100, 9.21508) }; 
        \end{axis}
    \end{tikzpicture} 
    \label{figure_running_time_micros_w16}
    }
    \hspace*{-3.1em}  
    ~
    \subfigure[$W=16$, Time in ``Niagara-units''] {
        \begin{tikzpicture}
            \begin{axis}[
                width=0.25\textwidth,
        		height=0.2\textwidth,
        		xmin=0, xmax=60,
        		xtick={0,10,20,30,40,50,60},
        		xlabel=Number of rules $n$,
        		ymin=0, ymax = 90,
                ytick={0,15,30,45,60,75,90},
                ylabel=Niagara time-units,
        		legend style={at={(0.01,0.99)}, anchor=north west,font=\footnotesize,},
        		label style={font=\footnotesize},
        		grid=both
            ]
            \addplot [black,mark=o] coordinates{ (5, 35.950975) }; 
            \addplot [blue,mark=square] coordinates{ (1, 42.50106608) }; 
            \addplot [purple,mark=*] coordinates{ (10, 58.40323713) }; 
            
            \addplot [black] coordinates{ (1, 40.45934364) (2, 38.85958979) (3, 38.5726728) (4, 36.04956677) (5, 35.950975) (6, 33.9489229) (7, 32.99516645) (8, 31.76029142) (9, 30.47651048) (10, 29.36185516) (11, 28.2309616) (12, 27.77027099) (13, 26.99291519) (14, 26.59088735) (15, 25.79430634) (16, 25.1196541) (17, 24.63283873) (18, 24.04264507) (19, 23.65532587) (20, 23.29401921) (21, 22.73535348) (22, 22.37656129) (23, 21.96704063) (24, 21.71228998) (25, 21.31831363) (26, 20.92102556) (27, 20.66355621) (28, 20.2945457) (29, 19.83820888) (30, 19.67844708) (31, 19.39515281) (32, 19.08108089) (33, 18.87930974) (34, 18.58935256) (35, 18.25749347) (36, 17.33347442) (37, 17.85788295) (38, 17.52935102) (39, 17.06003273) (40, 16.42418264) (41, 14.52020983) (42, 13.13179819) (43, 10.53729657) (44, 7.654072604) (45, 5.22187017) (46, 3.256043495) (47, 2.041200393) (48, 1.465471596) (49, 1.267335451) (50, 1.194533973) (51, 1.159222888) (52, 1.14530175) (53, 1.149794599) (54, 1.157781488) (55, 1.157210341) (56, 1.159298442) (57, 1.155399764) (58, 1.152677144) (59, 1.157486408) (60, 1.162808338) (61, 1.15915125) (62, 1.159685423) (63, 1.161538086) (64, 1.157793327) (65, 1.155991253) (66, 1.15423955) (67, 1.155734097) (68, 1.151514016) (69, 1.117266954) (70, 1.149484636) (71, 1.157352117) (72, 1.160466579) (73, 1.11771915) (74, 1.161307135) (75, 1.167983686) (76, 1.149314258) (77, 1.204267707) (78, 1.155010559) (79, 1.160210337) (80, 1.164443621) (81, 1.160183297) (82, 1.154126651) (83, 1.15798635) (84, 1.150252963) (85, 1.161911527) (86, 1.156545757) (87, 1.158519561) (88, 1.160943633) (89, 1.156490861) (90, 1.197495011) (91, 1.154866913) (92, 1.151266482) (93, 1.15518359) (94, 1.158045323) (95, 1.158770583) (96, 1.196341496) (97, 1.151721817) (98, 1.161931996) (99, 1.156111716) (100, 1.150296899) };
            
            \addplot [blue] coordinates{ (1, 42.50106608) (2, 43.89695374) (3, 43.05120155) (4, 39.74963681) (5, 39.00795984) (6, 36.99845333) (7, 35.79650965) (8, 34.00613043) (9, 32.56741929) (10, 31.02296614) (11, 29.76771265) (12, 29.0826519) (13, 28.25955817) (14, 27.58312772) (15, 26.83026935) (16, 26.00111901) (17, 25.48505804) (18, 24.75490517) (19, 24.24740933) (20, 23.7927555) (21, 23.25575639) (22, 22.86298301) (23, 22.35715107) (24, 22.05210164) (25, 21.57698916) (26, 21.2942863) (27, 20.95162708) (28, 20.58199093) (29, 20.13785812) (30, 19.90830757) (31, 19.5835362) (32, 19.29951737) (33, 18.98943985) (34, 18.69783472) (35, 18.46467048) (36, 17.51820019) (37, 18.00443523) (38, 17.67943199) (39, 17.18830391) (40, 16.55137184) (41, 14.7113919) (42, 13.27378748) (43, 10.63834391) (44, 7.681456579) (45, 5.267815965) (46, 3.225633976) (47, 1.992836246) (48, 1.416630593) (49, 1.219312766) (50, 1.14073522) (51, 1.108867004) (52, 1.092730154) (53, 1.093513526) (54, 1.102907653) (55, 1.101259883) (56, 1.101514925) (57, 1.100730693) (58, 1.098184228) (59, 1.101427794) (60, 1.101013444) (61, 1.101327256) (62, 1.102707497) (63, 1.103151658) (64, 1.105787712) (65, 1.100161288) (66, 1.098373628) (67, 1.135059368) (68, 1.093599578) (69, 1.063092912) (70, 1.094443459) (71, 1.105455662) (72, 1.101002864) (73, 1.06715126) (74, 1.101690498) (75, 1.104478497) (76, 1.095886476) (77, 1.103135061) (78, 1.099357977) (79, 1.105048947) (80, 1.10725332) (81, 1.101225266) (82, 1.099235376) (83, 1.100504649) (84, 1.095192417) (85, 1.110027973) (86, 1.147057776) (87, 1.0987278) (88, 1.109185342) (89, 1.107748653) (90, 1.104776959) (91, 1.098713845) (92, 1.095735902) (93, 1.105785863) (94, 1.10094275) (95, 1.10361306) (96, 1.102288824) (97, 1.093009561) (98, 1.099811369) (99, 1.104128322) (100, 1.091076638) }; 
            
            \addplot [purple] coordinates{ (1, 78.75225251) (2, 82.70359939) (3, 79.95346764) (4, 75.43214403) (5, 72.39447084) (6, 68.42640906) (7, 66.32291372) (8, 63.22507848) (9, 60.6649369) (10, 58.40323713) (11, 56.58504758) (12, 55.47218931) (13, 54.22040065) (14, 52.99983898) (15, 51.78329721) (16, 50.35862903) (17, 49.61633378) (18, 48.43530188) (19, 47.71067243) (20, 46.77948127) (21, 46.21359261) (22, 45.51934899) (23, 44.69575923) (24, 44.26305178) (25, 43.60789152) (26, 43.20861286) (27, 42.9443372) (28, 42.20606743) (29, 41.49827272) (30, 41.47940848) (31, 41.13111837) (32, 40.72185921) (33, 40.21877619) (34, 40.04934997) (35, 39.76556835) (36, 37.90502306) (37, 39.27115442) (38, 38.96855588) (39, 38.17227784) (40, 36.81877507) (41, 32.84665354) (42, 29.74774241) (43, 23.68761984) (44, 16.79769608) (45, 10.9401127) (46, 6.186903379) (47, 3.23642033) (48, 1.880474307) (49, 1.332339131) (50, 1.18913837) (51, 1.086755264) (52, 1.054162168) (53, 1.049926476) (54, 1.062780332) (55, 1.060170436) (56, 1.095412375) (57, 1.058383681) (58, 1.057066276) (59, 1.06370632) (60, 1.058761864) (61, 1.060025742) (62, 1.066620888) (63, 1.057877002) (64, 1.057131495) (65, 1.058015609) (66, 1.058706462) (67, 1.056264105) (68, 1.050434809) (69, 1.022498714) (70, 1.053344856) (71, 1.057388181) (72, 1.061814787) (73, 1.030438957) (74, 1.060670892) (75, 1.060067217) (76, 1.055601363) (77, 1.057577742) (78, 1.055838778) (79, 1.066256466) (80, 1.059801808) (81, 1.058037086) (82, 1.056617883) (83, 1.056015452) (84, 1.099603268) (85, 1.061243916) (86, 1.062745363) (87, 1.055779819) (88, 1.056218531) (89, 1.05823874) (90, 1.062897482) (91, 1.056223671) (92, 1.051455113) (93, 1.057517581) (94, 1.0558117) (95, 1.059959751) (96, 1.060406268) (97, 1.049553144) (98, 1.057008576) (99, 1.059930411) (100, 1.049662323) }; 
        \end{axis}
    \end{tikzpicture}
    \label{figure_running_time_relative_w16}
    }
    
    \subfigure[$W=32$, Running Time] {
        \hspace*{-1.5em}  
        \begin{tikzpicture}
            \begin{axis}[
                width=0.25\textwidth,
        		height=0.2\textwidth,
        		xmin=0, xmax=100,
        		xtick={0,20,40,60,80,100},
        		xlabel=Number of rules $n$,
        		ymin=0, ymax = 1100,
                ytick={0,200,400,600,800,1000},
                ylabel=Running time ${\mu}s$,
        		legend style={at={(0.00,1.35)}, anchor=north west,font=\footnotesize,legend columns=-1,},
        		label style={font=\footnotesize},
        		grid=both
            ]
            \addplot [black,mark=o] coordinates{ (5, 165.05593) }; 
            \addplot [blue,mark=square] coordinates{ (1, 148.93495) }; 
            \addplot [purple,mark=*] coordinates{ (10, 368.45916) }; 
            
            \addplot [black] coordinates{ (1, 136.81479) (2, 142.45338) (3, 154.31372) (4, 154.93881) (5, 165.05593) (6, 168.7069) (7, 174.50182) (8, 200.93901) (9, 211.65727) (10, 192.45977) (11, 194.26729) (12, 200.93081) (13, 208.01412) (14, 211.30305) (15, 214.04735) (16, 217.41275) (17, 221.08474) (18, 226.01619) (19, 230.11478) (20, 234.12943) (21, 239.79596) (22, 243.32662) (23, 248.1673) (24, 251.44946) (25, 255.44919) (26, 259.99979) (27, 265.42243) (28, 268.84813) (29, 273.02223) (30, 276.91753) (31, 279.82488) (32, 285.35715) (33, 289.74344) (34, 293.24313) (35, 297.25639) (36, 300.79816) (37, 317.69447) (38, 315.81252) (39, 316.05159) (40, 315.14702) (41, 320.99671) (42, 325.0727) (43, 327.03814) (44, 330.07365) (45, 334.34384) (46, 339.7822) (47, 339.86881) (48, 345.05709) (49, 346.82412) (50, 349.42882) (51, 353.23029) (52, 355.50431) (53, 358.26361) (54, 361.22269) (55, 364.16646) (56, 367.72208) (57, 369.13325) (58, 371.72685) (59, 374.63092) (60, 376.90614) (61, 379.51665) (62, 380.08211) (63, 383.46204) (64, 386.16597) (65, 387.49307) (66, 389.87912) (67, 391.28158) (68, 393.75626) (69, 395.23982) (70, 397.32877) (71, 398.9587) (72, 401.84153) (73, 401.86884) (74, 402.66328) (75, 403.126) (76, 406.37754) (77, 406.57785) (78, 407.98049) (79, 407.55482) (80, 407.95267) (81, 405.51573) (82, 402.46515) (83, 396.29454) (84, 381.85157) (85, 361.61999) (86, 330.92225) (87, 288.37388) (88, 242.81548) (89, 193.06991) (90, 148.04145) (91, 107.37695) (92, 73.66954) (93, 50.43711) (94, 35.71356) (95, 28.01387) (96, 23.05942) (97, 20.38499) (98, 19.27488) (99, 19.21971) (100, 18.62) };
            
            \addplot [blue] coordinates{ (1, 148.93495) (2, 169.80336) (3, 181.56012) (4, 181.91435) (5, 190.46236) (6, 194.6619) (7, 198.31484) (8, 227.32911) (9, 237.10528) (10, 215.06152) (11, 215.46645) (12, 220.64361) (13, 227.54833) (14, 231.21373) (15, 233.44506) (16, 236.87704) (17, 240.26219) (18, 244.27163) (19, 247.1847) (20, 250.81165) (21, 255.35377) (22, 259.57767) (23, 262.81252) (24, 266.97722) (25, 271.74925) (26, 274.56981) (27, 279.54925) (28, 282.81888) (29, 287.5818) (30, 291.29965) (31, 294.47813) (32, 298.74771) (33, 302.8215) (34, 306.43225) (35, 309.76289) (36, 313.36541) (37, 330.93685) (38, 328.04597) (39, 328.44723) (40, 327.19199) (41, 332.16545) (42, 336.47238) (43, 338.70961) (44, 340.09411) (45, 345.08397) (46, 348.94389) (47, 351.06569) (48, 354.79549) (49, 355.88065) (50, 359.98664) (51, 362.31826) (52, 365.10796) (53, 367.72687) (54, 370.80589) (55, 374.14355) (56, 377.01462) (57, 378.21581) (58, 380.48718) (59, 383.54651) (60, 385.84728) (61, 388.20076) (62, 388.48687) (63, 391.45596) (64, 393.86291) (65, 395.30015) (66, 397.6372) (67, 401.14865) (68, 400.81211) (69, 402.4633) (70, 404.00647) (71, 405.77504) (72, 409.50074) (73, 409.16173) (74, 409.88501) (75, 410.85912) (76, 412.40412) (77, 412.86939) (78, 414.84654) (79, 413.70134) (80, 415.27563) (81, 413.22165) (82, 409.52176) (83, 402.42761) (84, 387.67484) (85, 368.94793) (86, 335.84588) (87, 293.12193) (88, 247.30269) (89, 196.0505) (90, 150.36897) (91, 108.5088) (92, 73.85288) (93, 50.43389) (94, 35.62099) (95, 27.55902) (96, 22.54448) (97, 19.72429) (98, 18.64588) (99, 18.56001) (100, 17.96761) }; 
            
            \addplot [purple] coordinates{ (1, 238.32324) (2, 279.64704) (3, 296.0306) (4, 304.68877) (5, 313.72076) (6, 322.35931) (7, 333.66126) (8, 381.8363) (9, 403.89559) (10, 368.45916) (11, 372.93961) (12, 386.73408) (13, 403.11898) (14, 410.73362) (15, 416.89315) (16, 425.79915) (17, 436.00337) (18, 443.72967) (19, 454.7701) (20, 464.88017) (21, 478.39363) (22, 487.49399) (23, 497.28662) (24, 510.54332) (25, 521.00108) (26, 533.1096) (27, 546.07685) (28, 557.62633) (29, 568.40471) (30, 581.11112) (31, 591.08121) (32, 602.03235) (33, 612.31369) (34, 621.7447) (35, 633.44666) (36, 641.4068) (37, 677.43084) (38, 677.37554) (39, 678.08032) (40, 679.56819) (41, 692.08498) (42, 703.1355) (43, 708.2178) (44, 718.48428) (45, 727.89195) (46, 739.90074) (47, 745.24675) (48, 757.4221) (49, 763.67254) (50, 774.37312) (51, 784.67514) (52, 789.97573) (53, 799.55662) (54, 808.9567) (55, 817.96368) (56, 828.00523) (57, 834.20904) (58, 843.20856) (59, 852.05188) (60, 859.91964) (61, 869.24207) (62, 875.4893) (63, 886.38969) (64, 895.5684) (65, 900.68202) (66, 911.47093) (67, 918.39267) (68, 926.54927) (69, 933.39586) (70, 943.21033) (71, 950.70559) (72, 961.92467) (73, 965.6894) (74, 973.56838) (75, 979.754) (76, 987.49782) (77, 994.09217) (78, 1003.50576) (79, 1006.4694) (80, 1013.50131) (81, 1011.569) (82, 1008.74157) (83, 995.68677) (84, 962.42502) (85, 916.58938) (86, 835.89729) (87, 728.51203) (88, 610.10326) (89, 481.75821) (90, 361.42497) (91, 254.4793) (92, 164.06069) (93, 103.39515) (94, 64.01946) (95, 42.7172) (96, 29.44001) (97, 22.36452) (98, 19.34609) (99, 18.87257) (100, 17.52711) }; 
        \end{axis}
    \end{tikzpicture}
    \label{figure_running_time_micros_w32}
    }
    \hspace*{-1.6em}  
    ~
    \subfigure[$W=32$, Time in ``Niagara-units''] {
        \begin{tikzpicture}
            \begin{axis}[
                width=0.25\textwidth,
        		height=0.2\textwidth,
        		xmin=0, xmax=100,
        		xtick={0,20,40,60,80,100},
        		xlabel=Number of rules $n$,
        		ymin=0, ymax = 140,
                ytick={0,20,40,60,80,100,120,140},
                ylabel=Niagara time-units,
        		legend style={at={(0.01,0.99)}, anchor=north west,font=\footnotesize,},
        		label style={font=\footnotesize},
        		grid=both
            ]
            \addplot [black,mark=o] coordinates{ (5, 65.83311113) }; 
            \addplot [blue,mark=square] coordinates{ (10, 60.16043968) }; 
            \addplot [purple,mark=*] coordinates{ (10, 103.8505581) }; 
            
            \addplot [black] coordinates{ (1, 75.49587478) (2, 72.06044119) (3, 71.54419479) (4, 58.3874529) (5, 65.83311113) (6, 62.96133234) (7, 60.38027379) (8, 58.20964286) (9, 56.08524842) (10, 54.31430109) (11, 53.33175894) (12, 51.88427326) (13, 50.43119221) (14, 49.52154497) (15, 48.42012162) (16, 47.49578984) (17, 46.59105827) (18, 45.80329543) (19, 44.89235295) (20, 44.09515193) (21, 43.12964478) (22, 43.03922195) (23, 42.50389187) (24, 42.15022245) (25, 41.61709054) (26, 41.20241456) (27, 40.53541251) (28, 40.30940868) (29, 39.72415566) (30, 39.20090737) (31, 37.29508496) (32, 38.63567912) (33, 38.57179329) (34, 37.76775057) (35, 37.6625083) (36, 37.26902298) (37, 35.62022339) (38, 36.63306722) (39, 36.38473854) (40, 36.07189495) (41, 35.82492748) (42, 35.41811252) (43, 35.21393032) (44, 34.82747583) (45, 33.34924468) (46, 34.2310947) (47, 33.99405038) (48, 33.53647851) (49, 33.26002596) (50, 32.94718653) (51, 32.7570289) (52, 32.42967351) (53, 32.22668532) (54, 31.90584343) (55, 31.67790815) (56, 31.46624251) (57, 31.2153726) (58, 30.98284257) (59, 30.73752509) (60, 30.39150706) (61, 30.18072064) (62, 29.85144474) (63, 29.41827179) (64, 29.41503088) (65, 28.49462894) (66, 28.88849176) (67, 28.50122402) (68, 28.35831594) (69, 28.1102516) (70, 27.80009202) (71, 27.49167308) (72, 27.32918335) (73, 27.22612509) (74, 27.03337219) (75, 26.82044579) (76, 26.39678421) (77, 26.27760175) (78, 26.07112272) (79, 25.89530591) (80, 25.52984109) (81, 25.30081375) (82, 24.6309016) (83, 24.18263739) (84, 22.70550026) (85, 21.79649951) (86, 19.78823385) (87, 17.19326029) (88, 14.44715955) (89, 11.31711558) (90, 8.762813689) (91, 6.400971532) (92, 4.359731545) (93, 3.016825362) (94, 2.141705012) (95, 1.653863517) (96, 1.360930263) (97, 1.206618029) (98, 1.137864079) (99, 1.12101096) (100, 1.100653206) };
            
            \addplot [blue] coordinates{ (1, 81.38303901) (2, 85.16229346) (3, 83.55117413) (4, 67.54365808) (5, 75.2015716) (6, 71.57621565) (7, 67.93132786) (8, 65.13407191) (9, 62.22728141) (10, 60.16043968) (11, 58.53732829) (12, 56.40893595) (13, 55.0814749) (14, 53.63492624) (15, 52.42516316) (16, 51.14908798) (17, 49.95632597) (18, 48.87660272) (19, 48.02716592) (20, 46.99040451) (21, 45.79284972) (22, 45.62047053) (23, 44.88243665) (24, 44.48131033) (25, 43.81170503) (26, 43.35987962) (27, 42.49659872) (28, 42.1124583) (29, 41.62329872) (30, 40.9843184) (31, 38.92053194) (32, 40.27696473) (33, 39.97889492) (34, 39.30448964) (35, 39.00100196) (36, 38.7999052) (37, 36.86113199) (38, 37.85728098) (39, 37.58288874) (40, 37.12339575) (41, 36.88120727) (42, 36.50074423) (43, 36.18933853) (44, 35.83759287) (45, 34.31384317) (46, 35.1554262) (47, 34.86463396) (48, 34.37634902) (49, 34.09319963) (50, 33.77996546) (51, 33.54329169) (52, 33.19718707) (53, 33.05020432) (54, 32.68051953) (55, 32.3421825) (56, 32.1315478) (57, 31.88822902) (58, 31.61013767) (59, 31.27153312) (60, 31.00136033) (61, 30.6840431) (62, 30.4539815) (63, 29.97598687) (64, 29.92055742) (65, 28.91451812) (66, 29.49683634) (67, 28.99611347) (68, 28.91439841) (69, 28.59001838) (70, 28.18302367) (71, 27.96910087) (72, 27.80904615) (73, 27.67646684) (74, 27.38680506) (75, 27.13410969) (76, 26.77905415) (77, 26.61581446) (78, 26.40247315) (79, 26.1955221) (80, 25.91719674) (81, 25.72463656) (82, 25.00965434) (83, 24.55613403) (84, 23.01123599) (85, 22.14980341) (86, 20.08545196) (87, 17.40384669) (88, 14.70726835) (89, 11.42132769) (90, 8.867986242) (91, 6.453685813) (92, 4.381187993) (93, 3.006279772) (94, 2.104994321) (95, 1.622039147) (96, 1.329020651) (97, 1.177892383) (98, 1.10395434) (99, 1.096699515) (100, 1.060767603) }; 
            
            \addplot [purple] coordinates{ (1, 130.3609526) (2, 140.8129393) (3, 136.5294985) (4, 114.029939) (5, 124.7279532) (6, 119.39658) (7, 114.6887509) (8, 109.6922104) (9, 106.1909633) (10, 103.8505581) (11, 101.8230024) (12, 99.30982754) (13, 97.28404808) (14, 95.61630306) (15, 94.21082458) (16, 92.22472177) (17, 90.79346814) (18, 89.73595761) (19, 88.69491339) (20, 87.38398999) (21, 86.05780231) (22, 85.96866656) (23, 85.36280431) (24, 84.86004926) (25, 84.53049848) (26, 83.92700542) (27, 83.4533367) (28, 83.0764159) (29, 82.66539322) (30, 81.86639521) (31, 78.36071654) (32, 81.54210327) (33, 81.06151419) (34, 80.05657004) (35, 79.93768709) (36, 79.33054982) (37, 76.06731165) (38, 78.3406728) (39, 77.75940452) (40, 77.34906599) (41, 77.02699996) (42, 76.59267234) (43, 76.06636136) (44, 75.57423564) (45, 72.54264518) (46, 74.72846269) (47, 74.3563101) (48, 73.74261503) (49, 73.31771692) (50, 72.95326709) (51, 72.61041323) (52, 72.22746849) (53, 72.07476932) (54, 71.48118707) (55, 71.18579064) (56, 70.97645327) (57, 70.50977612) (58, 70.21053923) (59, 69.79623432) (60, 69.46249971) (61, 69.23801105) (62, 68.59573889) (63, 68.07610694) (64, 68.15673882) (65, 66.21367301) (66, 67.5255518) (67, 66.85541869) (68, 66.8544023) (69, 66.48479279) (70, 65.99521021) (71, 65.60524747) (72, 65.50592318) (73, 65.50190878) (74, 65.24475185) (75, 64.76765156) (76, 64.55107122) (77, 64.27660834) (78, 64.13906768) (79, 63.83382137) (80, 63.45212069) (81, 63.05134369) (82, 61.84147968) (83, 60.91473033) (84, 57.27767185) (85, 55.25506357) (86, 50.29335726) (87, 43.48114267) (88, 36.60048923) (89, 28.1544804) (90, 21.41755698) (91, 15.17862688) (92, 9.784921211) (93, 6.14692986) (94, 3.77900055) (95, 2.514342717) (96, 1.762873119) (97, 1.348940816) (98, 1.168752331) (99, 1.094845842) (100, 1.045118239) }; 
        \end{axis}
    \end{tikzpicture}
    \label{figure_running_time_relative_w32}
    }

    \caption{Running time of the algorithms, for fixed $k=16$, $W=16$ (top) or $W=32$ (bottom) and $n \in [1,60]$ or $n \in [1,100]$ respectively. For each $n$, the running time was averaged over $10,000$ random ordered-partitions (same set for all three algorithm). (a)+(c) show the running time in $\mu$s, (b)+(d) normalize the running time in ``units'' of Niagara, to show the loss in running time compared to the heuristic of truncating a Niagara sequence. These units grow from $0.8{\mu}s$ ($n\approx 1$) to $18{\mu}s$ ($n\approx 100$), monotonous in $n$, not affected much by $W$. At about $n \ge 40$ for $W=16$ and $n \ge 80$ for $W=32$ many partitions can be represented exactly for, so the approximation-logic is not required and the running-time drops.}
    \label{figure_running_time_all}
\end{figure}

\subsection{Unreachable Targets (Degeneracy Concerns)}
\label{subsection_experiments_degeneracy}
Our algorithms find the closest partition (in one of $\errLinf$, $\errLinfPos$ or $\errLinfPosRel$) that can be represented by $n$ rules to a given partition. As in Example~\ref{figure_tcam_trie_all}, such an optimal partition may be ``degenerate'' in the sense that it does not map any address to some of the $k$ targets.

Such a degeneracy is likely only when $\frac{n}{k}$ is very small. Fig.~\ref{figure_degeneracy} shows that for $n < 2k$ there is high probability to encounter such a degeneracy, but for $n \ge 3k$ this probability is already low. We note that if the objective is $\errLinfPosRel$ then it is slightly more likely to get a degenerate partition than for $\errLinfPos$, which in turn has higher chances for degeneracy than $\errLinf$, but overall these chances for the different objectives are  similar, and the degeneracy happens mostly when $\frac{n}{k}$ is no more than $2.5$. One can avoid degenerate partitions by incorporating various heuristics that postprocess the rules.

We note that in real-life situations the desired partitions are likely to be more balanced than in a uniform ordered-partition (which is skewed toward having few heavy parts), and therefore the probability of degeneracy is even smaller because loosely speaking each target is more likely to get ``its fair share'' of rules.

\begin{figure}[!t]
	\centering
	
	\subfigure[$\errLinf$]{
    \includegraphics[width=1.0\linewidth]{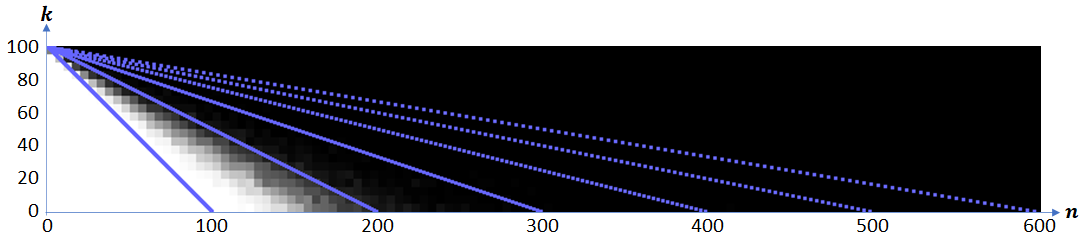}
    \label{figure_degeneracy_2sided}
    }
    ~
    \subfigure[$\errLinfPos$]{
    \includegraphics[width=1.0\linewidth]{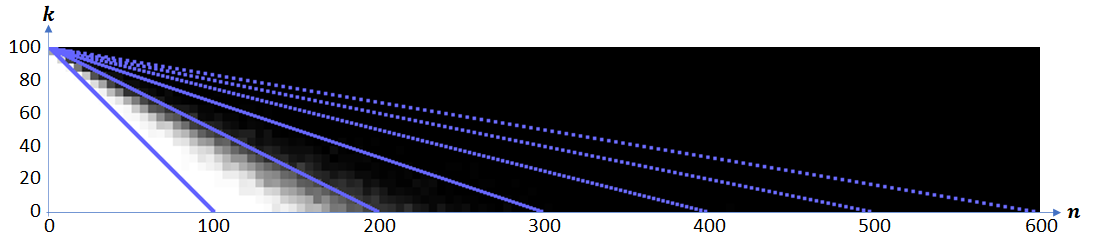}
    \label{figure_degeneracy_1sided}
    }
    ~
    \subfigure[$\errLinfPosRel$]{
    \includegraphics[width=1.0\linewidth]{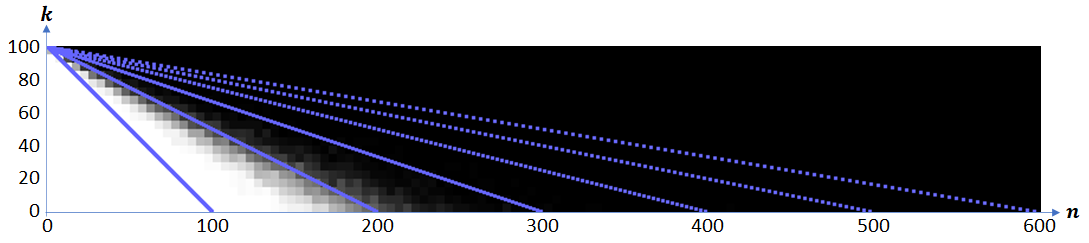}
    \label{figure_degeneracy_1sidedRelative}
    }

    \caption{A visualization of the probability that approximating a partition $P$ with $k$ targets using $n$ rules will result in a partition $P'$ with less than $k$ targets. $W=32$ is fixed, each large pixel represent a pair of $n$ and $k$ for $k=5,10,\ldots,100$ and $n=5,10,\ldots,600$. The dashed blue lines mark a constant integer ratio $\frac{n}{k} = c$ for $c=1,\ldots,6$. White ``pixels'' represent 100\% degeneracy, black represent 0\%, gray are intermediate values. Each pixel was computed by sampling 100 random ordered-partitions, note that when $n<k$ degeneracy is guaranteed for any partition.}
    \label{figure_degeneracy}
\end{figure}

\subsection{Real Data Partitions}
\label{subsection_experiments_real_data_partitions}
In this section we provide results for ``real data partitions''. 
We use the data of \cite{ftp_packet_captures},\footnote{Available for download at: https://ee.lbl.gov/anonymized-traces.html} containing 3.2 million packets in 22 thousand connections between 5832 distinct clients to 320 distinct servers. 

We made the assumption that the data traffic approximates the target partition in the following concrete sense.\footnote{Using any captured data can only be done if we assume that the data represents closely enough a desired partition.} We sliced the data into windows of one hour each, to get a total of 240 time frames. From each frame we extracted three partitions, according to three types of loads on the servers that communicated in the particular frame as follows: (1) the number of unique clients per server (``load balancing sessions''); (2) the number of incoming packets (``load balancing requests''); (3) the number of outgoing bytes (``load balancing data-processing''). Overall, we get 720 partitions with sums that range in $[18,298]$ (connections), $[2143,16306]$ (packets) and $[196637,1618691]$ (sent bytes). The number of parts in the partitions (targets) varies among $k \in \{4{-}18, 20, 21, 23, 28, 58, 64, 67, 260\}$.

Most of the partitions do not sum to a power of $2$, as could be expected. Since our model requires a universe of $2^W$ addresses, we normalized the partitions to be non-integer partitions that sum to a power of $2$ as explained in Section~\ref{section_non_integer_partitions_approximation} (Problem \ref{problem_tcam_approx_by_rules_non_integer}). On one hand we didn't want to normalize partitions with small sums using a large $W$, and on the other hand we wanted to have only a few values of $W$, to be able to compare partitions with similar parameters. Therefore, we normalized the sum of each partition to a multiple of $256$. All but one of the connections-partitions end up with $W=8$ (all sums are $\le 112$ except for one exception of $298$), all packets-partitions end up with $W=16$, and all bytes-partitions end up with $W=24$.

For each normalized partition we check the trade-off between error and percentage of rules out of the maximum necessary $n^*(P)$ for the best possible integer representation.\footnote{$n^*(P) = \lambda(P)$ if $P$ is integer, or $n^*(P) = \lambda(P')$ for an integer partition $P'$ closest to $P$.} The percentages that we sampled were multiples of $10$, that is, $\floor{\frac{i \cdot n^*(P)}{10}}$ for $i=1,\ldots,10$. Fig.~\ref{figure_real_data_all} shows this trade-off. To reduce cluttering, the data in the figure only relies on 72 out of the 240 time-frames, covering the first three days. The $x$-axis is the number of rules used to approximate a partition, in percentage ($100$ means ``best representation''). The $y$-axis for $\errLinf$ and $\errLinfPos$ is $\frac{\lg (\text{error})}{W}$, where we divide by $W$ to normalize the presentation for partitions with different sums. We see that just like the results for randomly ordered-partitions, the error decreases exponentially with the increase in the number of rules (linear in the graph, due to log-scale). The division by $W=8$, $W=16$ or $W=24$ (depending on the case) can be thought of as ``error per width unit''.

Regarding $\errLinfPosRel$, in all of the 720 partitions we derived, $\max_{1 \le i \le k}{p_i} > 0.0838 \sum_{i=1}^{k}{p_i}$, which means that even if we only use one rule, the $\errLinfPosRel$ error would be less than $\frac{\sum_{i=1}^{k}{p_i} - \max_{1 \le i \le k}{p_i}}{\max_{1 \le i \le k}{p_i}} \le \frac{1}{0.0838} - 1 < 11$ (allocating the match-all rule to the maximum weight). In each of the approximations that we computed, the maximum relative error is less than $2.6$, so taking $\lg$ of the relative error yields values that are mostly negative. Since the relative error is computed by dividing the absolute error by a part of the partition, which can be of the order of $2^W$, a non-zero relative error can get as low as $\approx 2^{-W}$. This is why its logarithm, divided by $W$, reaches values of approximately $-1$. (when the error is zero, its logarithm is undefined no matter whether the error is relative or absolute).

\begin{figure}[t]
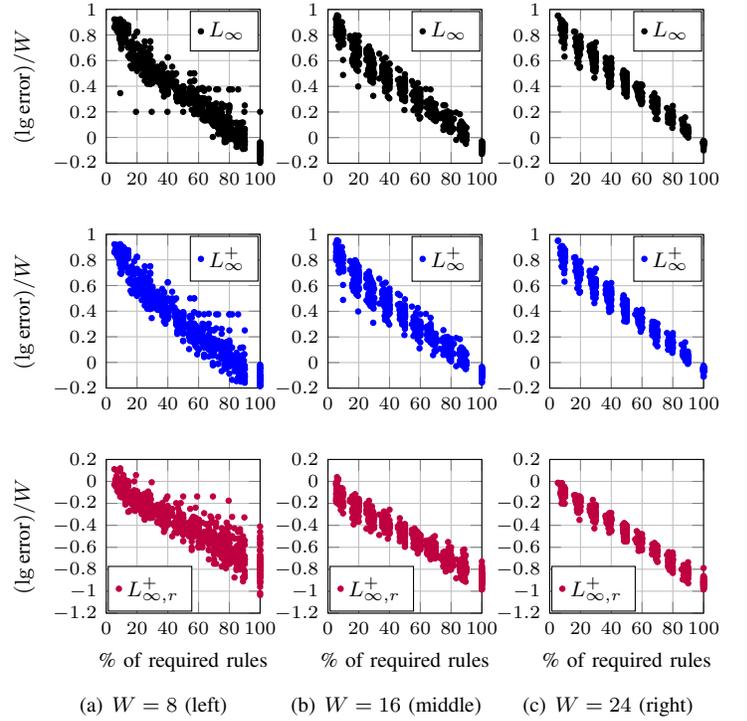

    \subfigure {
        \hspace*{-1.0em}  

    \label{figure_real_data_relative_bytes}
    }
    
    \caption{Scatter plots of error versus number of rules for ``real data'' partitions. The number of rules is normalized by $n^*(P)$, the maximum number of rules needed for the lowest possible error (e.g.\ zero-error if the partition is integer).
    All three kinds of errors, $\errLinf$, $\errLinfPos$ and $\errLinfPosRel$, are presented in log-scale normalized by the width $W$. The columns plot, from left to right: session-partitions ($W=8$), packets-partitions ($W=16$), bytes-partitions ($W=24$) as explained in Section~\ref{subsection_experiments_real_data_partitions}. To reduce cluttering, only the first 72 partition out of 240 are used for the plot. Each partition corresponds to 10 points: the $i$th point gives the error when we approximate this partition with  $\floor{\frac{i \cdot n^*(P)}{10}}$ rules for $i=1,\ldots,10$.
    }
    \label{figure_real_data_all}
\end{figure}

}

\end{document}